\documentclass[a4paper,onecolumn,11pt,accepted=2023-09-05,titlepage]{quantumarticle}
\pdfoutput=1

\usepackage[utf8]{inputenc}
\usepackage[english]{babel}
\usepackage{graphicx}

\usepackage{physics,enumitem,hyperref}

\usepackage{tabularx}
\usepackage{diagbox}
\usepackage{multirow}
\usepackage{nicematrix}
\usepackage{booktabs} 

\usepackage{dsfont}
\usepackage{complexity}
\usepackage{tcolorbox}

\usepackage[numbers]{natbib}
\usepackage{cleveref}

\DeclareMathOperator*{\argmin}{arg\,min}

\usepackage{mathtools}

\usepackage{amsthm}
\usepackage{amssymb} 
\newtheorem{prop}{Proposition}
\theoremstyle{definition}
\newtheorem{definition}{Definition}
\newtheorem{theorem}{Theorem}
\newtheorem{corollary}{Corollary} 
\newtheorem{example}{Example} 
\newtheorem{lemma}[theorem]{Lemma}
\newenvironment{claim}[1]{\par\noindent\underline{Claim:}\space#1}{}
\newenvironment{claimproof}[1]{\par\noindent\underline{Proof:}\space#1}{\hfill $\blacksquare$}
\newtheorem{remark}{Remark}

\begin{document}

\begin{titlepage}
	\title{Discovering optimal fermion—qubit mappings through algorithmic enumeration}
	
	\author{Mitchell Chiew}
	\affiliation{DAMTP, Centre for Mathematical Sciences, University of Cambridge, Cambridge CB30WA, UK}
	\email{mlc79@cam.ac.uk}
	\orcid{0009-0007-3562-6010 }
	\author{Sergii Strelchuk}
	\orcid{0000-0001-8390-3034}
	\affiliation{DAMTP, Centre for Mathematical Sciences, University of Cambridge, Cambridge CB30WA, UK}
	\maketitle
	
\end{titlepage}

	\begin{abstract}
	Simulating fermionic systems on a quantum computer requires a high--performing mapping of fermionic states to qubits. A key characteristic of an efficient mapping is its ability to translate local fermionic interactions into local qubit interactions, leading to easy--to--simulate qubit Hamiltonians. 
	
	\textit{All} fermion--qubit mappings must use a numbering scheme for the fermionic modes in order for translation to qubit operations. We make a distinction between the unordered, symbolic labelling of fermions and the ordered, numeric labelling of the qubits to which the fermionic system maps. This separation shines light on a new way to design fermion--qubit mappings by making use of the extra degree of freedom -- the enumeration scheme for the fermionic modes. The purpose of this paper is to demonstrate that this concept allows for notions of fermion--qubit mappings that are \textit{optimal}, with regard to any cost function one might choose. Our main example is the minimisation of the  average number of Pauli matrices in the Jordan--Wigner transformations of Hamiltonians for fermions interacting in square lattice--type arrangements.  In choosing the best ordering of fermionic modes for the Jordan--Wigner transformation, and unlike other popular modifications, our prescription does not cost additional resources such as ancilla qubits.
	
	We demonstrate how Mitchison and Durbin's enumeration pattern is optimal for minimising the average Pauli weight of Jordan--Wigner transformations of systems interacting in square fermionic lattices. This leads to qubit Hamiltonians consisting of terms with average Pauli weights 13.9\% shorter than previously known. Furthermore, by  adding only two ancilla qubits we introduce a new class of fermion--qubit mappings, and reduce the average Pauli weight of Hamiltonian terms by 37.9\% compared to previous methods. For $n$--mode fermionic systems of cellular arrangements of square lattices, we find enumeration patterns which result in $n^{1/4}$ improvement in average Pauli weight over na\"ive enumeration schemes.
\end{abstract}

\section{Introduction}
Simulating physical systems is one of the most promising applications of quantum computers. Fermionic systems are essential components in several fields of theoretical and experimental physics, from quantum physics \cite{wecker2015solving,d2014feynman,friis2013fermionic,barkoutsos2017fermionic} to quantum chemistry and condensed matter \cite{lanyon2010towards,salmhofer2019renormalization,kraus2009quantum} to quantum field theories \cite{moosavian2018faster}. Fermions pose complex, often intractable computational challenges when studied with the aid of classical computers, such as the electronic structure problem \cite{goings2022reliably}, studying properties of gauge theories that govern strong interactions between quarks and gluons \cite{hagiwara2002review}, determining ground state properties of fermionic Hamiltonians \cite{hastings2021optimizing} and many others.

One can break down all quantum algorithms for fermionic simulation into three sequential steps: 1) initialising the quantum register, 2) applying unitary gates to the qubits, and 3) measuring the result to obtain an estimate for the desired molecular property or other quantity of interest. Within this framework, algorithms may encode the fermionic Hamiltonian via first or second quantisation. Fermi--Dirac statistics impose asymmetry on fermionic systems' wavefunctions, and using first quantisation of the system's Hamiltonian, one can incorporate this asymmetry into the qubit basis itself or use the qubits to directly encode the wavefunction into real--time and real--space \cite{aspuru2005simulated,wang2008quantum,kassal2009quantum,kassal2008polynomial}. In contrast, second quantisation encodes the asymmetry into the qubit operators rather than the quantum states \cite{mcclean2014exploiting}, and provides a number of distinct advantages over representations in the first quantisation \cite{verstraete2005mapping,kassal2011simulating}.

Quantum algorithms employing second quantisation require an important fourth, pre-emptive step: the fermion--qubit mapping, which we label 0). In the second quantisation picture, a quantum algorithm must map each term of the fermionic Hamiltonian into a sequence of Pauli matrices acting on qubits. The original problem thus becomes a $k$--local Hamiltonian problem, where $k$ depends on the choice of fermion--qubit mapping.

This work introduces a new approach to defining and designing fermion--qubit mappings, which directly leads to a significant reduction in the complexity of some quantum simulation algorithms. There is practical value in any improvement to the sleekness of quantum simulation algorithms' designs. While the $k$--local Hamiltonian problem is $\QMA$--complete~\cite{kempe2006complexity}, enormous value remains in finding or approximating solutions in the average case--scenarios for practical problems, such as molecular electronic structure \cite{babbush2014adiabatic}. Turning from complexity theory, then, to focus on the more practical elimination of redundant costs in steps 0)--3) of existing fermionic simulation algorithms, there have been a number of recent developments in quantum computing that could make solutions to the above problems feasible \cite{aspuru2005simulated,mcclean2014exploiting,kassal2008polynomial,kassal2009quantum,ortiz2001quantum,babbush2014adiabatic,wang2008quantum}. Many of these approaches \cite{aspuru2005simulated,wang2008quantum} rely on phase estimation \cite{kitaev1995quantum,abrams1999quantum} and thus require an impractically large number of qubits and operations in order to keep the register of the quantum computer coherent \cite{peruzzo2014variational}. Algorithms for near--term quantum computers have sprung up in answer to these challenges, such as the variational quantum eigensolver \cite{peruzzo2014variational,mcclean2016theory,mcclean2014exploiting}, which is only just coming within reach of current technology.

\begin{figure}
	\centering
	\includegraphics[width=\linewidth]{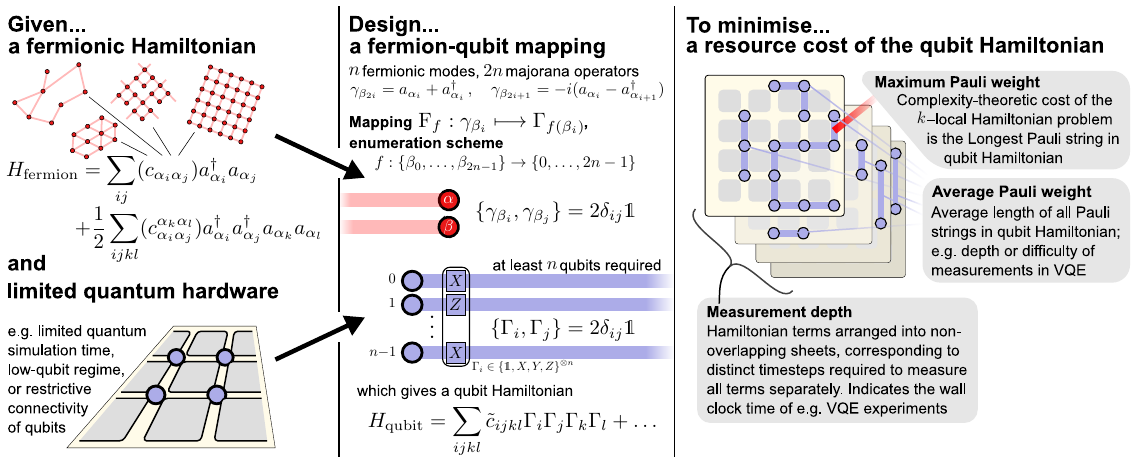}
	\caption{Schematic showing the purpose of a fermion--qubit mapping, the factors that influence its effectiveness, and various cost functions under which one might judge the effectiveness of a mapping.}
	\label{fig:poster-like}
\end{figure}

Some improvements to quantum simulation algorithms come with no caveats: for example, careful parallelisation of Hamiltonian terms during Trotterisation in step 2) of a phase estimation algorithm can cancel out long strings of Pauli matrices, a simple adjustment that does not come at the expense of any additional resources \cite{hastings2015improving}. However, the search for the best fermion--qubit mapping has largely not been met with much progress because the task of simulating fermionic interactions in a complexity--theoretic sense is very difficult to formalise. To bring the problem of fermionic simulation closer to meeting constraints of near--term quantum technology, it is crucially important to establish optimal protocols that strike a balance between competing resources, such as the number of qubits, local quantum operation counts, or quantum circuit depth. This is a complex challenge, as the definition of an optimal protocol may depend on the topology of the physical system to be simulated (e.g.\ the optimal strategy for simulating a fermionic square lattice may be different to that of a molecule), or the limitations of the quantum hardware at hand (e.g.\ only permitting nearest--neighbour interactions between qubits).

The century--old work of Jordan and Wigner \cite{Jordan1993} inspired the idea of fermion--qubit mappings, when Lieb et al.\ used their method to solve the $XY$--model Hamiltonian classically in 1961 \cite{lieb1961two}. In 2001, Ortiz employed the transformation as a fermion--qubit mapping in the first explicit proposal of a quantum simulation of a fermionic Hamiltonian \cite{ortiz2001quantum}. The Jordan--Wigner transformation is intuitive and performs well in simulating nearest--neighbour Hamiltonians acting on 1D chains of fermionic modes, but experiences impracticably large overheads in higher dimensions.

Fermionic systems that permit only local interactions on a 2D or 3D lattice are a major focus of study \cite{barends2015digital, chen2018exact, chen2019bosonization, chen2020exact, chen2023equivalence}. These systems are difficult to simulate, even for quantum computers, because an overwhelming number of local fermionic interactions become non--local -- and hence timely and costly to simulate -- once mapped to the qubit picture. In recent years, there have been a flurry of results tackling this challenge introducing new fermion--qubit mappings as well as generalising the existing ones to higher dimensions \cite{chen2018exact, chen2019bosonization, chen2020exact,  setia2019superfast,jiang2019majorana,phasecraft2020low,steudtner2019quantum,Jiang2020optimalfermionto}. A common theme among these proposals is the use of ancilla qubits \cite{phasecraft2020low, verstraete2005mapping, steudtner2019quantum}. Such mappings can produce local qubit Hamiltonians, and could make small instances of a problem within reach of modest quantum computers.

Our approach makes use of a new degree of freedom in fermion--qubit mappings: the ability to index the fermionic modes in any order, a process we dub \textit{choosing the fermionic enumeration scheme}. While trivial for a string of fermions interacting only with nearest--neighbours modes, the choice of enumeration scheme has the potential to dramatically improve the average locality of the fermion--qubit mappings for fermionic systems in 2D and above. Examining enumeration schemes, and finding the most efficient ones, we show how to reduce various cost functions of the target qubit Hamiltonian relating to scarce physical resources. This requires no additional resources such as additional Hamiltonian terms or ancilla qubits!

Thus, we arrive at the crux of our work: for a given fermionic system and a given quantum computing technology, a fermion--qubit mapping \textit{cannot} be considered optimal unless its fermionic enumeration scheme minimises the scarcest physical resource. Many innovations in fermion--qubit mappings have been in aid of reducing the maximum number of qubits upon which any one term in the Hamiltonian acts \cite{bravyi2002fermionic,jiang2019majorana,steudtner2019quantum,phasecraft2020low} This is an important metric in a complexity--theoretic sense -- it is the `$k$' in `$k$--local Hamiltonian'.

Other metrics may be more practical in near--term quantum algorithms. For example, the number of Pauli measurements in a variational quantum eigensolver depends on the overall number of Pauli matrices appearing in all terms in the qubit Hamiltonian, not just the longest term.
While our method could be tailored to improve many metrics of fermion--qubit mappings, the focus of this work is on improving the average locality of the qubit Hamiltonian terms. The only other work that has used this metric for fermion--qubit mappings is \cite{Jiang2020optimalfermionto}, which we discuss and incorporate into our work in Section \ref{sec:comparison}. 

In a broad class of problem fermionic Hamiltonians with non--local hopping terms, we show that minimising the average Pauli weight of a Jordan--Wigner type mapping is equivalent to the  edgesum problem from graph theory \cite{mitchison1986optimal,garey1974some,garey1978complexity,gareyjohnson1990book}. When compared to the only other ancilla--free mapping for the 2D fermionic lattice -- which uses the S--pattern enumeration scheme, introduced by Verstraete and Cirac~\cite{verstraete2005mapping} -- our method directly reduces the average Pauli weight of the terms in the qubit Hamiltonian by a constant factor ($\approx$\ $13.9\%$). Our optimal fermion--qubit mapping uses a carefully selected enumeration scheme based on a special pattern recognised by Mitchison and Durbin in their seminal work~\cite{mitchison1986optimal}, and culminates in:

{\bf Theorem \ref{thm:physical}} [Informal]
{\it For a fermionic Hamiltonian acting on a system of $n=N^2$ local fermionic modes with interactions only between nearest neighbours on the square $N${}$\times${}$N$ lattice, the Jordan--Wigner transformation that uses the Mitchison--Durbin pattern to enumerate the modes minimises the average Pauli weight of the qubit Hamiltonian.}

The structure of the paper is as follows: in Section~\ref{sec:background} we provide a self--contained introduction to fermion--qubit mappings and identify the role that fermionic enumeration schemes play. We give a new, broad definition that encapsulates all  $n$--mode to $n$--qubit mappings before narrowing our focus to the Jordan--Wigner transformation. We follow with results from complexity theory that will be used to prove our main result.

In Section~\ref{sec:enumeration}, we discuss the maximum Pauli weight and measurement depth of qubit Hamiltonians, which are practical figures of merit worth minimising in fermion--qubit mappings. We present Theorem \ref{thm:physical} for constructing Jordan--Wigner transformations for 2D fermionic lattice systems that they minimise the average Pauli weight of qubit Hamiltonian terms. We argue that our approach can improve simulations for many fermionic systems, and use an heuristic approach to find Jordan--Wigner transformations for certain $n$--mode fermionic systems that provide average Pauli weights shorter than those of naïve alternatives by a factor of $\mathcal{O}(n^{1/4})$.

In Section \ref{sec:aqm}, we explain auxiliary qubit mapping techniques and modify our fermion--qubit mapping to improve Theorem \ref{thm:physical}'s 13.9\% advantage over the Z-- and S--patterns to nearly 38\% using just two ancilla qubits. 

Finally, in Section~\ref{sec:discussion} we discuss open problems and directions for further research. We also mention the qubit routing problem as a potential generalisation of the optimisation problems described in this paper.

\section{Defining fermion--qubit mappings}\label{sec:background}
This section outlines the theory of fermion--qubit mappings:  Section \ref{sec:fqpurpose} describes the motivation, while Section \ref{sec:car} the requirements of a mapping. The naïve definition of the Jordan--Wigner transformation appears in Section \ref{sec:jwnaive}, before we introduce the Jordan--Wigner transformation with fermion enumeration schemes in Section \ref{sec:jw1}, which is the working definition we use for the rest of the paper. We generalise our definition of fermion--qubit mappings in Section \ref{sec:general} to demonstrate that the principle of optimising over enumeration schemes is a valid for improving all fermion--qubit mappings, with a specific example in Section \ref{sec:comparison}. In Section \ref{sec:math} we detail the graph theory related to the problem of finding optimal Jordan--Wigner transformations with fermion enumeration schemes.

\subsection{Goal of fermion--qubit mappings}\label{sec:fqpurpose}

\begin{figure}
	\centering
	\includegraphics[width=0.9\linewidth]{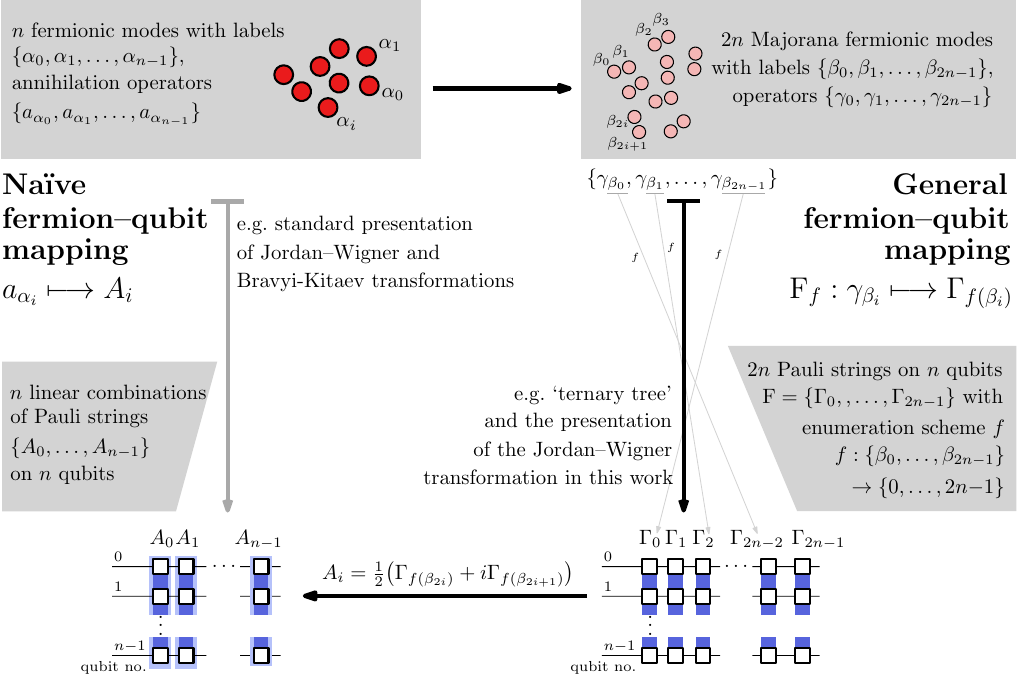}
	\caption{The role of enumeration schemes is hidden in naïve descriptions of fermion--qubit mappings, but becomes apparent in our general definition; it is also recognisable in the ternary tree transformation.}
	\label{fig:fq-mapping}
\end{figure}

Typical second--quantised fermionic Hamiltonians of interest, such as molecular electronic Hamiltonians, describe the energy of quantum  systems with $n$ sites, or \textit{modes}, each of which can be either occupied or unoccupied by a fermion. Use the distinct, unordered, symbolic labels $\{\alpha_0, \alpha_1, \dots, \alpha_{n-1}\}$ to distinguish the $n$ fermionic modes, and denoting the annihilation operator of a fermion in mode $\alpha_i$ by $a_{\alpha_i}$, physical Hamiltonians take the form
\begin{align}
	H_{\text{fermion}} &= \sum_{i,j=0}^{n-1} (c_{\alpha_i \alpha_j}) a_{\alpha_i}^\dagger a_{\alpha_j} \label{eqn:ham}
	 + \frac{1}{2}\sum_{i,j,k,l=0}^{n-1} (c_{\alpha_i \alpha_j}^{\alpha_k \alpha_l}) a_{\alpha_i}^\dagger a_{\alpha_j}^\dagger a_{\alpha_k} a_{\alpha_l}\, , 
\end{align}
where the $c$ are complex coefficients that ensure the hermiticity of the Hamiltonian. As an example, in the Fermi--Hubbard model \cite{hubbard1963electron} the fermions are electrons, and the coefficients $c_{\alpha_i \alpha_j}$ and $c_{\alpha_i \alpha_j}^{\alpha_k \alpha_l}$ are respectively one-- and two--electron overlap integrals. A fermionic Hamiltonian consisting just of quadratic single--particle terms $a_{\alpha_i}^\dagger a_{\alpha_j}$ may easily be diagonalised and its eigenvalues simply read off \cite{nielsen2005fermionic}. The presence of quartic terms $a^\dagger_{\alpha_i} a^\dagger_{\alpha_j} a_{\alpha_k} a_{\alpha_l}$ usually makes simulation of the Hamiltonian classically intractable, and as such an object of interest for prospective quantum algorithms.

The intended use of a fermion--qubit mapping is as a tool to translate from the fermionic picture into Hamiltonians that can be implemented on qubits in a laboratory:
\begin{align}
	&H_{\text{fermion}} \longmapsto  H_{\text{qubit}}  =\  \sum_{ij}( {c}_{\alpha_i \alpha_j}) A_i^\dagger A_j 
	+ \frac{1}{2} \sum_{ijkl} ({c}_{\alpha_i \alpha_j}^{\alpha_k \alpha_l}) A_i^\dagger A_j^\dagger A_k A_l \, . 
\end{align}
A fermion--qubit mapping is a representation of the fermionic algebra in Equation \ref{eqn:ccr1}, characterised by the set $\{A_i\}_{i=0}^{n-1}$ of complex matrix representations of the annihilation operators $\{a_{\alpha_i}\}_{i=0}^{n-1}$.

The operators should be straightforward to write in the basis of $n$--qubit Pauli operators  $\{ \{\mathds{1}, X, Y, Z\}^{\otimes n} \}$, to be practical for real--world quantum technology. We use the notation  $X_i = \sigma^x_i$, $ Y_i= \sigma^y_i$, $Z_i$  $=$ $\sigma^z_i$ for the Pauli operators on the qubit with label $i$.

A further practical caveat is that the fermionic vaccuum state $\ket{\Omega}$, which uniquely satisfies $a_{\alpha_i}\ket{\Omega} = 0$ for all $i$, becomes a product state such as $\ket{0}^{\otimes n}$ on the qubits. The name for this property is \textit{product--preserving}. The popular fermion--qubit mappings -- Jordan--Wigner, Bravyi--Kitaev and ternary tree -- all satisfy this constraint \cite{miller2023bonsai}.

\subsection{Canonical anticommutation relations}\label{sec:car}

As in \cite{bravyi2002fermionic}, our convention is to treat fermionic modes as either occupied or unoccupied, treating any spin--up and spin-down modes as distinct. Thus, a system of $n$ identical fermionic modes inhabits a $2^n$--dimensional complex--valued Hilbert space.

Given a quantum system with $n$ fermionic modes, its Hamiltonian $H_\text{fermion}$ is of the form of Equation \ref{eqn:ham}. The \textit{canonical anticommutation relations} fully describe the fermionic operators:
\begin{align} 
	\{a_{\alpha_i}, a_{\alpha_j}\} = 0\, , \, \{a_{\alpha_i}^\dagger, a_{\alpha_j}^\dagger\} = 0\, , \,  \{a_{\alpha_i}, a_{\alpha_j}^\dagger\} &= \delta_{ij} \mathds{1}\, . \label{eqn:ccr1}
\end{align}

For each fermionic label $\alpha_i \in \{\alpha_0, \dots \alpha_{n-1}\}$, there are two labels  $\beta_{2i}, \beta_{2i+1} \in \{\beta_0, \dots, \beta_{2n-1}\}$ which relate the operators $a_{\alpha_i}$ and $a_{\alpha_i}^\dagger$ via the \textit{Majorana operators}:
\begin{align}
	\gamma_{\beta_{2i}} &= a_{\alpha_i} + a^\dagger_{\alpha_i} \label{eqn:majorana1} \\
	\gamma_{\beta_{2i+1}} &= -i(a_{\alpha_i} - a_{\alpha_i}^\dagger)\, . \label{eqn:majorana2}
\end{align}
The Majorana operators $\{\gamma_{\beta_0}, \gamma_{\beta_1}, \dots, \gamma_{\beta_{2n-1}}\}$ form an alternative set of fermionic operators, equivalent to the complete set of creation and annihilation operators $\{a_{\alpha_0}, a^\dagger_{\alpha_0}, \dots, a_{\alpha_{n-1}}, a^\dagger_{\alpha_{n-1}} \}$. Each of the Majorana operators is its own Hermitian conjugate, and, by construction, neither $\gamma_{\beta_{2i}}$ nor $\gamma_{\beta_{2i+1}}$ are involved in the definition of any other $a_{\alpha_j}$ for $j \neq i$. The Majoranic equivalent to the anticommutation relations in Equation \ref{eqn:ccr1} is thus simply
\begin{align}
	\gamma_{\beta_i} = \gamma_{\beta_i}^\dagger\, , \quad 
	\{\gamma_{\beta_i}, \gamma_{\beta_j}\} = 2 \delta_{ij} \mathds{1}\, .\label{eqn:majoranacc}
\end{align}
for $i,j \in \{0,1,\dots,2n-1\}$.

In Sections \ref{sec:enumeration} and \ref{sec:aqm}, we use the language of creation and annihilation operators $a_\alpha^\dagger, a_{\alpha}$, since they are the ingredients of the Jordan--Wigner transformation. However, Majorana operators are necessary to explore the complete picture of enumeration schemes' role within fermion--qubit mappings, as we explain in Section \ref{sec:general}.

\subsection{The standard (naïve) definition of the Jordan--Wigner transformation}\label{sec:jwnaive}

Suppose that the labels for the fermionic modes were ordered via the scheme  $\alpha_0 \mapsto 0, \alpha_1 \mapsto 1,\dots,\alpha_{n-1}\mapsto n-1$. Then, with the notion of order for the fermionic modes, we could define the system's Fock space with the occupancy number basis $\{\ket{j_{0},j_{1},\dots,j_{{n-1}}} : j_{i} \in \{0,1\}\}$, where $j_{i}$ denotes the occupancy of the $i$th fermionic mode. The annihilation operators act as
\begin{align}
	&a_{\alpha_i} |j_0 \dots \underbrace{0}_{\makebox[0pt]{\footnotesize $i$th mode}}\dots j_{{n-1}}\rangle = 0 \label{eqn:ferm1} \\
	&a_{\alpha_i} \vert j_{0} \dots \underbrace{1}_{\makebox[0pt]{\footnotesize $i$th mode}}\dots j_{{n-1}}\rangle  \label{eqn:ferm2}  =  (-1)^{\sum_{k=1}^{i-1}j_k} \vert j_{0}\dots\underbrace{0}_{\makebox[0pt]{\footnotesize $i$th mode}}\dots j_{{n-1}}\rangle\, ,  
\end{align}
while $a^\dagger_{\alpha_i}$ acts as the Hermitian conjugate of $a_i$. Equations \ref{eqn:ferm1} and \ref{eqn:ferm2} are equivalent to Equation \ref{eqn:ccr1} \cite{nielsen2005fermionic}. This formulation allows us to define the Jordan--Wigner transformation in the original way \cite{Jordan1993}, as a map from an $n$--mode fermionic system to an $n$--qubit system, where the qubits also have ordered labels $0,1,\dots,n-1$:

\begin{definition} \emph{(Naïve Jordan--Wigner transformation.)} \label{defn:naivejw}
	The \textit{naïve Jordan--Wigner transformation} for a system with $n$ fermionic modes is a $2^n$--dimensional representation $\overline{\text{JW}}_{\text{naïve}}$ of the fermionic algebra. In particular, $\overline{\text{JW}}_{\text{naïve}}$ restricts to a bijection between the annihilation operators and a set $\overline{\text{JW}}=\{A_0,A_1,\dots,A_{n-1}\}$:
	\begin{align}
		\overline{\text{JW}}_{\text{naïve}} &:  a_{\alpha_i} \longmapsto A_i   =  \left(\bigotimes_{k=0}^{i-1}Z_k\right)\frac{1}{2} \left(X + iY \right)_{i}\, \, . \label{eqn:aminus} 
	\end{align}
\end{definition}

To verify that this is a valid fermion--qubit mapping, i.e.\ that the $A_i$ replicate Equation \ref{eqn:ccr1}, observe that $\{Z_i,X_i\} = \{Z_i,Y_i\} = 0$. Therefore, the $\{A_i\}_{i=0}^{n-1}$ that characterise $\overline{\text{JW}}_{\text{naïve}}$ satisfy
\begin{align}
	\{A_i^{(\dagger)},A_j^{(\dagger)}\}=0\, , \, \{A_i,A_j^\dagger\} = \delta_{ij}\mathds{1} \, . \label{eqn:fqmapping}
\end{align}

\subsection{A definition for the Jordan--Wigner transformation that incorporates enumeration schemes} \label{sec:jw1}

The standard description of the Jordan--Wigner transformation does not make clear that one of its inherent components is an \textit{enumeration scheme} for the fermionic modes, a bijective mapping from unordered fermionic labels to the natural number labels of qubits. The following definition makes the role of the enumeration scheme explicit:

\begin{definition} \label{defn:jwfenum}
	The \textit{Jordan--Wigner transformation with a fermionic enumeration scheme} for a system with $n$ fermionic modes is a $2^n$--dimensional representation $\overline{\text{JW}}_f$ of the fermionic algebra, equipped with a bijective enumeration scheme $f$ for the fermionic modes:
	\begin{equation}
		f:\{\alpha_0,\dots,\alpha_{n-1}\} \longrightarrow \{0,\dots,n-1\}\, . \label{eqn:fenumerationscheme}
	\end{equation}
	In particular, $\overline{\text{JW}}_f$ restricts to a bijection between the annihilation operators and a set $\overline{\text{JW}}=\{A_0, \dots, A_{n-1}\}$:
	\begin{align}
		\overline{\text{JW}}_f &: a_{\alpha_i} \longmapsto \, A_{f(\alpha_i)}\, , \label{eqn:naivejw}
	\end{align}
	where $A_i$ has the same expression as in Definition \ref{defn:naivejw}. 
\end{definition}

The purpose of generalising the Jordan--Wigner transformation to Definition \ref{defn:jwfenum} is to demonstrate that there is complete freedom in labelling the fermionic modes: no matter the ordering $f$, we recover the canonical anticommutation relations of Equation \ref{eqn:ccr1}. The operators $A_i$ are still drawn from the set $\overline{\text{JW}}$, which satisfies Equation \ref{eqn:fqmapping}. For example, the enumeration scheme $f(\alpha_i) = i$ recovers the naïve mapping from Definition \ref{defn:naivejw}. 

The topic of discussion of Sections \ref{sec:enumeration} and \ref{sec:aqm} is defining the degree of freedom $f$ in $\overline{\text{JW}}_f$ and how to exploit it for material gain. Before we come to that, first let us fully generalise the concept of fermion--qubit mappings.

\subsection{A general definition for fermion--qubit mappings that incorporates enumeration schemes} \label{sec:general}

As it turns out, the freedom of choice to associate each of the annihilation operators $a_{\alpha_i}$ with a unique linear combination of of Pauli strings $A_{f(\alpha_i)}$ in Definition \ref{defn:jwfenum} does not capture the full degree of freedom in the choice of enumeration schemes in general. The Majorana fermionic operators $\{\gamma_{\beta_0},  \dots, \gamma_{\beta_{2n-1}}\}$ from Equations \ref{eqn:majorana1} and \ref{eqn:majorana2} are equivalent building blocks for any fermionic system described by Equation \ref{eqn:ccr1}, and have the benefit of obeying the more concise anticommutation relations of Equation \ref{eqn:majoranacc}.

We propose that the broadest definition of a fermion--qubit mapping from $n$ modes to $n$ qubits is as follows:
\begin{definition} \label{defn:generalmapping} \emph{(General fermion--qubit mapping.)} A \textit{fermion--qubit mapping} for a system with $n$ fermionic modes system is a $2^n$--dimensional representation $\text{F}_f$ of the fermionic algebra, equipped with a bijective enumeration scheme $f$ for the Majorana modes
	\begin{align}
		f:\{\beta_0,\dots,\beta_{2n-1}\} \longrightarrow \{0,1,\dots,2n-1\}\, .
	\end{align}
	The restriction of $\text{F}_f$ to the Majorana operators is a bijection to $2n$ pairwise--anticommuting Pauli strings $\text{F}=\{\Gamma_0,\dots,\Gamma_{2n-1}\}$, with $\Gamma_i \in \{\mathds{1}, X, Y, Z\}^{\otimes n}$. That is,
	\begin{align}
		\text{F}_f & : \gamma_{\beta_i}  \longmapsto \Gamma_{f(\beta_i)}\, . \label{eqn:truemapping}
	\end{align}
\end{definition}

Since the elements $\Gamma_i$ of $\text{F}$ are Pauli strings, they are Hermitian. Thus, with their pairwise--anticommutation,
\begin{equation}
	\Gamma_i = \Gamma_i^\dagger \, , \quad \{\Gamma_i,\Gamma_j\}=2\delta_{ij}\mathds{1} \label{eqn:gammas}
\end{equation}
for all $i,j=0,1,\dots,2n-1$, reproducing the fermionic system described by Equation \ref{eqn:majoranacc} as required.

\begin{corollary}
	The restriction of the fermion--qubit mapping $\text{F}_f$ to the annihilation operators $\{a_{\alpha_0}, \dots, a_{\alpha_{n-1}}\}$ is a bijection to a set $\{A_0,\dots,A_{n-1}\}$ defined by
	\begin{align}
		\text{F}_f : a_{\alpha_i} & \longmapsto A_i = \frac{1}{2} \left( \Gamma_{f(\beta_{2i})} + i \Gamma_{f(\beta_{2i+1})} \right)\, ,
	\end{align}
	for $i=0,1,\dots,n-1$. This arises from Equations \ref{eqn:majorana1} and \ref{eqn:majorana2}, which give for $i=0,1,\dots,n-1$
	\begin{equation}
		a_{\alpha_i} = \frac{1}{2}\left(\gamma_{\beta_{2i}} +i \gamma_{\beta_{2i+1}} \right)\, .
	\end{equation}
	The $A_i$ satisfy Equation \ref{eqn:fqmapping} by construction.
\end{corollary}

Figure \ref{fig:fq-mapping} visualises the relation between the operators $A_i$ arising from a general fermion--qubit mapping in Definition \ref{defn:generalmapping} and the $A_i$ operators that arise from naïve mappings such as the Jordan--Wigner transformation in Definition \ref{defn:naivejw}. The most general form of the Jordan--Wigner transformation is thus:
\begin{example} \emph{(General Jordan--Wigner transformation)} \label{exm:jw}
	The \textit{general Jordan--Wigner transformation} for a system with $n$ fermionic modes is a fermion--qubit mapping $\text{JW}_f$ which restricts to a bijection between the Majorana operators and a set of $2n$ pairwise--anticommuting Pauli strings  $\text{JW} = \{\Gamma_0, \dots , \Gamma_{2n-1}\}$, defined for $i=0,\dots,n-1$ by
	\begin{align}
		\Gamma_{2i} &= \left( \bigotimes_{k=0}^{i-1} Z_k \right) X_i\, , \label{eqn:generalJWmajorana1}
		\\ \Gamma_{2i+1} &=  \left( \bigotimes_{k=0}^{i-1} Z_k \right) Y_i\, . \label{eqn:generalJWmajorana2}
	\end{align}
	That is,
	\begin{align}
		\text{JW}_f &: \gamma_{\beta_{i}} \longmapsto \Gamma_{f(\beta_i)}\, .
	\end{align}
	Note that the Pauli strings $\Gamma_i$ have length $\mathcal{O}( n)$.
	The annihilation operators of the general Jordan--Wigner transformation map to:
	\begin{align}
		\text{JW}_f : a_{\alpha_i}  \longmapsto A_i = \frac{1}{2} \left(\Gamma_{f(\beta_{2i})} + i \Gamma_{f(\beta_{2i+1})} \right) \, , \label{eqn:generaljw}
	\end{align}
	which satisfy Equation \ref{eqn:fqmapping} by construction.
\end{example}

\begin{remark}
	If the fermionic enumeration scheme for the Majorana modes $f$ satisfies $f(\beta_{2i+1}) = f(\beta_{2i})+1$ for all $i=0,\dots,n-1$, then ${\text{JW}}_f$ reduces to $\overline{\text{JW}}_f^{\,^{\,}}$ from Definition \ref{defn:jwfenum}. Thus, the Jordan--Wigner transformation with a fermionic enumeration scheme in Definition \ref{defn:jwfenum} does not capture the true degree of freedom in enumeration schemes -- it is only a special case of the general Jordan--Wigner transformation described in Example \ref{exm:jw}.
\end{remark}

\begin{remark} \emph{(Notation.)}
	The most general form of a fermion--qubit mapping, which we write as $\text{F}_f$, maps Majorana operators to Pauli strings. We use the overline notation $\overline{\text{F}}_{f}$ for mappings that enumerate the $n$ fermionic modes. The expression $\overline{\text{F}}_{\text{naïve}}$ denotes mappings that unnecessarily force a specific enumeration scheme such as $f(\alpha_i) = i$, with the aim to highlight the naïveté. For example, the relation between these different classes of mappings for the Jordan--Wigner transformation is:
	\begin{align}
		\overline{\text{JW}}_{\text{naïve}} &= \overline{\text{JW}}_f \text{ where } f(\alpha_i) = i \, ; \\
		\overline{\text{JW}}_f &= \text{JW}_{g} \text{  where  } g(\beta_{2i})=f(\alpha_i) \text{ and} \\ & \qquad \qquad \qquad \quad  g(\beta_{2i+1})=f(\alpha_{i}+1) \, . \nonumber
	\end{align}
\end{remark}

Definition \ref{defn:generalmapping} encapsulates all forms of the Jordan--Wigner, Bravyi--Kitaev \cite{bravyi2002fermionic}, ternary tree \cite{Jiang2020optimalfermionto}, and other $n$--qubit transformations, and thus allows a fair comparison between them as in Section \ref{sec:comparison} and \cite{nextsteps}. By expanding the definition of the pairwise--anticommuting Pauli strings to $\Gamma_i \in \{ \mathds{1}, X, Y, Z\}^{\otimes m}$ for $m \geq n$ such that \ref{eqn:gammas} still holds, the general definition here could be extended to include mappings with ancilla qubits such as \cite{verstraete2005mapping, steudtner2019quantum, phasecraft2020low}; we leave such extension for future work.

\subsection{Comparison between different mapping types and the notion of an optimal fermion--qubit mapping}\label{sec:comparison}

Through the lens of Definition \ref{defn:generalmapping}, the feature that identifies a mapping as being of Jordan--Wigner type is that it associates Majorana fermionic operators with a set of Pauli strings $\text{JW}$ of the forms in Equations \ref{eqn:generalJWmajorana1} and \ref{eqn:generalJWmajorana2}.

In the  literature, the search for more efficient mappings beyond the Jordan--Wigner transformation resulted in the Bravyi--Kitaev transformation \cite{bravyi2002fermionic}, which yields exponentially shorter Pauli strings in the asymptotic limit. The Bravyi--Kitaev transformation does not outperform the Jordan--Wigner on modest fermionic systems, however, and has a similar T--gate count \cite{tranter2018comparison}. Aside from its intuitive definition, the Jordan--Wigner mapping has also gained widespread use because it demands far fewer degrees of connectivity from the qubit architecture \cite{steudtner2019quantum}.

We can construct the most general definition for other types of fermion--qubit mappings, where the identifying feature of each mapping type is the unique set of Pauli strings to which it maps the Majorana fermionic operators. For example, the most general definition for the Bravyi--Kitaev transformation \cite{bravyi2002fermionic} is:

\begin{example} \emph{(General Bravyi--Kitaev transformation.)} \label{exm:bk}
	The \textit{general Bravyi--Kitaev transformation} for a system with $n$ fermionic modes is a fermion--qubit mapping $\text{BK}_f$ which restricts to a bijection between the Majorana operators and a set of $2n$ pairwise--anticommuting Pauli strings  $\text{BK} =\{\Gamma_0, \dots, \Gamma_{2n-1}\}$, defined for $i = 0, \dots, n-1$ by
	\begin{align}
		\Gamma_{2i} &= X_{U(i)} X_{i} Z_{P(i)}\, , \\
		\Gamma_{2i+1} &=\begin{cases} X_{U(i)} Y_i Z_{P(i)}\, ,  & i \text{ even,} \\ X_{U(i)} Y_i Z_{R(i)} \, , & i \text{ odd.}  \end{cases}
	\end{align}
	That is,
	\begin{align}
		\text{BK}_f &: \gamma_{\beta_{i}} \longmapsto \Gamma_{f(\beta_{i})}\, 
	\end{align}
	Here, using the notation of Seeley et al.\ \cite{seeley2012bravyi}, the sets $U(i)$, $P(i)$ and $R(i)$ are subsets of $\{0,1,\dots, n-1\}$ of size $\sim \log_2n$.
	
	Much like the Jordan--Wigner transformation, this general definition of the Bravyi--Kitaev transformation is more extensive than the usual naïve presentation, which simply maps annihilation operators to their qubit equivalents $\overline{\text{BK}}$ with the trivial enumeration scheme $f:\alpha_i \mapsto  i$:
	\begin{align}
		\overline{\text{BK}}_{\text{naïve}} : a_{\alpha_i} &\longmapsto A_i = \frac{1}{2} \left( \Gamma_{2i} + i \Gamma_{2i+1} \right) \, .    
	\end{align}
\end{example}

\begin{example} \emph{(Ternary tree transformation.)} \label{exm:tt}
	Unlike the naïve versions of the Bravyi--Kitaev and Jordan--Wigner transformations, the base definition of the ternary tree transformation \cite{Jiang2020optimalfermionto} is already similar to Definition \ref{defn:generalmapping}: it is a map $\text{TT}_f$ that identifies the Majorana fermionic operators $\{\gamma_{\beta_0}, \dots, \gamma_{\beta_{2n-1}}\}$ with Hermitian, pairwise--anticommuting Pauli strings $\text{TT} = \{\Gamma_0, \dots, \Gamma_{2n-1}\}$ defined by
	\begin{align}
		\text{TT} = \big\{&\Gamma_0 = X_0 X_1 X_4 ... X_{\frac{n-1}{3}}\, , \label{eqn:ttgammas} \\
		& \Gamma_1 =  X_0 X_1 X_4 ... Y_{\frac{n-1}{3}},  \nonumber \\
		&  \Gamma_2 = X_0 X_1 X_4 ... Z_{\frac{n-1}{3}}\, , \nonumber \\
		& \dots \, ,  \nonumber \\ 
		& \Gamma_{2n-2} = Z_0 Z_3 Z_{12} ... X_{n-1}\, , \nonumber \\
		&\Gamma_{2n-1} =  Z_0 Z_3 Z_{12} ... Y_{n-1} \big\}\, . \nonumber
	\end{align}

	\begin{figure}
		\centering
		\includegraphics[width=\linewidth]{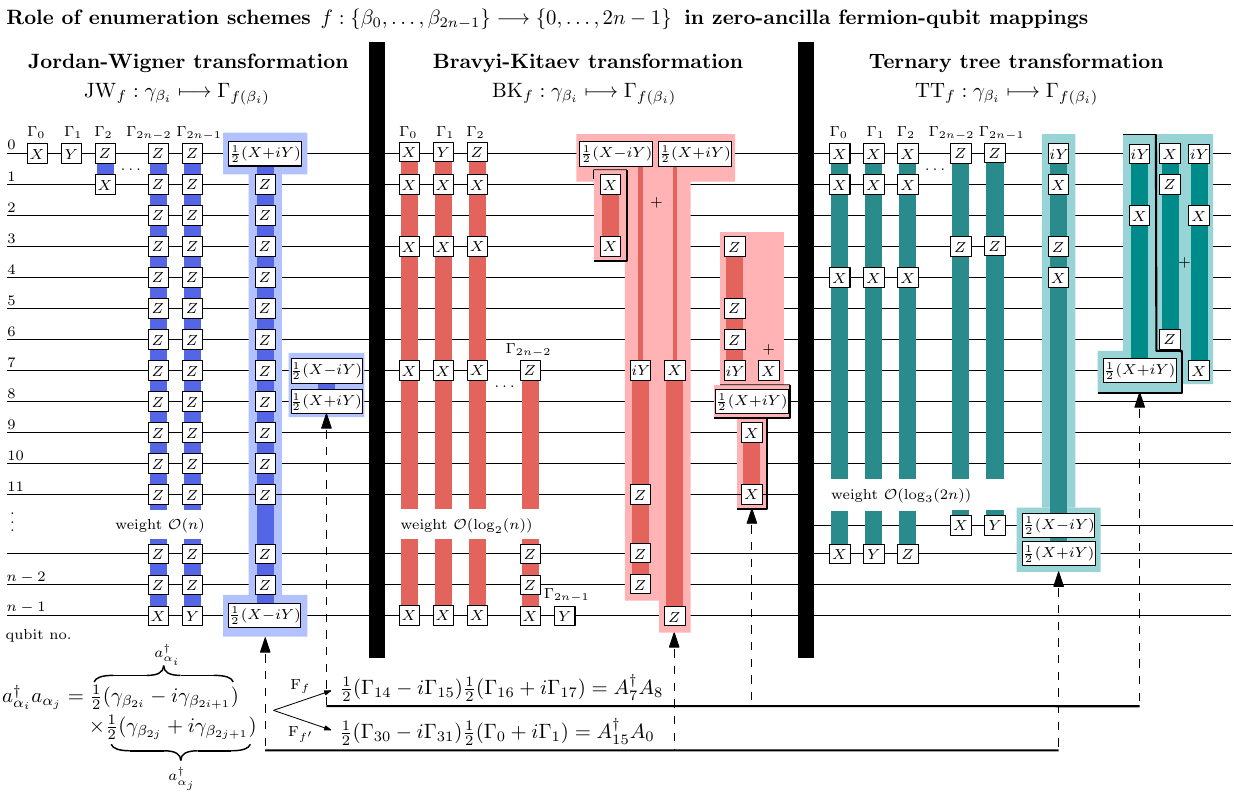}
		\caption{A general fermion--qubit mapping $\text{F}_f$ for a system with $n$ fermionic modes described by Majorana operators $\{\gamma_{\beta_0}, \gamma_{\beta_1}, \dots, \gamma_{\beta_{2n-1}}\}$ consists of a set of pairwise--anticommuting Pauli strings $\text{F}=\{\Gamma_0, \Gamma_1, \dots, \Gamma_{2n-1}\}$ along with an enumeration scheme $f : \{\beta_0, \dots, \beta_{2n-1}\} \mapsto \{0,1,\dots,2n${\,}$-${\,}$ 1\}$. This diagram displays three well--known fermion--qubit mappings: the Jordan--Wigner, Bravyi--Kitaev, and ternary tree mappings, and shows how the enumeration scheme dictates the properties physical operators, which are products of $a_{\alpha_i}^\dagger = \frac{1}{2} (\gamma_{\beta_{2i}} - i \gamma _{\beta_{2i+1}})$ and $a_{\alpha_j} = \frac{1}{2} (\gamma_{\beta_{2j}} + i \gamma _{\beta_{2j+1}})$ for  $i,j \in \{0,1,\dots,n-1\}$. Suppose that $f$ and $f'$ are two distinct Majorana enumeration schemes with $f(\{\beta_{2i},\beta_{2i+1},\beta_{2j},\beta_{2j+1} \})$ $=$ $\{14,15,16,17\}$ respectively, and $f'(\{\beta_{2i},\beta_{2i+1},\beta_{2j},\beta_{2j+1} \})$ $=$ $\{30,31,0,1\}$ respectively. The interaction terms $A_7^\dagger A_8$ and $A^\dagger_{15} A_0$ are noticeably distinct under $f$ compared to $f'$, for all three fermion--qubit mapping types.}
		\label{fig:allmappings}
	\end{figure}

	We mention the mapping here for the purpose of making its dependence on an enumeration scheme $f$ explicit, rather than implicit as in its initial presentation in \cite{Jiang2020optimalfermionto}:
	\begin{align}
		\text{TT}_f & : \gamma_{\beta_i} \longmapsto \Gamma_{f(\beta_{i})}\, .
	\end{align}
	Strictly speaking, Equation \ref{eqn:ttgammas} only describes $\text{TT}$ if $2n+1$ is a power of 3. The mapping can be altered slightly to accommodate other values of $n$ without suffering any ill effects: namely, the remarkable $\sim \log_3 (2n)$ Pauli weight of each of the $\Gamma_i$ qubit operators, which is the provably optimal average Pauli weight of any such set of Pauli strings $\text{F}$ in an $n$--fermion to $n$--qubit mapping \cite{Jiang2020optimalfermionto}.
\end{example}

Examples \ref{exm:jw}--\ref{exm:tt} demonstrate that the Pauli strings $\Gamma_i$ in the Jordan--Wigner, Bravyi--Kitaev and ternary tree mappings have maximum weight $n$, $\sim \log_2(n)$ and $\sim \log_3(2n)$. This might incline one to declare the Bravyi--Kitaev transformation to be an improvement on Jordan--Wigner, and the ternary tree mapping to be the best of all possible mappings.

However, recall from Section \ref{sec:fqpurpose} that the goal of any fermion--qubit mapping is to simulate a fermionic Hamiltonian $H_{\text{fermion}}$ with physical interaction terms of the form $a^\dagger_{\alpha_i} a_{\alpha_j}$ and $a^\dagger_{\alpha_i} a^\dagger_{\alpha_j} a_{\alpha_k} a_{\alpha_l}$. These terms are linear combinations of products of particular Majorana operators $\gamma_{\beta_i}$. Given a fermion--qubit mapping $\text{F}_f$, the terms in the qubit Hamiltonian $H_{\text{qubit}}$ are of the form $A_i^\dagger A_j$ and $A^\dagger_i A^\dagger _j A_k A_l$, which are linear combinations of products of the Pauli strings $\Gamma_i \in \text{F}=\{\Gamma_0, \dots, \Gamma_{2n-1}\}$. Even for the ternary tree mapping, the Pauli weight of the physical interaction terms varies depending on the enumeration scheme $f$, as Figure \ref{fig:allmappings} demonstrates.

Moreover, there are a multitude of properties of the qubit Hamiltonian $H_{\text{qubit}}$ that one might want to minimise depending on the quantum technology at hand. For example, the \textit{average Pauli weight} or the \textit{Pauli measurement depth} of the terms in $H_\text{qubit}$ are quantities that can determine the resource cost in near--term algorithms such as the variational quantum eigensolver. The \textit{maximum Pauli weight} of any one term in $H_\text{qubit}$ is the `$k$' in its $k$--local Hamiltonian problem, and hence a measure of computational complexity. Thus, it could be an appropriate target to minimise for long--term algorithms such as those involving phase estimation.

Therefore, we argue that the notion of an optimal fermion--qubit mapping is only well--defined given these two contexts: the problem Hamiltonian $H_{\text{fermion}}$, and a physical resource cost $C$ to minimise in $H_\text{qubit}$. 
Given these inputs, the total search space for an optimal fermion--qubit mapping $\text{F}_f$ of the form given in Definition \ref{defn:generalmapping} is characterised by:
\begin{enumerate}
	\item The set of mutually--anticommuting Pauli strings $\text{F}=  \{\Gamma_0, \dots, \Gamma_{2n-1}\}$, determined by the type of fermion--qubit mapping, e.g.\ Jordan--Wigner, Bravyi--Kitaev, or ternary tree; and,
	\item The enumeration scheme $f:\{\beta_{0}, \dots, \beta_{2n-1}\} \rightarrow \{0,\dots, 2n-1\}$.
\end{enumerate}
A brute--force search over all the enumeration schemes for one set of Pauli strings $\text{F}$ must reckon with a $(2n)!$--size dataset; once we include the search over all mapping types for $\text{F}$, the task will become even more unwieldy.


In this work, we argue that it \textit{is} possible to make meaningful headway in the search for an optimal fermion--qubit mapping. In this paper, we propose the approach of fixing the mapping type $\text{F}$ and searching for the optimal fermion enumeration scheme $f$.

\begin{definition}\label{defn:opttype} \emph{(Optimal fermion--qubit mapping of type $\emph{F}$.)}
	Let $H_{\text{fermion}}$ be a fermionic Hamiltonian on $n$ modes, let $\text{F}$ be a fermion--qubit mapping type,  and suppose we possess a qubit architecture with some limiting resource. Then for any fermion--qubit mapping $\text{F}_f$, let $C=C(f)$ be the cost function of the qubit Hamiltonian $H_{\text{qubit}}\coloneqq H_{\text{qubit}}(f) = \text{F}_f(H_{\text{fermion}})$ with respect to that resource. The \textit{$C$--optimal fermion--qubit mapping of type} F \textit{for $H_\emph{fermion}$} is the fermion--qubit mapping $\text{F}_{f^*}$, where $f^*$ is a Majorana fermionic enumeration scheme satisfying
	\begin{align}
		f^* = \argmin_{f} C(f)\, .
	\end{align}
\end{definition}

In Section \ref{sec:enumeration}, we find optimal fermion--qubit mappings of Jordan--Wigner type $\overline{\text{JW}}_f$ for various common problem Hamiltonians $H_\text{fermion}$, and various common cost functions $C$ such as the average and maximum Pauli weight of terms in $H_{\text{qubit}}$. In Section \ref{sec:aqm}, we loosen Definition \ref{defn:jwfenum} of the Jordan--Wigner transformation $\overline{\text{JW}}_f$ to find fermion--qubit with two ancilla qubits that outperform the optimal mappings from Section \ref{sec:enumeration}.

\subsection{Complexity--theoretical preliminaries}\label{sec:math}

\begin{figure}
	\centering
	\includegraphics[width=\linewidth]{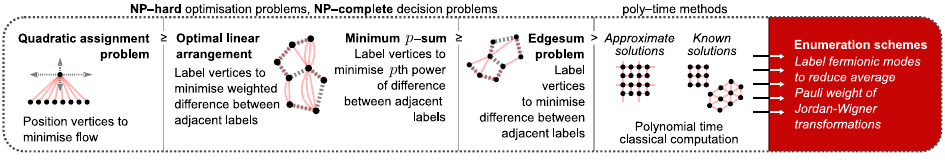}
	\caption{Complexity--theoretic hierarchy of problems generalising the edgesum problem.}
	\label{fig:overview}
\end{figure}

This section introduces the notation and complexity theoretic problems that underpin the results of Section \ref{sec:enumeration}. Awareness of these problems is necessary to find optimal fermion--qubit mappings of Jordan--Wigner type as described in Definition \ref{defn:opttype}. Figure \ref{fig:overview} gives a summary of the hierarchy of these problems, all of which start with some variation of the following ingredients:

\begin{enumerate}
	\item A \textbf{graph} $G=(V,E)$ with $|V|=n$ vertices and edge set $E \subseteq V \times V$, and
	\item a \textbf{weight function} $w:E \longrightarrow \mathds{R}$ which assigns a value $w(\alpha,\beta)$ to the edge $(\alpha,\beta)$ between vertices $\alpha,\beta \in V$, and
	\item a list $L$ of possible \textbf{locations} for the vertices, and
	\item a \textbf{distance function} $d:L\times L \rightarrow \mathds{R}$ describing the spatial separation between the locations.
\end{enumerate}

All optimisation problems in this section have \NP--complete decision versions, and are thus \NP--hard. They share the objective of finding an injective \textit{assignment function} $f:V\longrightarrow L$ to place the vertices in locations so as to minimise a cost function. We call $f$ a \textit{vertex enumeration scheme} when $L =\{0,1,\dots,n-1\}$.

The problems in this section appear in order of descending complexity, in that subsequent problems are special cases of earlier problems. The following problem was introduced by Koopmans and Beckmann \cite{koopmans1957assignment}

\begin{tcolorbox}
	[ title={\bf QUADRATIC ASSIGNMENT}]
	
	{\bf INSTANCE}: Graph $G=(V,E)$, weight function $w:E \rightarrow \mathds{R}$, vertex locations $L$, distance function $d:L\times L \rightarrow \mathds{R}$.\\
	{\bf PROBLEM}: Find the location assignment function $f:V\rightarrow L$ in such a way as to minimise the assignment function
	\begin{align}
		C(f)=\sum_{(\alpha,\beta)\in E} w(\alpha,\beta) \cdot d\left(f(\alpha),f(\beta)\right)\, .
	\end{align}
\end{tcolorbox}


The optimal linear arrangement problem was first studied by Garey and Johnson ~\cite{garey1974some}. Note that it is different to the ``optimal linear arrangement" in their book~\cite{gareyjohnson1990book}.
\begin{tcolorbox}
	[ title={\bf OPTIMAL LINEAR ARRANGEMENT}]
	{\bf INSTANCE}: Graph $G=(V,E)$, weight function $w:E\rightarrow \mathds{Z}$.\\
	{\bf PROBLEM}: Find the enumeration scheme $f:V\rightarrow \{0,1,\dots,n-1\}$ that minimises the assignment function
	\begin{equation}
		C(f) = \sum_{(\alpha,\beta) \in E} w(\alpha,\beta) \cdot \left\vert f(\alpha) - f(\beta) \right \vert\, .
	\end{equation}
\end{tcolorbox}
This is a special case of the quadratic assignment problem: the weight function $w$ is  restricted to integer values, the vertex locations $L$ are $\{1,2,\dots,n\}$, and the distance metric is $d(i,j)=\left\vert i - j\right\vert$.

The minimum $p$--sum problem was studied by Mitchison and Durbin~\cite{mitchison1986optimal}, Garey and Johnson~\cite{garey1978complexity}, and Juvan and Mohar~\cite{JUVAN1992153}.
\begin{tcolorbox}
	[ title={\bf MINIMUM $p$--SUM}]
	{\bf INSTANCE}: Graph $G=(V,E)$, integer $p \in \mathds{Z}$.\\
	{\bf PROBLEM}: Find the enumeration scheme $f:V\rightarrow \{0,1,\dots,n-1\}$ that minimises the assignment function
	\begin{align}
		C^p(f) = \left( \sum_{(\alpha,\beta) \in E} \left\vert f(\alpha) - f(\beta) \right\vert^p \right)^{1/p}\, .
	\end{align}
\end{tcolorbox}
This is a special case of the optimal linear arrangement problem, because its weight function is effectively $w(\alpha, \beta)=|f(\alpha)-f(\beta)|^{p-1}\in \mathds{Z}$. The decision version of the problem is \NP--complete for $p=1$ \cite{garey1974some}, $p=2$ \cite{george1997spectralenvelope} and $p\rightarrow \infty$ \cite{horton1997optimal}. One could also consider the broader class of problems where $p \in \mathds{R}^+$, as done by Mitchison and Durbin \cite{mitchison1986optimal}. These problems are likely to be at least as hard as their integer--$p$ equivalents. 

\subsubsection{Special cases of minimum \texorpdfstring{$p$}{p}--sum}\label{sec:np}

The edgesum problem was introduced in~\cite{garey1974some}, and the study of this problem is the key ingredient to our optimal Jordan--Wigner transformations in Section \ref{sec:enumeration}:
\begin{tcolorbox}
	[ title={\bf EDGESUM (or SIMPLE OPTIMAL LINEAR ARRANGEMENT, MINIMUM 1--SUM)}]
	{\bf INSTANCE}: Graph $G=(V,E)$.\\
	{\bf PROBLEM}: Find the enumeration scheme $f: V \rightarrow \{0,1,\dots,n-1\}$ that minimises the assignment function
	\begin{align}
		C(f)=\sum_{(\alpha,\beta)\in E} \left\vert f(\alpha) - f(\beta) \right \vert \, .
	\end{align}
\end{tcolorbox}

\begin{figure}
	\centering
	\includegraphics[width=\textwidth]{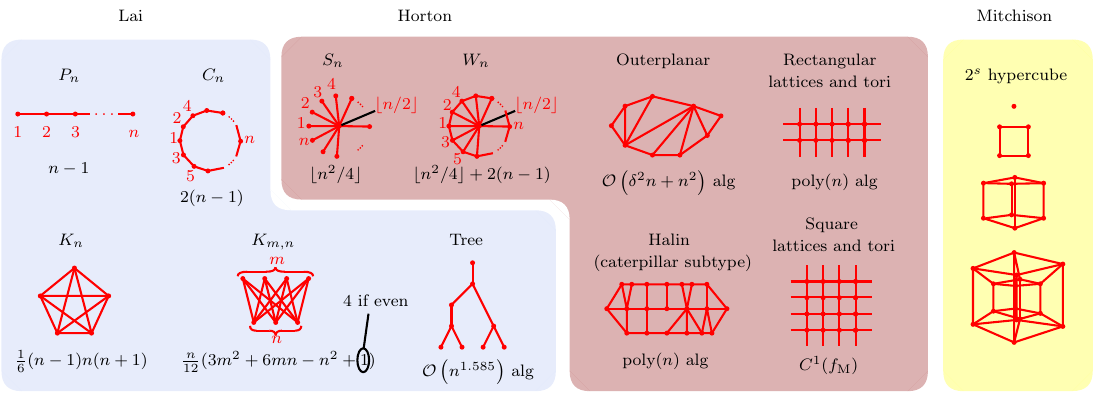}
	\caption{Graph families with known solutions to their edgesum problems. The formula for the minimum edgesum appears beneath each graph's image, if it exists, or the cost of the best--known classical algorithm for solving that graph's edgesum problem if not. References for the solution of these graphs appear in Lai \cite{lai1999survey}, Horton \cite{horton1997optimal}, and Mitchison's \cite{mitchison1986optimal} works as indicated.}
	\label{fig:solvedgraphs}
\end{figure}

The decision version of this problem is known to be $\NP$--complete via a reduction to the simple max--cut problem  \cite{garey1974some}. Figure \ref{fig:solvedgraphs} surveys the graphs with known solutions, as of the date of publication of Horton's PhD thesis \cite{horton1997optimal} which cites a solution for outerplanar graphs by Frederickson and Hambrusch \cite{frederickson1988planar}. Other sources include a survey by Lai \cite{lai1999survey}, and solution for tree graphs by Chung \cite{chung1984optimal}.

As an example of a comparable problem, consider the following:
\begin{tcolorbox}
	[ title={\bf BANDWIDTH (or MINIMUM $\infty$--SUM)}]
	{\bf INSTANCE}: Graph $G=(V,E)$.\\
	{\bf PROBLEM}: Find the enumeration scheme $f: V\rightarrow \{0,1,\dots,n-1\}$ that minimises cost function
	\begin{align}
		C^\infty(f) = \max_{(\alpha,\beta) \in E} \left \vert f(\alpha) - f(\beta) \right \vert \, .
	\end{align}
\end{tcolorbox}
This is the minimum $p$--sum problem as $p\rightarrow \infty$. The decision version of the bandwidth problem is $\NP$--complete for e.g.\ an arbitrary tree graph \cite{garey1978complexity}; conversely, Saxe proves that the decision problem as to whether the bandwidth is less than or equal to $k=O(1)$ is efficiently solvable \cite{doi:10.1137/0601042}.

	\section{Optimal fermion--qubit mappings of Jordan--Wigner type}\label{sec:enumeration}

This section demonstrates the key idea of our approach: using enumeration schemes to optimise fermion--qubit mappings for practical cost functions.
Section \ref{sec:objective} establishes the scope of the search for optimal Jordan--Wigner mappings in this work, defining optimality in terms of practical cost functions for quantum computers. Section \ref{sec:hams} sets out the broad class of problem Hamiltonians that this search will consider. Section \ref{sec:enum1} derives expressions for the cost functions of the search, which are related to problems in graph theory. Section \ref{sec:mpattern} states Theorem \ref{thm:physical} for minimising the average Pauli weight of a Jordan--Wigner transformation of a system of fermions interacting in a square lattice, with proof in Section \ref{sec:proofs}. In Section \ref{sec:cellular}, we consider an heuristic approach to minimising the average Pauli weight of a Jordan--Wigner transformation of cellular lattice fermionic systems. In Section \ref{sec:pg1}, we consider the task of minimising the average $p$th power of Pauli weight of a Jordan--Wigner transformation of square lattice fermionic systems, for $p>1$.


\subsection{Objective} \label{sec:objective}

In practice, it is very difficult to find the optimal fermion--qubit mapping for a given Hamiltonian $H_\text{fermion}$ as defined in Definition \ref{defn:opttype}. Within the scope of this paper, instead consider the subproblem of finding the optimal Jordan--Wigner transformation with a fermionic enumeration scheme as in Definition \ref{defn:jwfenum}:

\begin{definition}\label{defn:optjw}
	\emph{(Optimal Jordan--Wigner transformations with fermionic enumeration schemes.)} Let $H_\text{fermion}$ be a fermionic Hamiltonian on $n$ modes. Let $C=C(f)$ be the cost function of the qubit Hamiltonian $H_{\text{qubit}} = H_{\text{qubit}}(f) \coloneqq \overline{\text{JW}}_f(H_\text{fermion})$ with respect to some scarce physical resource. The \textit{$C$--optimal Jordan--Wigner transformation with a fermionic enumeration scheme for $H_\emph{fermion}$} is the fermion--qubit mapping $\overline{\text{JW}}_{f^*}$, where $f^*$ is a fermionic enumeration scheme satisfying
	\begin{align}
		f^* = \argmin_f C(f) \, .
	\end{align}
\end{definition}

In this work, we discuss optimal Jordan--Wigner transformations with fermionic enumeration schemes for three example cost functions:

\begin{example}\label{exm:apv} \emph{($C=\text{APV}$, average Pauli weight.)}
	Given a fermionic Hamiltonian $H_{\text{fermion}}$, a \textit{Jordan--Wigner mapping with minimum average Paul weight} is a fermion--qubit mapping $\overline{\text{JW}}_{f^*}$ where the fermionic enumeration scheme $f^*$ satisfies
	\begin{align}
		f^* &=  \argmin_f \big( \text{APV}(H_\text{qubit}(f)) \big) \, ,
	\end{align}
	where $\text{APV}(H_\text{qubit})$ is the average Pauli weight of all terms in $H_\text{qubit}(f) = \overline{\text{JW}}_f(H_\text{fermion})$.
	
	The average Pauli weight of a qubit Hamiltonian could be a limiting resource in near--term algorithms such as the variational quantum eigensolver (VQE), where $\text{APV}(H_{\text{qubit}})$ represents the total number of single--qubit measurements involved in each step of the VQE.
\end{example}

\begin{example}
	\label{exm:mpv}
	\emph{($C=\text{MPV}$, maxium Pauli weight.)}
	As in Example \ref{exm:apv}, except that $f^*$ satisfies
	\begin{align}
		f^* &= \argmin_f \big( \text{MPV}(H_\text{qubit}(f)) \big) \, ,
	\end{align}
	where $\text{MPV}(H_{\text{qubit}})$ is the maximum Pauli weight of any one term in the qubit Hamiltonian.
	
	The maximum Pauli weight of $H_{\text{qubit}}$ is a complexity--theoretic measure of the cost of simulating $H_{\text{qubit}}$ via long--term algorithms such as phase estimation. In the $k$--local Hamiltonian problem for $H_{\text{qubit}}$, the quantity $\text{MPV}(H_{\text{qubit}})$ is the value $k$.
\end{example}

\begin{example}
	\label{exm:mdepth}
	\emph{($C=\text{MD}$, Pauli measurement depth)}
	As in Example \ref{exm:apv}, except that $f^*$ satisfies
	\begin{align}
		f^* = \argmin_f \big( \text{MD}(H_{\text{qubit}}(f)) \big) \, ,
	\end{align}
	where $\text{MD}(H_\text{qubit})$ is the minimum number of groupings of the terms in $H_\text{qubit}$ such that each group contains no terms that act on shared qubits.
	
	The quantity $\text{MD}(H_\text{qubit})$ corresponds to the number of distinct timesteps required to run all Hamiltonian terms in parallel, indicating, for example, the total clock time required to perform all measurements during an iteration of a variational quantum eigensolver algorithm.
\end{example}

\subsection{Problem Hamiltonians}\label{sec:hams}

The terms in a general fermionic Hamiltonian $H_{\text{fermion}}$ as specified in Equation \ref{eqn:ham} represent quadratic one--particle and quartic two--particle interactions. The Jordan--Wigner transformation $\overline{\text{JW}}_f$ maps each of the one--particle terms to 
\begin{align}
	&\overline{\text{JW}}_f : a^\dagger_{\alpha_i} a_{\alpha_i} \longmapsto  \frac{1}{2}\left(\mathds{1}-Z\right)_{f(\alpha_i)}\\
	&\overline{\text{JW}}_f:	a^\dagger_{\alpha_i} a_{\alpha_j} \longmapsto  \label{eqn:simplehopping} \frac{1}{2} \left(X-iY\right)_{f(\alpha_i)} \left(\bigotimes_{k=f(\alpha_i)}^{f(\alpha_j)-1} Z_k \right)\frac{1}{2} \left(X + iY\right)_{f(\alpha_j)} \, ,
\end{align}
where $f(\alpha_i)<f(\alpha_j)$ in Equation \ref{eqn:simplehopping}. However, the term $c_{\alpha_i \alpha_j}a^\dagger_{\alpha_i} a_{\alpha_j}$ will always appear in $H$ alongside its conjugate term $(c_{\alpha_ i \alpha_j})^*a^\dagger_{\alpha_j} a_{\alpha_i}$. The expression $c_{\alpha_i \alpha_j}a^\dagger_{\alpha_i} a_{\alpha_j} +(c_{\alpha_ i \alpha_j})^*a^\dagger_{\alpha_j} a_{\alpha_i}$ is the \textit{hopping term} between fermionic modes $\alpha_i$ and $\alpha_j$. The hopping terms of $H_\text{fermion}$ lead to non--local terms in the qubit Hamiltonian of the form
\begin{align} 
	\overline{\text{JW}}_f : c_{\alpha_i \alpha_j}a^\dagger_{\alpha_i} a_{\alpha_j} &+(c_{\alpha_ i \alpha_j})^*a^\dagger_{\alpha_j} a_{\alpha_i} \label{eqn:pairs}   \\
	&\longmapsto  \frac{\Re(c_{\alpha_i \alpha_j})}{2} \left(\bigotimes_{k=f(\alpha_i)+1}^{f(\alpha_j)-1} Z_{k} \right)   \times  \left(X_{f(\alpha_i)} \otimes X_{f(\alpha_j)} + Y_{f(\alpha_i)} \otimes Y_{f(\alpha_j)} \right) \nonumber  \\
	&	\quad + \frac{\Im(c_{\alpha_i \alpha_j})}{2} \left(\bigotimes_{k=f(\alpha_i)+1}^{f(\alpha_j)-1} Z_k\right)  \nonumber  \times  \left( X_{f(\alpha_i)} \otimes Y_{f(\alpha_j)} +Y_{f(\alpha_i)} \otimes X_{f(\alpha_j)}\right)\, . \nonumber 
\end{align}
In the case where the coefficients of the conjugate one--particle interaction terms are $c_{ij}=(c_{ij})^*=1$, the transformation of the hopping term is simply
\begin{align}
	& \overline{\text{JW}}_f	: a^\dagger_{\alpha_i} a_{\alpha_j} + a^\dagger_{\alpha_j} a_{\alpha_i} \longmapsto  	\label{eqn:grp}   \frac{1}{2} \left(\bigotimes_{k=f(\alpha_i)+1}^{f(\alpha_j)-1} Z_k \right) (X_{f(\alpha_i)} \otimes X_{f(\alpha_j)} + Y_{f(\alpha_i)} \otimes Y_{f(\alpha_j)})\, .
\end{align}
Figure \ref{fig:fullhopping} illustrates the conversion of these simpler hopping terms into qubit Hamiltonian terms.

\begin{figure}
	\centering
	\includegraphics[width=0.9\linewidth]{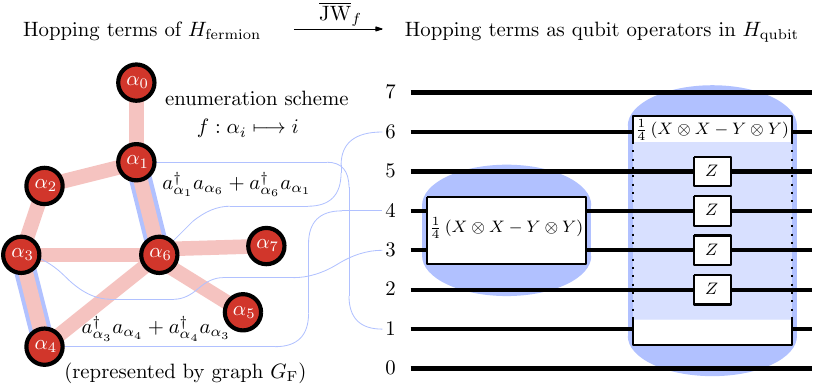}
	\caption{Jordan--Wigner transformation with a fermionic enumeration scheme $\overline{\text{JW}}_f$, in this diagram with the naïve enumeration scheme $f: \alpha_i \mapsto i$. For any enumeration scheme $f$, the fermion--qubit mapping $\overline{\text{JW}}_f$ transforms hopping terms $a_{\alpha_i}^\dagger a_{\alpha_j} + a_{\alpha_j}^\dagger a_{\alpha_i} \mapsto A_{f(\alpha_i)}^\dagger A_{f(\alpha_j)} + A_{f(\alpha_j)}^\dagger A_{f(\alpha_i)}$.}
	\label{fig:fullhopping}
\end{figure}

Literature on quantum simulation algorithms tends to permit quartic terms in the fermionic Hamiltonian that transform into local qubit operations. For example, Verstraete and Cirac \cite{verstraete2005mapping} and Derby and Klassen \cite{phasecraft2020low} consider terms of the form
\begin{align}
	& \overline{\text{JW}}_f: 	a_{\alpha_i}^\dagger a_{\alpha_i} a^\dagger_{\alpha_j} a_{\alpha_j} \longmapsto \label{eqn:quartic} \frac{1}{4} (\mathds{1} - Z)_{f(\alpha_i)} (\mathds{1}-Z)_{f(\alpha_j)}\, . 
\end{align}
The quartic terms make diagonalisation of the fermionic Hamiltonian $H_{\text{fermion}}$ exponentially difficult to solve on a classical computer. However, they become simple 1-- or 2--local Pauli operations after the Jordan--Wigner transformation maps them into the local Pauli operations in Equation \ref{eqn:quartic}. Thus, a fermionic Hamiltonian $H_{\text{fermion}}$ consisting of self--interactions $a^\dagger_{\alpha_i} a_{\alpha_i}$, hopping terms, and quartic terms of the form $a^\dagger_{\alpha_i} a_{\alpha_i} a^\dagger_{\alpha_j} a_{\alpha_j}$ transforms under $\overline{\text{JW}}_f$ to a qubit Hamiltonian $H_{\text{qubit}}$ where the only non--local terms are of the form in Equation \ref{eqn:simplehopping}. Moreover, the hopping terms are the only Pauli strings with weight dependent on the enumeration scheme $f$ of the Jordan--Wigner mapping.

Consider fermionic Hamiltonians of the general form
\begin{align}
	H_\text{fermion} &= \sum_{i,j=0}^{n-1} (c_{\alpha_i \alpha_j}) a^\dagger_{\alpha_i} a_{\alpha_j}  + \sum_{i,j=0}^{n-1}(c'_{\alpha_i \alpha_j}) a^\dagger _{\alpha_i} a_{\alpha_i} a_{\alpha_j}^\dagger  a_{\alpha_j}   \\
	& \qquad  \qquad \qquad \qquad \qquad  \text{(hereon referred to as Eqn.\ \ref{eqn:quartic} terms)} \nonumber \\
	& = \sum_{\text{pairs }(i,j)} \left( (c_{\alpha_i \alpha_j} ) a^\dagger_{\alpha_i} a_{\alpha_j} + (c_{\alpha_i \alpha_j}^*) a_{\alpha_i} a_{\alpha_j}^\dagger \right)   +  \text{Eqn.\ \ref{eqn:quartic} terms} \\ 
	&= \sum_{(\alpha, \beta) \in E} \left( (c_{\alpha \beta} ) a_{\alpha}^\dagger a_\beta + (c_{\alpha \beta} ^* ) a_\alpha a_\beta^\dagger \right) \label{eqn:problemhamtemp} + \text{Eqn.\ \ref{eqn:quartic} terms}\, , 
\end{align}
where $E$ is the edge set of the graph $G_\text{F}(V,E)$ with $V=\{\alpha_0, \dots, \alpha_{n-1}\}$ and $(\alpha, \beta) \in E \subseteq V \times V$ if and only if $c_{\alpha \beta} \neq 0$. 

For convenience of notation, assume that $c_{\alpha \beta} = 1$ for all $(\alpha, \beta) \in E$, so that only hopping terms of the form in Equation \ref{eqn:simplehopping} are of concern. If our task were to find the optimal fermion--qubit mapping as in Definition \ref{defn:opttype}, this assumption would reduce the variety of problem Hamiltonians described by Equation \ref{eqn:problemhamtemp}. However, focused as we are on finding optimal Jordan--Wigner transformation as defined in Examples \ref{exm:apv}--\ref{exm:mdepth}, the assumption $c_{\alpha \beta} = 1$ comes with no loss of generality in the form of the problem Hamiltonians. This is because the weights of the Pauli strings of $H_\text{qubit} = \overline{\text{JW}}_f(H_\text{fermion})$ in Equation \ref{eqn:pairs} with coefficients $\Re(c_{\alpha \beta})$ and $\Im(c_{\alpha \beta})$ are equal. Properties that are only concerned with the length of Hamiltonian terms such as $\text{APV}$, $\text{MPV}$ and $\text{MD}$ are unchanged by setting all nonzero $c_{\alpha \beta}$ to 1. For the remainder of this paper, we can thus assume problem Hamiltonians take the form:
\begin{align}
	H_\text{fermion} &= \sum_{(\alpha, \beta) \in E} \left( a^\dagger_\alpha a_\beta + a_\alpha a^\dagger _\beta \right) \label{eqn:problemham} + \text{Eqn.\ \ref{eqn:quartic} terms} \, . 
\end{align}

This assumption would not hold if instead the task were to find, for example, the optimal enumeration scheme for Bravyi--Kitaev type mappings $\overline{\text{BK}}_f$. This is because the Bravyi--Kitaev transformation does not transform the hopping terms of $H_{\text{fermion}}$ in Equation \ref{eqn:problemhamtemp} into a sum of Pauli strings of equal length \cite{seeley2012bravyi,nextsteps}.

\subsection{Expressions for cost functions of qubit Hamiltonians}\label{sec:enum1}

For problem Hamiltonians of the form in Equation \ref{eqn:problemham}, there is a simple relation for the Pauli weight of the Jordan--Wigner transformation of hopping terms:
\begin{align}
	\text{Pauli weight}\big(\overline{\text{JW}}_f(a_{\alpha}^\dagger a_\beta &+ a_\alpha a_\beta^\dagger) \big) \label{eqn:pauliweight}  = |f(\alpha) - f(\beta)| + 1\, , 
\end{align}
since the Pauli strings of $\overline{\text{JW}}_f(a_{\alpha}^\dagger a_\beta + a_\alpha a_\beta^\dagger)$ contain $(|f(\alpha) - f(\beta)| - 1)$ $Z$ matrices and two $X$ or $Y$ matrices, from Equation \ref{eqn:simplehopping}.

Recall from Section \ref{sec:math} that $C^p(f)$, the $p$--sum of a graph $G(V,E)$ with vertex enumeration scheme $f$, is 
\begin{align}
	C^p(f) \coloneqq \bigg(\sum_{(\alpha,\beta) \in E} |f(\alpha)-f(\beta)|^p \bigg)^{1/p}\, .
\end{align}
This notation, along with Equation \ref{eqn:pauliweight}, allows for simple expressions for the cost functions $\text{APV}$ and $\text{MPV}$ from Examples \ref{exm:apv} and \ref{exm:mpv}, respectively. 

\textbf{Average Pauli weight (Example \ref{exm:apv}).}
The formula for the average Pauli weight of a Jordan--Wigner transformation  $\overline{\text{JW}}_f$ of a problem Hamiltonian $H_\text{fermion}$ of the form in Equation \ref{eqn:problemham} is
\begin{align}\label{eqn:avgpauli}
	\text{APV}(f) \coloneqq \frac{C^1(f)}{|G_{\mathrm{F}}|}+1\, ,
\end{align}
where $G_\text{F}$ is the number of edges in $G_{\mathrm{F}}$. In fact, Equation \ref{eqn:avgpauli} neglects the contribution of the quartic terms from Equation \ref{eqn:quartic}  to the average Pauli weight of $H_\text{fermion}$, because they are independent of the fermion enumeration scheme $f$ and hence have no effect in determining the optimal enumeration scheme $f^*$ as described in Example \ref{exm:apv}. The extra term of 1 accounts for the fact that the weight of a Pauli string between qubits with labels $i$ and $j$ is $|i-j|+1$.

\begin{figure}
	\centering
	\includegraphics[width=\linewidth]{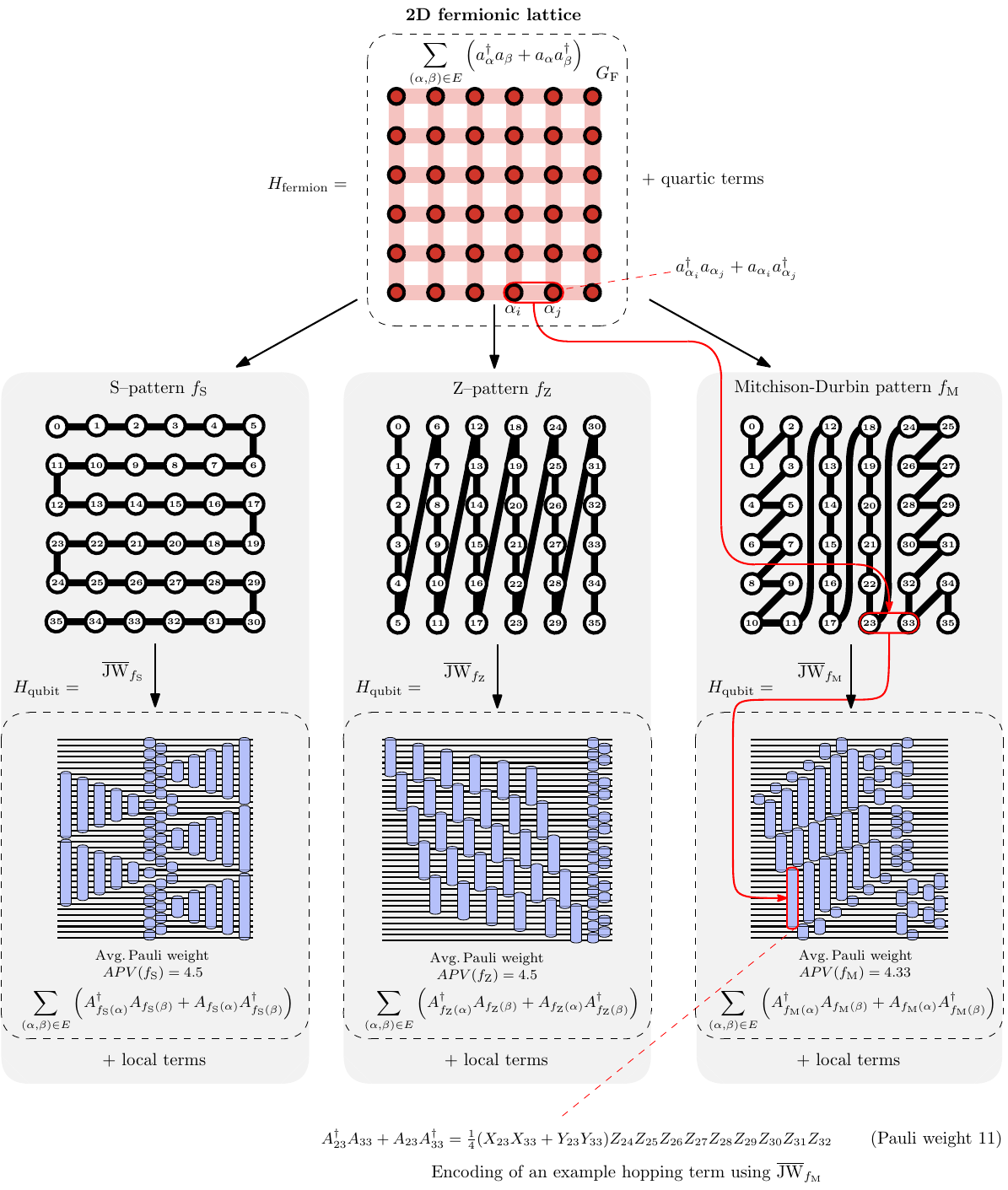}
	\caption{The Jordan--Wigner transformation $\overline{\text{JW}}_f$ produces qubit Hamiltonians that depend on the choice of enumeration scheme $f$. If the fermionic Hamiltonian contains hopping terms that form a square lattice, the Mitchison--Durbin pattern produces qubit Hamiltonians with hopping terms of the minimum possible average weight.}
	\label{fig:latticehop}
\end{figure}

Figure \ref{fig:latticehop} lists the average Pauli weight of the Jordan--Wigner transformation $\overline{\text{JW}}_f$ of a fermionic Hamiltonian $H_\text{fermion}$ with fermionic interaction graph $G_\text{F}$ equal to the 6$\times$6 lattice for three enumeration schemes: the S--pattern $f_\mathrm{S}$, the Z--pattern $f_{\mathrm{Z}}$, and the Mitchison--Durbin pattern, $f_{\mathrm{M}}$, which underlies the results of Section \ref{sec:mpattern}.

Finding Jordan--Wigner transformations with minimum average Pauli weight is equivalent to finding enumeration scheme of $G_\text{F}$ that minimise $C^1(f)$, which is the \NP--hard edgesum problem from Section \ref{sec:math}. In Section \ref{sec:mpattern}, we present Theorem \ref{thm:physical} for an enumeration scheme $f_{\mathrm{M}}$ that minimises the average Pauli weight of a Jordan--Wigner transformation of $H_\text{fermion}$ where if $G_\mathrm{F}$ is the $N${}$\times${}$N$ square lattice. Through Corollary \ref{cor:comparison}, this is an improvement of up to $\approx 13.9$\% upon existing methods.

\textbf{Average $p$th power of Pauli weight.} One might want to penalise the qubit Hamiltonian via some other measure, such as the average $p$th power of the Pauli weights where $p \in \mathds{R}^+$. Minimising this property corresponds to the \NP--hard minimum $p$--sum problem of minimising $C^p(f)$ from Section \ref{sec:math}.
In Section \ref{sec:pg1}, we present numerical results for the minimum $p$--sum problem to illustrate how one might optimise Jordan--Wigner transformations for this objective function.

\textbf{Maximum Pauli weight (Example \ref{exm:mpv})}
The cost function for the maximum Pauli weight of a Jordan--Wigner transformation $\overline{\text{JW}}_f$ is
\begin{align}
	\text{MPV}(f) = C^{\infty}(f) + 1\, ,
\end{align}
which corresponds to the \NP--hard bandwidth problem from Section \ref{sec:math}. However, for fermionic Hamiltonians $H_\text{fermion}$ with hopping terms on a square--lattice $G_\text{F}$, it is straightforward to see that any fermionic enumeration scheme with $C^\infty(f) = N$  will minimise the maximum Pauli weight of the qubit Hamiltonian. One such enumeration scheme is the Z--pattern $f_\text{Z}$.

\textbf{Pauli measurement depth (Example \ref{exm:mdepth}).} Finding Jordan--Wigner transformations with minimum Pauli measurement depths is beyond the scope of this work. Figure \ref{fig:paulidepth} demonstrates the variation of Pauli measurement depth of $H_\text{qubit} = \overline{\text{JW}}_f(H_\text{fermion})$, where the fermionic interaction graph $G_\text{F}$ is the 6$\times$6 lattice, for different enumeration schemes: $f_\mathrm{S}$, $f_{\mathrm{Z}}$, and $f_{\mathrm{M}}$.

\textbf{Other properties.} One may also consider other cost functions. For example, if the qubits have a connectivity graph $G_\mathrm{Q}$, it may be desirable to make the Pauli strings in $H_\text{qubit}$ as local as possible on $G_\mathrm{Q}$. A discussion of this optimisation appears in Section \ref{sec:routing}.

\begin{figure}
	\centering
	\includegraphics[width=0.5\linewidth]{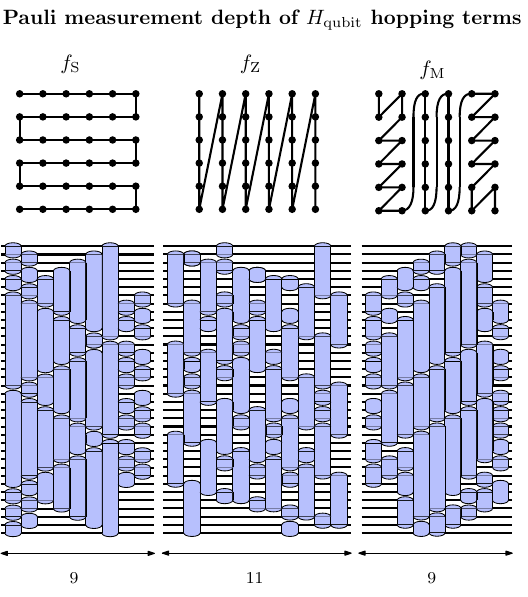}
	\caption{
		Comparison of the Pauli measurement depth of the hopping terms in qubit Hamiltonians $H_\text{qubit} = \overline{\text{JW}}_f(H_\text{fermion})$ produced by three different enumeration schemes on a 6$\times$6 square lattice of fermions.
	}
	\label{fig:paulidepth}
\end{figure}

\subsection{Jordan--Wigner transformations of square--lattice fermionic Hamiltonians with minimum average Pauli weight}\label{sec:mpattern}

Consider fermionic problem Hamiltonians $H_\text{fermion}$ of the form in Equation \ref{eqn:problemham} with interaction graphs $G_\text{F}$ equal to the $N${}$\times${}$N$ square lattice. Prior to this work, the default way to enumerate the fermionic modes, regardless of target cost function, has been to number them row--by--row in either a zig-zagging or snake-like pattern, as in Figure \ref{fig:paulidepth} variations~\cite{verstraete2005mapping,steudtner2019quantum}.
The solution to the edgesum problem for $G_\text{F}$ is related to the work of Graeme Mitchison and Richard Durbin, from their studies of the organisation of nerve cells in the brain cortex~\cite{mitchison1986optimal}. Figure \ref{fig:mitchisondurbin} displays the Mitchison--Durbin pattern $f_\text{M}$, which is dramatically different to these more conventional patterns. The proof of our main result  below, Theorem \ref{thm:physical}, makes use of this curious arrangement to construct a Jordan--Wigner transformation with minimum average Pauli weight.

\begin{figure}
	\centering
	\includegraphics[width=0.5\linewidth]{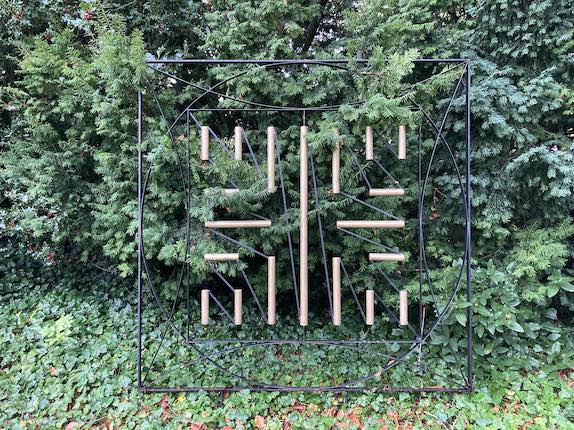}
	\caption{Graeme Mitchison's sculpture of the optimal numbering of a 7$\times$7 array. Reproduced with the permission of Richard Durbin.}
	\label{fig:mitchisondurbin}
\end{figure}

\begin{theorem}\label{thm:physical} \emph{(Jordan--Wigner transformations of square--lattice fermionic Hamiltonians with minimum average Pauli weight.)}
	Given a system of $n=N^2$ fermionic modes, suppose that the system has a Hamiltonian $H_\text{fermion}$ of the form in Equation \ref{eqn:problemham} with a square--lattice interaction graph $G_\text{F}$. Then, the fermion--qubit mapping $\overline{\text{JW}}_{f_\text{M}}$ is a Jordan--Wigner transformation with minimum average Pauli weight for $H_\text{fermion}$, where $f_\text{M}$ is the Mitchison--Durbin pattern.
\end{theorem}

\begin{proof}
	Proof in Section \ref{sec:proofs}.
\end{proof}

\begin{remark}\label{cor:edgesumnp}
	Finding Jordan--Wigner transformations with minimum average Pauli weight for arbitrary Hamiltonians $H_\text{fermion}$ of the form in Equation \ref{eqn:problemham} is $\NP$--hard. By extension, this applies to the general class of Hamiltonians in Equation \ref{eqn:ham}.
\end{remark}

Remark \ref{cor:solved} details all known scenarios to date where the optimal fermionic enumeration scheme is solvable in \poly$(n)$ time.

\begin{remark}\label{cor:solved}\emph{(Solutions for other graph types $G_\mathrm{F}$.)}
	If the Hamiltonian $H_\text{fermion}$ is as defined in Theorem \ref{thm:physical}, and if $G_{\mathrm{F}}$ belongs to any of the graph families in Figure \ref{fig:solvedgraphs}, then a classical computer can efficiently find the Jordan--Wigner transformation with minimum average Pauli weight for $H_\text{fermion}$.
\end{remark}



\begin{corollary} \label{cor:comparison}Using the Mitchison--Durbin pattern $f_{\mathrm{M}}$ in a Jordan--Wigner transformation for a fermionic system with square--lattice--interacting hopping terms produces Pauli strings in the qubit Hamiltonian $H_\text{qubit}$ with an average weight of $\frac{1}{3}\left(4-\sqrt{2}\right) \approx 0.86$ times the corresponding average Pauli weight that the Z--pattern $f_\mathrm{Z}$ and the S--pattern $f_\mathrm{S}$ produce. 
\end{corollary}

\begin{proof}
	On an $N${}$\times${}$N$ lattice, the Mitchison--Durbin pattern $f_\mathrm{M}$ and the S--pattern $f_\mathrm{S}$ have edgesums
	\begin{align}
		C^1(f_\mathrm{Z}) &= N^3 - N \, ,\\
		C^1(f_\mathrm{S}) &= N^3-N \, ,\\ 
		\label{eqn:mitchsum}
		C^1(f_\mathrm{M}) &= N^3-xN^2+2x^2N-\frac{2}{3}x^3+N^2 -xN -2N+\frac{2}{3}x \, ,
	\end{align} 
	respectively. The edgesum $C^1(f_\mathrm{Z})$ is straightforward to calculate; see Appendix \ref{sec:s-pattern} for the derivation of $C^1(f_\mathrm{S})$, and Section \ref{sec:proofs} for the derivation of $C^1(f_\mathrm{M})$. In Equation \ref{eqn:mitchsum}, the value of $x$ is the closest integer to $N-\frac{1}{2}\sqrt{2 N^2-2 N+\frac{4}{3}}$. Using Equation \ref{eqn:avgpauli} and the number of hopping terms $|G_\text{F}|=2N(N-1)$, the average Pauli weights are
	\begin{align}
		\text{APV}(f_\mathrm{Z}) = \text{APV}(f_\mathrm{S}) &= \frac{1}{2}N+\frac{3}{2}
	\end{align}
	and, for large $N$,
	\begin{align}
		\text{APV}(f_\mathrm{M}) &\approx \frac{1}{6}(4-\sqrt{2}) N + \frac{1}{12}(20+\sqrt{2})  = 0.43N+1.78\, .
	\end{align}
	For small $N$, explicit calculation verifies that  $\text{APV}(f_{\mathrm{S}})>\text{APV}(f_\mathrm{M}) $ for $N \geq 6$.
	The ratio of average Pauli weights is thus $\text{APV}(f_\mathrm{M})/\text{APV}(f_\mathrm{S}) =\text{APV}(f_\mathrm{M})/\text{APV}(f_\mathrm{Z}) \approx\frac{1}{3}(4-\sqrt{2})\approx0.86$.
\end{proof}

We can conclude from Corollary \ref{cor:comparison} that, simply by labelling the fermionic modes using the Mitchison--Durbin pattern rather than the S--pattern as proposed in \cite{verstraete2005mapping} or Z--pattern, we can produce a qubit Hamiltonian with terms that are 13.9\% more local on average. Our proposal immediately translates to a reduction by the same amount in the number of single--qubit measurements required in a VQE protocol. Even in the case of simulating small fermionic systems, this method provides a worthwhile advantage. For example, in Figure \ref{fig:latticehop} where $N=6$, the ratio of the average Pauli weight using $f_\mathrm{M}$ versus $f_\mathrm{S}$ is $4.33/4.5 \approx 0.96$. That is, even for the $6${}$\times${}$6$ lattice, applying Theorem \ref{thm:physical} will reduce the number of Pauli measurements by 4\%. 

For $2\leq N \leq 100$, Figure \ref{fig:apv2} shows the resulting average Pauli weights of various enumeration schemes on $N${}$\times${}$N$ lattices. Even for small lattices, the Mitchison--Durbin pattern can yield a meaningful reduction: for $N=20$, $C^1(f_\mathrm{M})/C^1(f_\mathrm{S})=\frac{7140}{7980}\approx 89.5\%$.

\begin{figure}
	\centering
	\includegraphics[width=\linewidth]{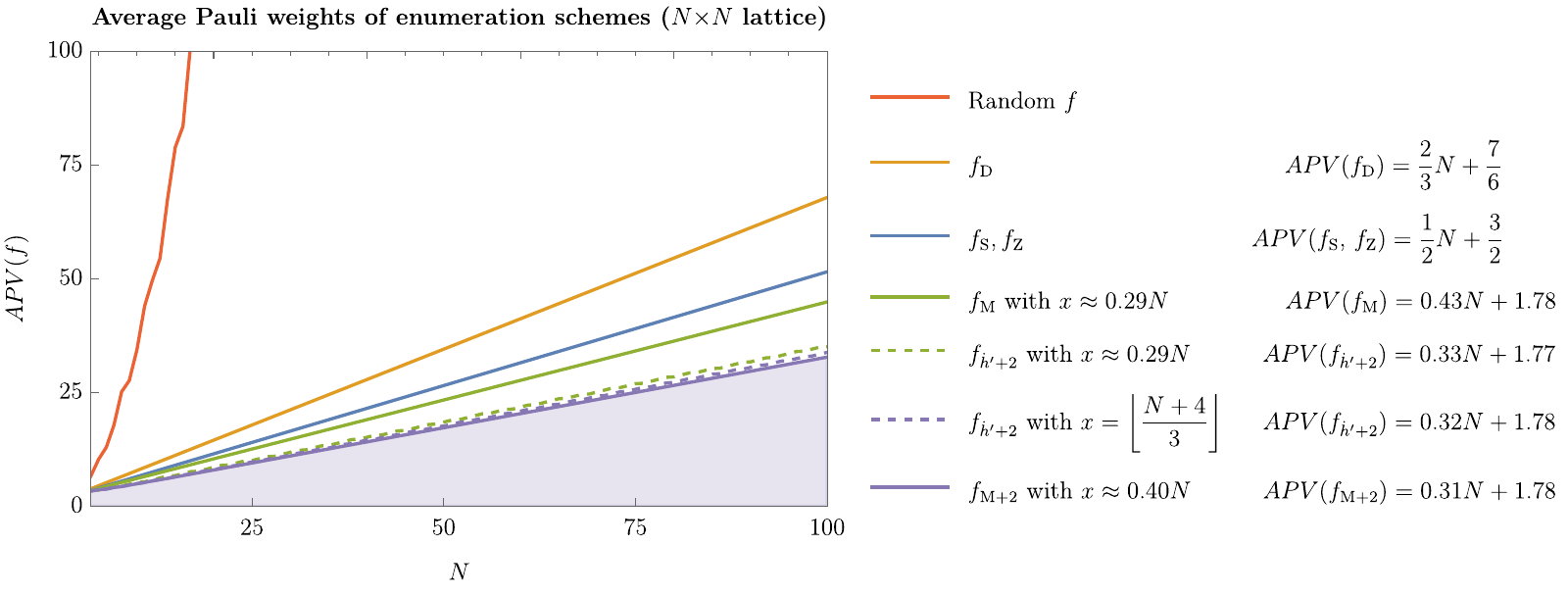}
	\caption{Average Pauli weights $\text{APV}(f)$ of enumeration schemes $f$
		on $N${}$\times${}$N$ lattices. Random $f$ is a randomly--generated enumeration scheme.}
	\label{fig:apv2}
\end{figure}

\subsection{Reducing the average Pauli weight for cellular fermionic lattices}\label{sec:cellular}

Theorem \ref{thm:physical} reduces the average Pauli weight of qubit Hamiltonians by making judicious use of the solutions to the edgesum problem. Remark \ref{cor:solved} refers to the scarce number of other graph families for $G_\text{F}$ for which edgesum solutions are known. It is tempting to think that there is not much use for our approach if the fermionic interaction graph $G_\mathrm{F}$ does not belong to one of these families. In this section, however, we show example fermionic Hamiltonians for which the enumeration schemes arising from approximate edgesum solutions can provide order--of--magnitude reductions in the average Pauli weight of the qubit Hamiltonians.

Consider fermionic Hamiltonians of the form in Equation \ref{eqn:problemham} with hopping terms between modes such that $G_\text{F}$ is an $(n${}$\times${}$n)${}$\times${}$(N${}$\times${}$N)$ \textit{cellular arrangement} of square lattices, where each $n${}$\times${}$n$ sub-lattice connects to adjacent sub-lattices via a single edge. Here we use the cellular pattern in Figure \ref{fig:cellular-patterns}, where the connections are from each $n${}$\times${}$n$ lattice's top left vertex to the two closest vertices from neighbouring $n${}$\times${}$n$ lattices.

It is possible to enumerate the fermionic modes with the Z--pattern $f_\mathrm{Z}$ or the S--pattern $f_\mathrm{S}$ from Section \ref{sec:mpattern}, treating the whole graph $G_\text{F}$ as a single square lattice. As Figure \ref{fig:cellular-patterns} shows, another way of enumerating the modes is to enumerate each sub-lattice locally before moving on to the next, progressing through the entire graph via an S-- or Z--pattern. We call these two enumeration procedures $f_{\mathrm{S}'}$ and $f_{\mathrm{Z}'}$, respectively.

The edgesums of these schemes on the cellular arrangement graph are:
\begin{align}
	C^1(f_\mathrm{S}) &= \begin{cases}
		(Nn)^3-Nn - N(N{-}1)(n{-}1) - \sum_{k=0}^{N{-}1} \sum_{i=kn+2}^{(k+1)n} (2i{-}1)\, , &  n \text{ even} \\
		(Nn)^3-Nn - N(N{-}1)(n{-}1) - \sum_{k=0}^{N{-}1} \sum_{i=kn+1}^{(k+1)n-1} (2i{-}1)\, , & n \text{ odd,}
	\end{cases}\\
	C^1(f_\mathrm{Z}) &= (Nn)^3-Nn - N(N-1)(n-1) - nN^2(n-1)(N-1)\,  , \\
	C^1(f_{\mathrm{Z}'}) &= \begin{cases}
		N^3n^2{+}n^3N^2  {-}n^2N^2  {-} 2n N^2 {-} n^2N {+} 2N^2  {+}2nN {+} n^2 {-}2N{-}n \, , & n \text{ even} \\
		N^3n^2 {+} n^3N^2 {-} n^2N^2 {-}nN^2 {-} n^2N {+} N^2 {+}2nN {+} n^2 {-}2N{-}2n{+}1 \, , & n \text{ odd,}
	\end{cases}
\end{align}
and while the value of $C^1(f_{\mathrm{S}'})$ also depends on the parities of $n$ and $N$, it is strictly greater than $C^1(f_\mathrm{Z'})$.

\begin{figure}
	\centering
	\includegraphics[width=\linewidth]{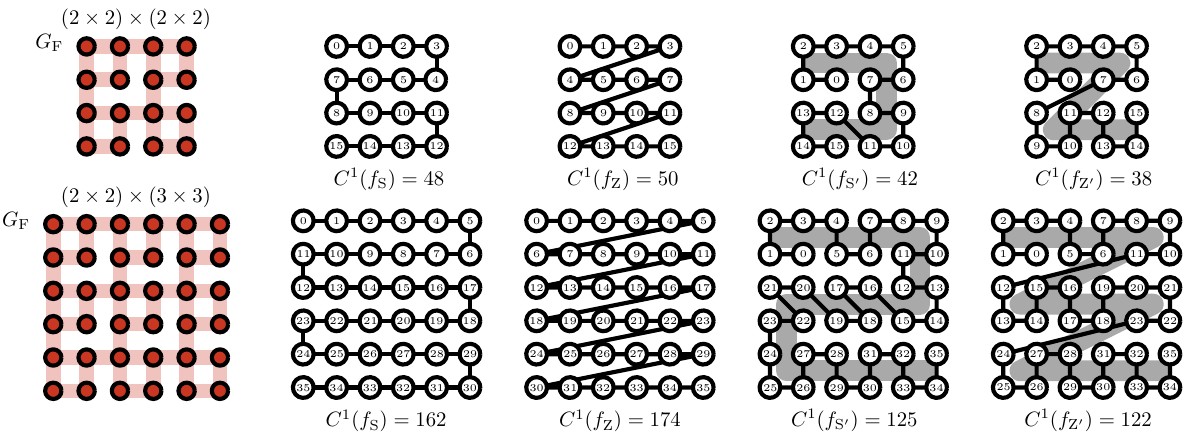}
	\caption{Edgesums for some enumeration schemes of $(n${}$\times${}$n) \times(N${}$\times N)$ cellular arrangement patterns.}
	\label{fig:cellular-patterns}
\end{figure}

By using $f_{\mathrm{Z}'}$ rather than $f_\mathrm{Z}$, the edgesum of a large cellular arrengement graph reduces by an order of magnitude:
\begin{align}
\lim_{N\rightarrow \infty} \frac{C^1(f_{\mathrm{Z}'})}{C^1(f_\mathrm{Z})} &= \frac{n^2}{n^3-n^2+n} = \mathcal{O}\left(\frac{1}{n}\right)\, ; \\
\lim_{n \rightarrow \infty} \frac{C^1(f_{\mathrm{Z}'})}{C^1(f_\mathrm{Z})} &= \frac{1}{N}\, .
\end{align}
If $n=\mathcal{O}(N)$, then the system has $\mathcal{O}(N^4)$ fermionic modes. From choosing $f_{\mathrm{Z}'}$ rather than $f_\mathrm{Z}$, the $\mathcal{O}(1/N)$ factor reduction in the average Pauli weight of the qubit Hamiltonian is thus proportional to the fourth root of the number of fermionic modes. This is much more striking than the constant--factor improvement proffered by Corollary \ref{cor:comparison}, even though there is no proof that $f_{\mathrm{Z}'}$ is an optimal enumeration scheme for this type of fermionic interaction graph $G_\text{F}$.

\subsection{Reducing the average \texorpdfstring{$p$}{p}th power of Pauli weight for a square lattice}\label{sec:pg1}
Section \ref{sec:enum1} introduced the average $p$th power of Pauli weight for $p>1$ as a cost function for the qubit Hamiltonian. For a fermionic Hamiltonian $H_\text{fermion}$ with interaction graph $G_\mathrm{F}$, consider the task of finding an enumeration scheme $f$ that minimises the $p$--sum
\begin{align}
C^p(f) = \left(\sum_{(u,v) \in G_\text{F}}\vert f(u) - f(v) \vert^p\right)^{1/p}\, .
\end{align}

\begin{figure}
\centering
\includegraphics[width=\linewidth]{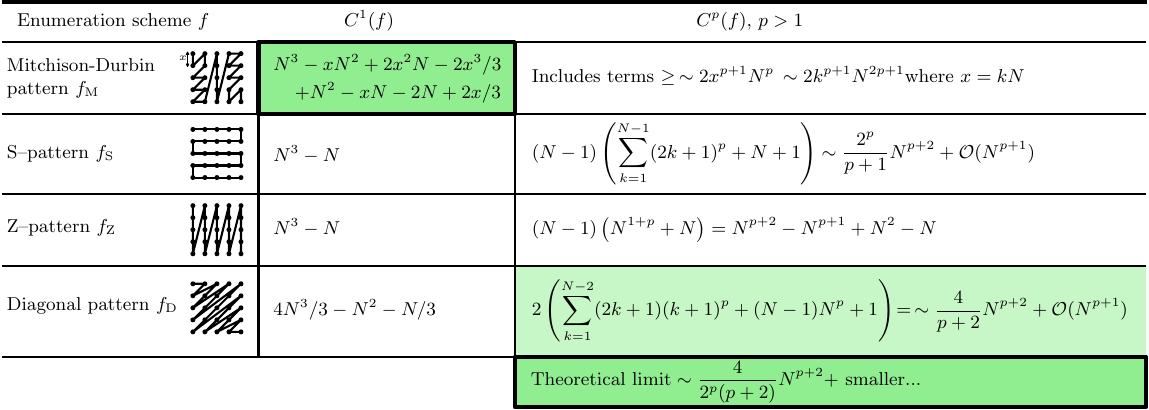}
\caption{Results for the minimum $p$--sum problem for the square lattice for $p \geq 1$. While the Mitchison--Durbin pattern $f_\mathrm{M}$ minimises $C^p(f)$ for $p=1$, it performs the worst of all the options for $p > 1$. The diagonal pattern $f_\mathrm{D}$ performs better than the rest for $p>2$.}
\label{fig:p2table}
\end{figure}

While the $p=1$ and $p=\infty$ cases correspond to metrics that seem physically motivated, we are not yet aware if current quantum algorithm designers would consider penalising finite powers $p>1$ of Pauli weight.
The $p=1$ case is the edgesum problem and via Remark \ref{cor:edgesumnp} is \NP--hard. George and Pothen showed that the minimum 2--sum problem is also \NP--hard  \cite{george1997spectralenvelope}. They develop spectral approaches to solving the minimum 1--sum and 2--sum problems for \textit{any} graph $G_\mathrm{F}$. Separately, Mitchison and Durbin \cite{mitchison1986optimal} explored the minimum $p$--sum problem for square lattices, and conclude with Proposition \ref{prop:psum}:

\begin{prop} \label{prop:psum}	\emph{(Mitchison and Durbin's results for the minimum $p$--sum problem on an $N${}$\times${}$N$ lattice \cite{mitchison1986optimal})} Let $G$ be the $N${}$\times${}$N$ square lattice. Then the following are the lower limits for the value of $C^p(f)$:
\begin{enumerate}[label=\emph{(\alph*)}]
	\item If $0 < p < \frac{1}{2}$, then there exists an enumeration scheme $f$ with $C^p(f) = \mathcal{O}(N^2)$.
	\item If $\frac{1}{2}<p<1$, then $C^p(f) \displaystyle \geq \frac{1}{2p+1}N^{1+2p}+\mathcal{O}\left(N^{2p}\right)$. \label{propb}
	\item If $p=1$, then $f_\mathrm{M}$ gives the minimum value as shown in Theorem \ref{thm:physical}.
	\item If $p>1$, then $C^p(f) \displaystyle \geq \frac{4}{2^p(p+2)}N^{p+2}$. \label{propd}
\end{enumerate}
The lower bounds in \ref{propb} and \ref{propd} are theoretical limits, and to date there are no known enumeration schemes that achieve their low coefficients. However, the ``diagonal pattern" $f_{\mathrm{D}}$ has edgesum
\begin{align}
	C(f_{\mathrm{D}}) \approx \frac{4}{p+2}N^{p+2} + \mathcal{O}(N^{p+1})\, ,
\end{align}
which has the same order of magnitude in $N$ as the theoretical limit in \ref{propd}.
\end{prop}
\begin{proof}
See Propositions 2--4 of \cite{mitchison1986optimal}.
\end{proof}


\begin{remark}\label{rem:psum}
\emph{(Best known results for the minimum $p$--sum problem on a square lattice, $p \geq 1$)}

Figure \ref{fig:p2table} provides a table of known results for four enumeration schemes $f_{\mathrm{S}}$, $f_{\mathrm{M}}$, $f_{\mathrm{Z}}$, and $f_{\mathrm{D}}$ (the diagonal pattern) in the regimes $p=1$ and $p>1$. The diagonal pattern is the solution to the bandwidth problem (from Section \ref{sec:math}) on rectangles \cite{mitchison1986optimal}, and is the best--performing pattern for the minimum $p$--sum problem for $p>1$, $N \rightarrow \infty$.
\end{remark}

\section{Improving beyond optimality with ancilla qubits} \label{sec:aqm}

The fermion--qubit mappings in Definitions \ref{defn:naivejw}--\ref{defn:opttype}, Examples \ref{exm:jw}--\ref{exm:tt}, and in Section \ref{sec:enumeration} are all mappings between $n$--mode fermionic systems and $n$--qubit systems. Theorem \ref{thm:physical}, for example, details the Jordan--Wigner mapping with minimum average Pauli weight when using $N^2$ qubits to simulate a fermionic system on $N^2$ modes. However, there is a storied interplay between the qubit count of a fermion--qubit mapping and the locality of its Hamiltonian. As both qubits and quantum gates are costly with current technology, the tradeoff begs the question: what is the maximum benefit a mapping could gain by using an extra one or two qubits?

In Section \ref{sec:aqmintro}, we describe Steudtner and Wehner's template for auxiliary qubit mappings \cite{steudtner2019quantum}. In Section \ref{sec:auxiliary}, we give an example of this process in the form of a new auxiliary qubit mapping. We augment the Mitchison--Durbin pattern to produce an auxiliary qubit mapping $\overline{\text{JW}}_{f_{\text{M}+2}}$ which achieves a 37.9\% reduction over the average Pauli weight of $\overline{\text{JW}}_{f_\text{S}}$, using only two extra qubits. Thus, taking advantage of the Mitchison--Durbin pattern we see that a constant number of additional qubits almost triples the advantage provided by Theorem \ref{thm:physical}, thus improving our earlier results `beyond optimality'.

\subsection{Background: auxiliary qubit mappings} \label{sec:aqmintro}

Moll et al.\ \cite{moll2016optimizing} show how to reduce the number of qubits required by a Jordan--Wigner transformation at the cost of introducing more terms into the molecular Hamiltonian. Auxiliary qubit mappings~\cite{verstraete2005mapping,steudtner2019quantum,phasecraft2020low,jiang2019majorana,whitfield2016} show the converse: that by introducing more qubits into a simulation of $H_\text{fermion}$, it is possible to vastly simplify the qubit Hamiltonian $H_\text{qubit}$. Verstraete and Cirac demonstrated a fermion--qubit mapping between a fermionic Hamiltonian $H_\text{fermion}$ with a square--lattice interaction graph $G_\text{F}$ \cite{verstraete2005mapping} and a qubit Hamiltonian $H_\text{qubit}$ on $2N^2$ qubits consisting of only local terms. That is, rather than producing long strings of Pauli terms of weight $\mathcal{O}(N)$ that wind around the qubit lattice (as in, for example, Figure \ref{fig:enumvsembed}), the auxiliary qubit mappings produce operations of weight $\mathcal{O}(1)$.

Steudtner and Wehner's formulation of auxiliary qubit mappings \cite{steudtner2019quantum} is thus: suppose that a fermion--qubit mapping produces the qubit Hamiltonian
\begin{align}
	H_\text{fermion} \longmapsto H_\text{q} \coloneqq H_{\text{qubit}} = \sum_{h \in \Lambda} c_h h\, ,
\end{align}
where $\Lambda \subseteq \{\mathds{1},X,Y,Z\}^{\otimes n} $, and the $c_h \in \mathds{C}$ ensure that $H_\text{q}$ is Hermitian. Denote by the subscript `dat' the $n$--qubit data system upon which the Hamiltonian $H_\text{q}$ acts. Given a state $\ket{\psi}_\text{dat}$ of the data system, the goal is to reduce the complexity of simulating the time--evolved state $e^{-iH_\text{q}t}\ket{\psi}_{\text{dat}}$. An auxiliary qubit mapping achieves this goal by instead mapping 
\begin{equation}
	H_\text{fermion} \longmapsto \widetilde{H}_\text{q}\, ,
\end{equation}
where $\widetilde{H}_\text{q}$ is a simplified qubit Hamiltonian that acts on both the data register and ancilla qubits in an auxiliary register, denoted `aux'. The following three steps detail how to construct Steudtner and Wehner's type of auxiliary qubit mapping.

\begin{enumerate}
	\item Choose a subset $\Lambda_{\text{non-loc}} \subseteq \Lambda$ of terms in the Hamiltonian $H_\text{q}$. These are the terms which will be shorter in $\widetilde{H}_\text{q}$. For each $h \in \Lambda_{\text{non-loc}}$, choose a Pauli string $p^h$ where $h=h_{\text{loc}} \, p^h$ such that the weight of $h_\text{loc}$ is less than that of $h$. In total, only use $r$ distinct $p$--strings. That is, several $h \in \Lambda_\text{non-loc}$ may contain the same $p^h$.
	
	Denote the set of all $p^h$ by $\{p_i\}_{i=1}^r$. That is, for any $h \in \Lambda_{\text{non-loc}}$, then $p^h = p_i$ for some $i \in \{1,\dots,r\}$.
	
	\item Introduce $r$ ancilla qubits in the auxiliary register, and devise a unitary mapping \begin{align} V:\ket{\psi}_{\text{dat}}\ket{0}^{\otimes r}_{\text{aux}} \longmapsto |\widetilde{\psi}\rangle_{\text{dat,aux}}\end{align} where $\ket{\smash{\widetilde{\psi}}}_{\text{dat,aux}}$ is such that for each $i \in \{1,\dots,r\}$, there exists $\sigma_i \in \{\mathds{1},X,Y,Z\}$ acting on the $i$th auxiliary qubit such that
	\begin{align}
		((p_i)_\text{dat} \otimes (\sigma_i)_\text{aux}) |\widetilde{\psi}\rangle_{\text{dat,aux}} = \ket{\smash{\widetilde{\psi}}}_{\text{dat,aux}}\, . \label{eqn:stabs}
	\end{align}
	The $p_i \otimes \sigma_i$ are the \textit{stabilisers} of the combined data and auxiliary system.
	
	\textbf{Requirement 1:}
	The choice of $\{p_i\}_{i=1}^r$ must allow for the existence of $\ket{\smash{\widetilde{\psi}}}_{\text{dat,aux}}$ obeying Equation \ref{eqn:stabs}. A necessary but insufficient condition is that $[p_i,p_j]=0$ for all $i,j \in \{1,\dots,r\}$.
	
	\item
	Modify each $h \in \Lambda$ via
	\begin{align}
		h_{\text{dat}} \longmapsto h_{\text{dat}} \otimes \kappa^h_{\text{aux}}\, ,
	\end{align}
	where $\kappa^h$ is a Pauli string acting on the auxiliary qubits such that
	\begin{align}
		[h_\text{dat} \otimes (\kappa^h)_\text{aux}, \, (p_i)_\text{dat} \otimes (\sigma_i)_\text{aux}]=0
	\end{align}
	for all $i \in \{1,\dots,r\}$.
	The rationale for this step is as follows: Suppose that $h=h_\text{loc} \, p^h$, where $p^h = p_i$. Then, we can multiply each $h = h_\text{loc} \, p^h \in \Lambda_{\text{non-loc}}$ by its corresponding stabiliser, thus producing a lower--weight Hamiltonian term, $\big((h_{\text{loc}})_\text{dat} \otimes (\sigma_i)_\text{aux}\big)$:
	\begin{align}
		 (h_\text{dat} \otimes \mathds{1}_{\text{aux}}) \ket{\smash{\widetilde{\psi}}}_{\text{dat,aux}}  & =        \big((h_{\text{loc}} \, p^h)_{\text{dat}} \otimes \mathds{1}_{\text{aux}}\big) \ket{\smash{\widetilde{\psi}}}_{\text{dat,aux}} \\
		& =   \big(\left(h_{\text{loc}} \, p_i\right)_{\text{dat}} \otimes \mathds{1}_{\text{aux}}\big) \ket{\smash{\widetilde{\psi}}}_{\text{dat,aux}}  \\
		& = \big(\left(h_{\text{loc}} \, p_i\right)_{\text{dat}} \otimes \mathds{1}_{\text{aux}}\big) \label{eqn:stabgen} \big( (p_i)_\text{dat} \otimes (\sigma_i)_{\text{aux}} \big) \ket{\smash{\widetilde{\psi}}}_{\text{dat,aux}} \\
		& = \big((h_{\text{loc}})_\text{dat} \otimes (\sigma_i)_\text{aux}\big) \ket{\smash{\widetilde{\psi}}}_{\text{dat,aux}}\, .
	\end{align}
	To replicate the time evolution $e^{-iH_{\text{q}}t} \ket{\psi}_{\text{dat}}$ of the system, it is necessary to be able to generate any stabiliser $((p_j)_\text{dat} \otimes (\sigma_j)_\text{aux})$ at the step in Equation \ref{eqn:stabgen} and commute it past \textit{any} Hamiltonian term $h_\text{dat} \otimes \mathds{1}_{\text{aux}}$ where $h$ can now be any element of $\Lambda$. If $[h_\text{dat} \otimes \mathds{1}_\text{aux}, \, (p_j)_\text{dat} \otimes (\sigma_j)_\text{aux}] =[h,p_j]\neq 0$, then application of the Hamiltonian term produces a state
	\begin{align} 
		|\widetilde{\psi}\,\smash{'}\rangle_{\text{dat aux}} = (h_\text{dat} \otimes \mathds{1}_{\text{aux}}) \ket{\smash{\widetilde{\psi}}}_{\text{dat,aux}}
	\end{align} that no longer satisfies Equation \ref{eqn:stabs} for all $i$, and thus it is impossible to generate and apply more stabilisers as required by Equation \ref{eqn:stabgen}.
	
	\textbf{Requirement 2:}
	Crucially, the $\kappa^h$ must be such that they preserve the time evolution of the data system. That is,
	\begin{align}
		&V^\dagger (h_\text{dat} \otimes (\kappa^h)_\text{aux}) \ket{\smash{\widetilde{\psi}}}_{\text{dat,aux}}  =   h\ket{\psi}_{\text{dat}} \otimes \ket{0}^{\otimes r}_{\text{aux}}\, .
	\end{align}
	Otherwise, the system would not retain any useful information.
	The requirement for the existence of $\kappa^h$ for $h \in \Lambda$ places a restriction on $V$.
\end{enumerate}

\begin{figure}
	\centering
	\includegraphics[width=0.8\linewidth]{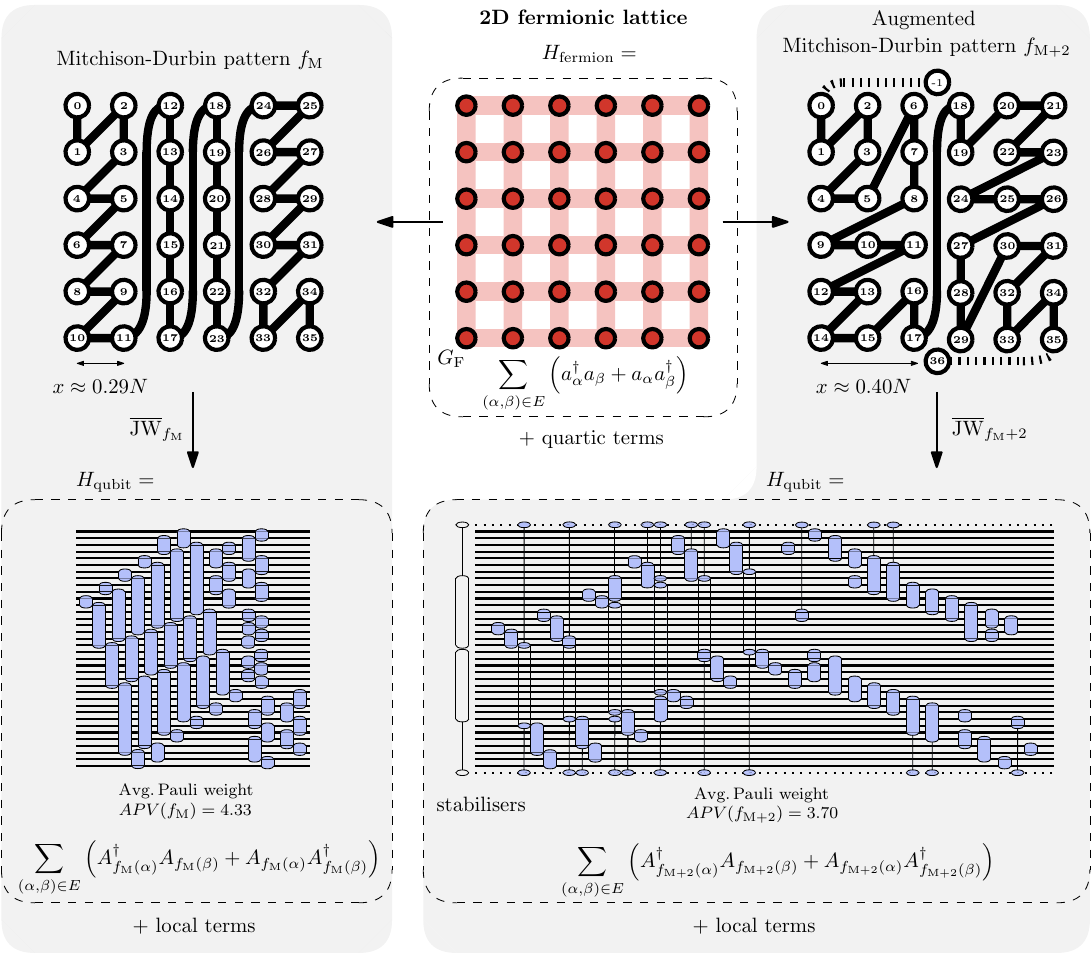}
	\caption{Comparison between the average Pauli weights of Jordan--Wigner mappings generated by the Mitchison--Durbin pattern $f_\mathrm{M}$ (left) and our auxiliary qubit pattern $f_{\mathrm{M}+2}$ (right). The bottom right image shows the reduction of the Pauli weight of long hopping terms after multiplication by stabilisers. Terms that partially overlap the stabilisers must also be adjusted.}
	\label{fig:6x6aug}
\end{figure}

After satisfying the above three steps, observe that
\begin{align} \label{eqn:shortcut}
	& \left(e^{-iH_{\text{q}}t} \ket{\psi}_{\text{dat}}\right) \otimes \ket{0}^{\otimes r}_\text{aux}  = \left(V^\dagger e^{-i\widetilde{H}_{\text{q}}t} \right)\ket{\smash{\widetilde{\psi}}}_\text{dat,aux}\, ,
\end{align}
where the effective qubit Hamiltonian $\widetilde{H}_{\text{q}}$ achieves the goal of reducing the Pauli weight of the strings in $\Lambda_\text{non-loc}$, and acts on both the data and auxiliary systems:
\begin{align}
	\widetilde{H}_{\text{q}} = &\sum_{h \in \Lambda_\text{non-loc}} c_h \big((h_\text{loc})_\text{dat} \otimes (\kappa^h \sigma^h)_\text{aux}\big)  \label{eqn:morelocal} + \sum_{h \in \Lambda \backslash \Lambda_\text{non-loc}} c_h (h_\text{dat} \otimes (\kappa^h)_\text{aux})\, . 
\end{align}
In Equation \ref{eqn:morelocal}, the Pauli matrix $\sigma^h$ is the matrix $\sigma_i$ from the stabiliser $p_i \otimes \sigma_i$ corresponding to $h = h_\text{loc} \, p^h$ via $p^h= p_i$. 
Assuming the cost of implementing $V$ is insignificant compared to cost of implementing time evolution by $\widetilde{H}_{\text{q}}$, then Equation \ref{eqn:shortcut} shows the advantage of the auxiliary qubit mapping: $V^\dagger e^{-i\widetilde{H}_{\text{q}}t}$ is less costly to implement than $e^{-iH_\text{q}t}$.

Note that this is not the most general form of auxiliary qubit mapping. An example of one that does not fit this template is the Verstraete--Cirac mapping, which uses stabilisers that are of the form $p_\text{dat} \otimes \tau_\text{aux}$ where $\tau_{\text{aux}}$ acts on more than one qubit \cite{verstraete2005mapping}.

\subsection{Using 2 ancilla qubits to reduce average Pauli weight by 37.9\% compared to the S--pattern}
\label{sec:auxiliary}

By modifying the proposal in Theorem \ref{thm:physical}, it is possible to use techniques inspired by the Mitchison--Durbin pattern to produce a `super--efficient' fermion--qubit mapping $\overline{\text{JW}}_{f_\text{M}+2}$ for problem Hamiltonians $H_\text{fermion}$, with interaction graphs $G_\text{F}$ equal to the $N${}$\times${}$N$ square lattice, using just $N^2+2$ qubits. This auxiliary Jordan--Wigner transformation has an average Pauli weight approximately 37.9\% less than that of the original $\overline{\text{JW}}_{f_\text{S}}$ presented in \cite{verstraete2005mapping}, and 27.9\% less than our result $\overline{\text{JW}}_{f_\text{M}}$ presented in Theorem \ref{thm:physical}, which was provably optimal for $n$--mode to $n$--qubit Jordan--Wigner transformations of the form in Definition \ref{defn:jwfenum}.

At this point, we direct our reader to Appendix \ref{sec:augmentedapp} for the full creation process of $f_{\mathrm{M}+2}$, and derivation of its average Pauli weight
\begin{align}
	\text{APV}\left(f_{\mathrm{M}+2}\right)  &\approx 0.31N+1.68\, .
\end{align}

Comparing our auxiliary qubit mapping to the S--pattern of \cite{verstraete2005mapping} and the Mitchison--Durbin pattern of Theorem \ref{thm:physical},
\begin{align}
	\lim_{N\rightarrow \infty} \frac{\text{APV}(f_{\mathrm{M}+2})}{\text{APV}(f_\text{S})} &\approx 0.62 \label{eqn:m+2s} \\
	\lim_{N\rightarrow \infty} \frac{\text{APV}(f_{\mathrm{M}+2})}{\text{APV}(f_\text{M})} &\approx 0.72\, , \label{eqn:m+2m}
\end{align}
yielding a 37.9\% improvement over the Z-- and S--patterns.

Figure \ref{fig:apv2} shows the average Pauli weights of fermion--qubit mappings $\overline{\text{JW}}_f$ using all the enumeration schemes $f$ for the $N${}$\times${}$N$ fermionic lattice discussed in this paper.

Figure \ref{fig:6x6aug} concludes by demonstrating the advantage of our auxiliary qubit mapping $\overline{\text{JW}}_{\text{M}+2}$ for simulating fermionic Hamiltonians $H_\text{fermion}$ with $G_\text{F}$ equal to the 6$\times$6 lattice. When $N=6$ and $x=3$, the qubit Hamiltonian $\widetilde{H}_\text{q}$ has a total Pauli weight of 222\footnote{This differs from direct substitution into Equation \ref{eqn:tpv} in Appendix \ref{sec:augmentedapp}, as the 6$\times$6 lattice has no region $D$ and some hopping terms benefit from multiplication by both stabilisers.}. Using $\overline{\text{JW}}_{\text{M}+2}$ produces a qubit Hamiltonian $\widetilde{H}_\text{q}$ with an average Pauli weight of $\frac{222}{60} = 3.70$, which is a 14.6\% improvement on the result $\frac{260}{60}=4.33$ of $\overline{\text{JW}}_{f_\text{M}} (H_\text{fermion})$, which is optimal for $n$--mode to $n$--qubit mappings. It is also a 17.8\% improvement on the average of $\frac{270}{60} = 4.5$ for $\overline{\text{JW}}_{f_\text{Z}} (H_\text{fermion})$ and $\overline{\text{JW}}_{f_\text{S}} (H_\text{fermion})$.

\section{Discussion}\label{sec:discussion}

\subsection{Summary of results}

By identifying the inherent place of fermion enumeration schemes in all $n$--fermionic mode to $n$--qubit mappings, we identified novel methods for creating new and improved mappings, demonstrating two main results.

The first (Section \ref{sec:enumeration}): when used as enumeration schemes for fermionic modes in Jordan--Wigner transformations of a broad class of $n$--mode problem fermionic Hamiltonians, solutions to well--known graph problems can minimise practical cost functions of the qubit Hamiltonians. Theorem \ref{thm:physical} illustrates the effect for fermionic Hamiltonians with hopping terms forming a square lattice of interactions, and identifies an $n$--qubit Jordan--Wigner transformation that uses the Mitchison--Durbing pattern to produce a qubit Hamiltonian with the minimum possible average Pauli weight. This is an improvement of 13.9\% compared to previous methods using the S--pattern. Even for fermionic interaction graphs without known edgesum solutions, classical heuristic techniques may provide enumeration schemes that drastically reduce the average Pauli weight, as in the case of cellular arrangements in Section \ref{sec:cellular}.

The second result combines our fermionic enumeration technique with the common strategy of employing auxiliary qubits to further preserve locality of interactions. In Section \ref{sec:aqm}, we provide an improved Jordan--Wigner transformation for fermionic Hamiltonians with square--lattice interactions. Unlike all other auxiliary qubit mappings in the literature, this mapping requires only two auxiliary qubits, regardless of the number of fermionic modes $n$, and thus can be compared with the results of Theorem \ref{thm:physical}, which it improves upon by reducing the average Pauli weight by 37.9\% compared to previous methods using the S--pattern.

Summarising our results, Table \ref{tab:mappings} compares Verstraete and Cirac's S--pattern to our optimal ancilla--free mapping using $f_\mathrm{M}$, and our low--ancilla mapping $f_{\mathrm{M}+2}$. Auxiliary qubit mappings with scalable numbers of ancilla qubits are included for comparison. 

\subsection{Nonlinear qubit architectures and qubit routing} \label{sec:routing}

Most results involving the Jordan--Wigner and Bravyi--Kitaev transformations rely on qubits with connectivity constraints. In particular, Verstraete and Cirac \cite{verstraete2005mapping} pursue a Jordan--Wigner type auxiliary qubit mapping to transform a local Hamiltonian on a fermionic lattice to a local Hamiltonian on a qubit lattice. Subsequent works have followed suit, usually working towards preserving geometric locality between the fermions and qubits, where the assumption is that the qubit architecture is the same as the fermionic graph \cite{phasecraft2020low,whitfield2016,steudtner2019quantum,jiang2019majorana}.

Fermionic enumeration schemes influence the distribution of quantum gates throughout nonlinear qubit architectures. To see this, consider decomposing a fermion--qubit mapping from a fermionic graph $G_\mathrm{F}$ to a qubit graph $G_\mathrm{Q}$ into two components: 1) the \textit{enumeration} of the fermionic modes, which projects the fermionic connectivity graph $G_\mathrm{F}$ onto a 1D array of qubits, and then 2) the \textit{embedding} of the 1D array of qubits into the qubit architecture $G_\mathrm{Q}$. Figure \ref{fig:embedding} visualises the distinction between the two processes.

Verstraete and Cirac's approach is to introduce the S--pattern to enumerate the fermions, and to use an S--pattern embedding to weave Pauli strings into the qubit lattice, as in the top row of Figure \ref{fig:enumvsembed}. In this work we emphasise that these are two distinct processes, and the choice of fermionic enumeration scheme can be done in an entirely separate way to the method of the embedding process. Indeed, as the second row of Figure \ref{fig:enumvsembed} shows, the path of the fermionic enumeration scheme need not follow the connectivity of the qubit architecture at all to provide a valid mapping.

As mentioned in Section \ref{sec:fmcost}, it is possible to conceive of metrics for qubit Hamiltonians that depend on qubit architectures. For example, the qubit routing problem concerns the distribution of Pauli strings on a nonlinear qubit connectivity $G_{\mathrm{Q}}$ \cite{routing2019}. The cost function for qubit routing should penalise strings of Pauli gates that spread through $G_{\mathrm{Q}}$ sparsely, and reward clustered Pauli strings. This poses a different problem to minimising the average Pauli weight of the qubit Hamiltonian, as the optimal enumeration scheme for the fermion--qubit mapping will need to cluster the Pauli strings of the qubit Hamiltonian while also minimising their weight. We expect that finding optimal enumeration schemes for cost functions that incorporate qubit architecture is a much more difficult family of problems than simply minimising a cost function of the Pauli weight as we have done here.

Sections \ref{sec:enumeration} and \ref{sec:aqm} only concerned Jordan--Wigner transformations of the form in Definition \ref{defn:jwfenum}. A more general class of mappings to search over would be the general Jordan--Wigner transformation in Example \ref{exm:jw}. The potential of using enumeration schemes as a method to define optimal Bravyi--Kitaev or ternary--tree mappings of the forms in Example \ref{exm:bk} and \ref{exm:tt} is completely unexplored and the subject of our upcoming work \cite{nextsteps}.

\begin{figure}
	\centering
	\includegraphics[width=0.8\linewidth]{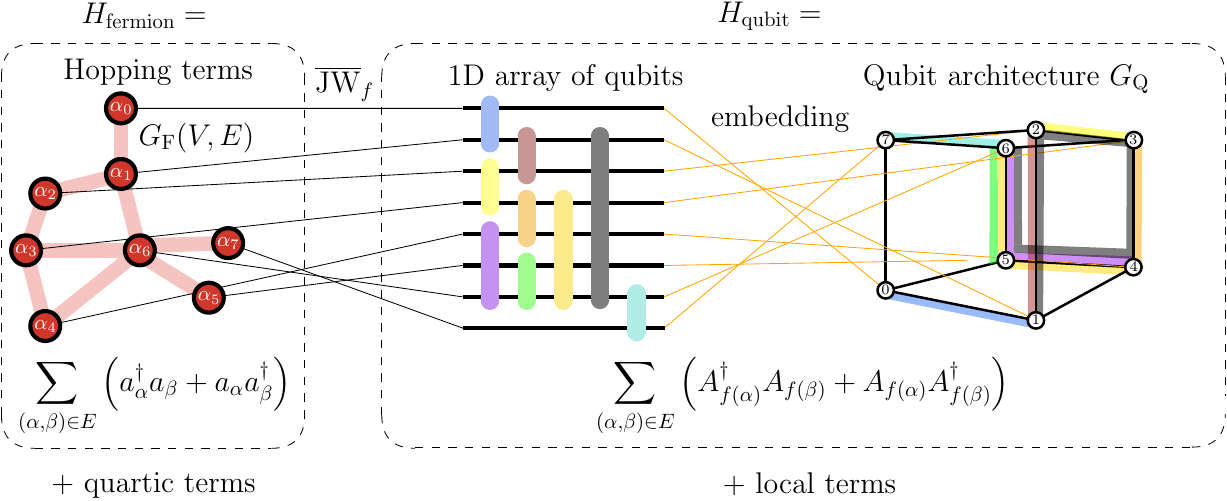}
	\caption{Fermionic enumeration techniques can influence the layout of Pauli strings on nonlinear qubit architectures.}
	\label{fig:embedding}
\end{figure}

\begin{figure}
	\centering
	\includegraphics[width=\linewidth]{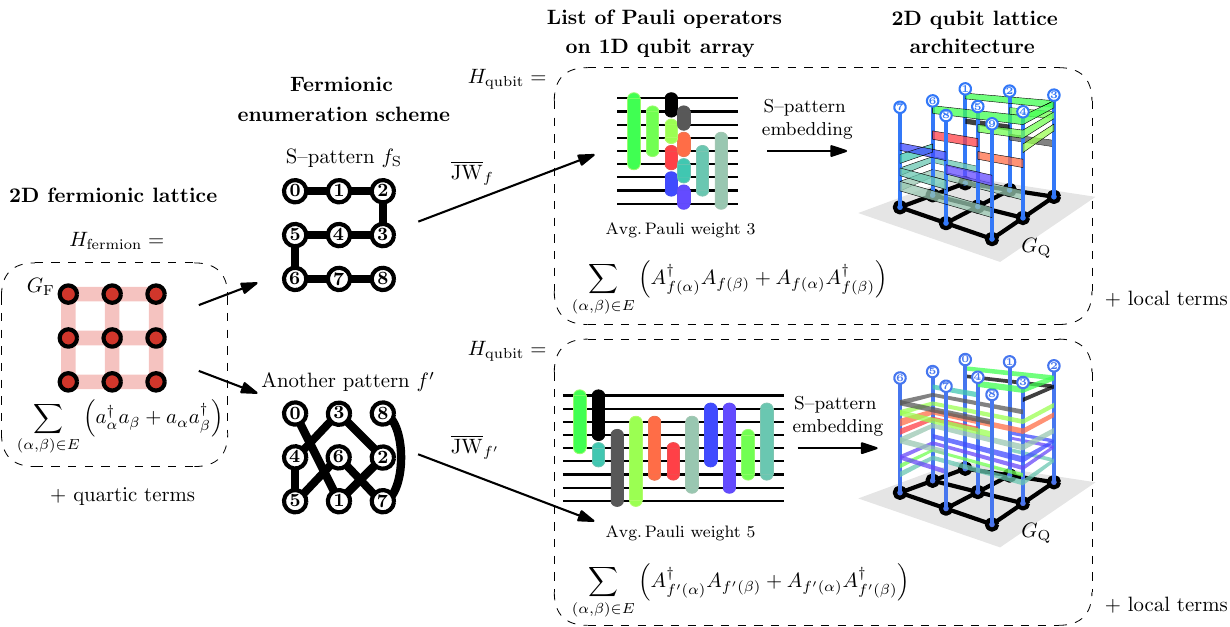}
	\caption{When working with elaborate qubit architectures such as the 2D lattice, the fermionic enumeration scheme need not follow the underlying qubit connectivity. This figure highlights the freedom to choose any scheme to enumerate fermions, yielding averages for the Pauli weights independent of the actual qubit connectivity.}
	\label{fig:enumvsembed}
\end{figure}

\renewcommand\arraystretch{2}
\begin{minipage}{\linewidth}
	\centering
	\captionof{table}{Fermion--qubit mappings for $H_\text{fermion}=N{\times}N$ square lattice} \label{tab:mappings} 
	\resizebox{\textwidth}{!}{
		\begin{tabular}{lcccccc}\toprule  & $\overline{\text{JW}}_{f_{\text{S}}}$ \cite{verstraete2005mapping} & $\overline{\text{JW}}_{f_\text{M}}$ (\ref{sec:enumeration}) & $\overline{\text{JW}}_{f_\text{M}+2}$ (\ref{sec:aqm}) & BK superfast \cite{bravyi2002fermionic} & VC \cite{verstraete2005mapping} & AQM \cite{steudtner2019quantum} \\ & (S--pattern) & (pattern from \cite{mitchison1986optimal}) & (+2 qubits) & & & \\ \hline
			Qubit number & $N^2$ & $N^2$ & $N^2+2$ & $2N^2-2N$ & $2N^2$ & $2N^2-N$ \\
			Qubit/mode ratio & 1 & 1 & $1+\frac{2}{N^2}$ & $2-\frac{2}{N}$ & 2 & $2-\frac{1}{N}$\\
			Avg.\ $X{...}X$  weight & $\frac{1}{2}N+\frac{3}{2}$ & $0.43N+1.78$ & $0.31N+1.78$ & $\mathcal{O}(1)$ &$\mathcal{O}(1)$ & $\mathcal{O}(1)$ \\ \hline
	\end{tabular}}
	\bigskip 
	\caption*{Comparison between fermion--qubit mappings based on the number of qubits required and the average Pauli weight of hopping terms  of the form $XZZ...ZX$ in the qubit Hamiltonian. In this case, the comparison is for fermionic Hamiltonians of $N^2$--mode systems where hopping terms form a square lattice of interactions. The mapping from Section \ref{sec:enumeration} gives the minimum average Pauli weight without the assistance of ancilla qubits, and the mapping from Section \ref{sec:aqm} improves this result while also maintaining a qubit/mode ratio of $1+\mathcal{O}\left({\frac{1}{N^2}}\right)$. In comparison, BK superfast, VC and AQM use twice as many qubits as ancilla--free mappings to reduce the weight of Hamiltonian terms, albeit to terms of weight $\mathcal{O}(1)$. Which mapping is best ultimately depends on factors such as how many qubits are available.}
\end{minipage}

{\bf Acknowledgements}
MC is funded by Cambridge Australia Allen and DAMTP Scholarship and the Royal Society PhD studentship. SS acknowledges support from the Royal Society University Research Fellowship scheme.

\newpage
\appendix
\section{Edgesum of the S--pattern}\label{sec:s-pattern}

\begin{lemma}\label{lem:ssum}
	On an $N${}$\times${}$N$ lattice, the S--pattern $f_\mathrm{S}$ has edgesum $C^1(f_\mathrm{S})=N^3-N$.
\end{lemma}
\begin{proof}
	Figure \ref{fig:s-patternnormal} displays the S--pattern enumeration procedure for an $N${}$
	\times${}$N$ square lattice, as well as the differences between adjacent vertices' indices.
	
	Regardless of whether $N$ is odd or even, the cost is thus
	\begin{align}
		C^1(f_\mathrm{S}) &= (N-1) \times \text{no. of rows} + \sum_{i=1}^{N} (2i-1) \times (\text{no. of rows} - 1)\\
		&=  N^2-N + (N-1) (N(N+1)-N) \\
		&=  N^2 - N - N^2 + N + N(N^2-1) \\
		&=  N^3-N\, . \qedhere
	\end{align}
\end{proof}

The average Pauli weight for the S--pattern on a square lattice is then
\begin{align}
	\text{APV}(f_\mathrm{S})&=\frac{C^1(f_{\mathrm{S}})}{|E|}+1 = \frac{N^3-N}{2N(N-1)}+1 = \frac{1}{2}N+\frac{3}{2} . \label{eqn:spatavg}
\end{align}

\begin{figure}[hbtp]
	\centering
	\includegraphics[width=0.8\linewidth]{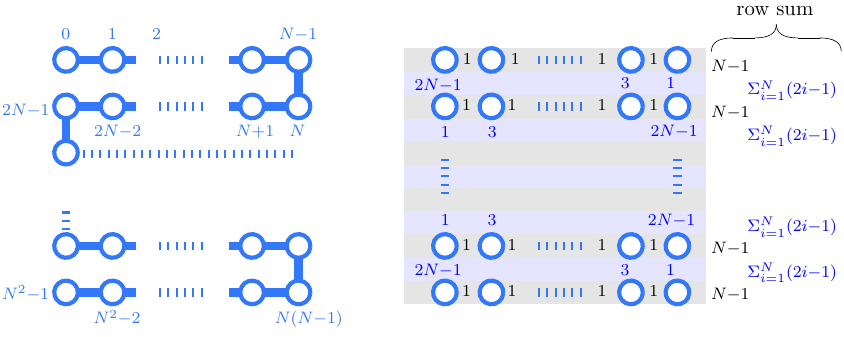}
	\caption{\textbf{Top:} S--pattern enumeration on the $N${}$\times${}$N$ square lattice. \textbf{Bottom:} Differences between vertex labels using the S--pattern enumeration, with row totals on the far right.}
	\label{fig:s-patternnormal}
\end{figure}

\section{Proof of Theorem \ref{thm:physical}} \label{sec:proofs}
This section contains a more elaborate version of the original proof \cite{mitchison1986optimal} that the Mitchison--Durbin pattern $f_\mathrm{M}$ solves the edgesum problem for an $N${}$\times{}${}$N$ lattice. Figure \ref{fig:proofsummary} chronicles the proof, which starts with an arbitrary enumeration scheme for the $N^2$ vertices before performing five sequential modifications (Proposition \ref{prop:order} and Lemmas \ref{lem:maxUV}--\ref{lem:mitchcost}) which exhaust all possible ways to reduce the edgesum, to give the Mitchison--Durbin pattern $f_\mathrm{M}$.

The edgesum problem is the minimum--$p$--sum problem described in Section \ref{sec:math} with $p=1$. The first part of the proof, Proposition \ref{prop:order}, holds for all $p \geq 1$; the rest of the proof concerns only the case $p=1$.

Given a graph $G=(V,E)$ and an enumeration scheme $f : V \rightarrow \{0,1,\dots,N^2-1\}$ for $G$, the cost function $C^p(f)$ for the minimum $p$--sum problem is
\begin{align}
	\left(C^p(f)\right)^p &= \sum_{(\alpha,\beta) \in E} \vert f(\alpha) - f(\beta) \vert^p \, . \label{eqn:psumcost}
\end{align}
In the case that $G$ is a square lattice, the edge set $E$ consists of pairs of horizontally and vertically adjacent vertices, leading to horizontal and vertical contributions:
\begin{align}
	\left(C^p(f)\right)^p \label{eqn:psum}  &= \hspace{-2em} \sum_{\substack{(\alpha, \beta) \in E; \\ \text{horizontally adjacent}}} \hspace{-2em} \vert f(\alpha) - f(\beta) \vert^p 
	+ \hspace{-2em} \sum_{\substack{(\alpha, \beta) \in E; \\ \text{vertically adjacent}}} \hspace{-2em} \vert f(\alpha) - f(\beta) \vert^p\, . 
\end{align}

\begin{figure}
	\centering
	\includegraphics[width=\linewidth]{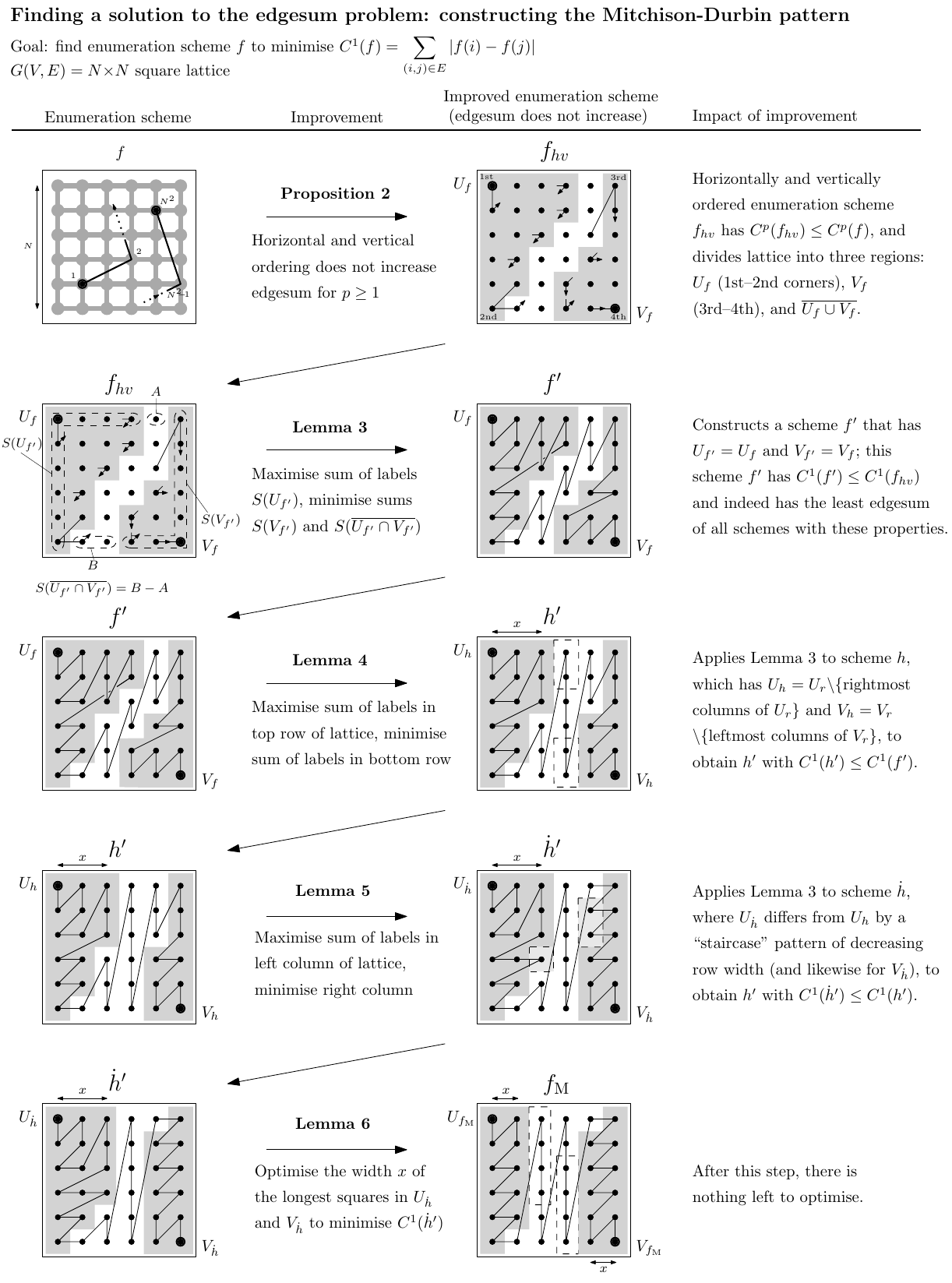}
	\caption{Overview of the proof that the Mitchison--Durbin pattern is a solution to the edgesum problem for the square lattice.}
	\label{fig:proofsummary}
\end{figure}

\begin{definition}
	Let $G=(V,E)$ be the $N${}$\times${}$N$ square lattice, and let $f$ be an enumeration scheme for $G$. Say $f$ is \textit{horizontally ordered} if, for any vertex $\alpha \in V$, $f(\beta)>f(\alpha)$ for all vertices $\beta$ to the right of $\alpha$, and say $f$ is \textit{vertically ordered} if $f(\beta)>f(\alpha)$ for all vertices $\beta$ below $\alpha$. 
\end{definition}

\begin{figure}
	\centering
	\includegraphics[width=0.6\linewidth]{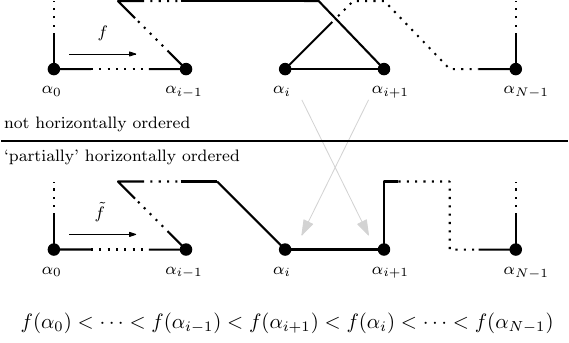}
	\caption{Swapping the labels of vertices in a row of the $N${}$\times${}$N$ square lattice as part of the process of constructing a horizontally ordered enumeration scheme.}
	\label{fig:inequalities}
\end{figure}

\begin{figure}
	\centering
	\includegraphics[width=0.9\linewidth]{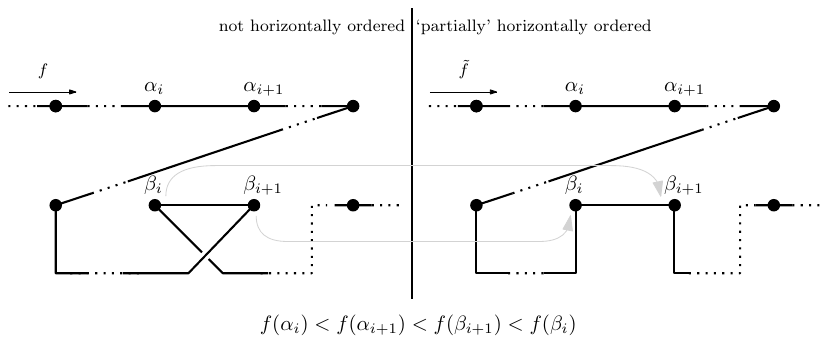}
	\caption{Swapping vertex labels within a row will affect some of the contributions to the $p$--sum of the enumeration scheme from vertically adjacent vertices.}
	\label{fig:propproof}
\end{figure}

\begin{figure}
	\centering
	\includegraphics[width=\linewidth]{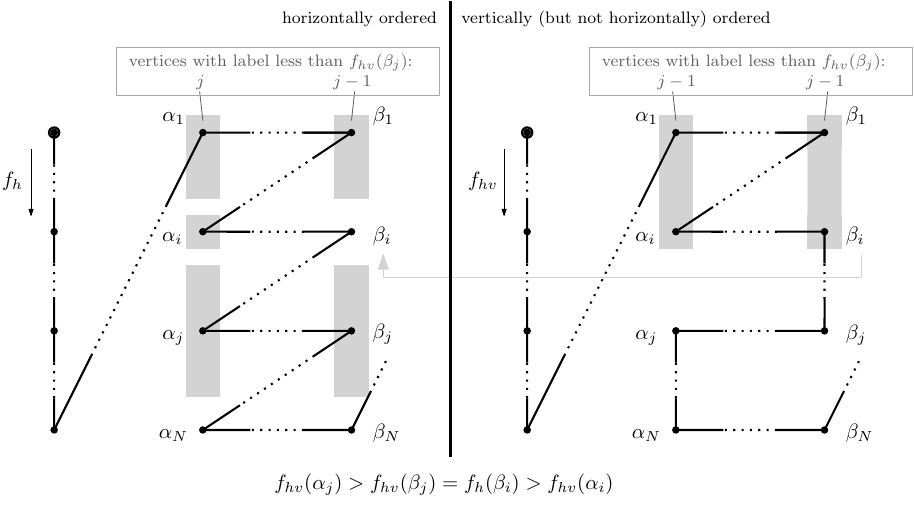}
	\caption{Illustration of the contradiction that would arise if a vertical ordering $f_{hv}$ of a horizontally ordered enumeration scheme $f_h$ were not itself horizontally ordered.}
	\label{fig:claimproof}
\end{figure}

\begin{prop}\label{prop:order}
	Let $G=(V,E)$ be the $N${}$\times${}$N$ square lattice, let $p \geq 1$, and let $f$ be an enumeration scheme $f$ for $G$. Then there exists a horizontally and vertically ordered enumeration scheme $g$ such that $C^p(g) \leq C^p(f)$.
\end{prop}

\begin{proof}
	From $f$, construct the enumeration scheme $f_h$ by permuting each row of the lattice so that vertex index values ascending from left to right. Call $f_h$ the \textit{horizontal ordering of $f$}. Similarly, construct $f_{hv}$ from $f_h$ by permuting vertex labels under $f_h$ within each column to increase from top to bottom. Call $f_{hv}$ the \textit{vertical ordering of $f_h$}. We will prove Proposition \ref{prop:order} via the statements S1 and S2 below:
	\begin{enumerate}[label={S\arabic*.}]
		\item $C^p(f_h) \leq C^p(f)$.
		\item \textit{The enumeration scheme $f_{hv}$ is horizontally ordered.}
	\end{enumerate}
	To prove S1, first consider the impact of horizontal ordering on the horizontal contributions to $C^p(f)$ in Equation \ref{eqn:psum}.
	
	\begin{claim}
		The horizontal contributions to $C^p(f_h)$ are less than or equal to the horizontal contributions to $C^p(f)$.

	\begin{claimproof}
		(See Figure \ref{fig:inequalities}.) Suppose $f$ is not horizontally ordered, and so there is a row of the lattice consisting of vertices $\alpha_0 , \alpha_1 , \dots, \alpha_{N-1}$ such that $f(\alpha_i) > f(\alpha_{i+1})$ for some $0 < i \leq N-1$. The horizontal contributions to $C^p(f)$ from this row are
		\begin{align}
			\label{eqn:contrib1} \vert f(\alpha_0) - f(\alpha_1)\vert^p  + \dots &+ \vert f(\alpha_{i-1}) - f(\alpha_{i})\vert^p \\ \nonumber &+ \vert f(\alpha_i) - f(\alpha_{i+1}) \vert^p \\ \nonumber  
			+ \dots &+ \vert f(\alpha_{N-2}) - f(\alpha_{N-1}) \vert^p\, . 
		\end{align}
		Consider the enumeration scheme $\tilde{f}$ that is identical to $f$ save that $\tilde{f}(\alpha_{i}) = f(\alpha_{i+1})$ and $\tilde{f}(\alpha_{i+1}) = f(\alpha_i)$, i.e. $\tilde{f}$ swaps the labels on vertices $\alpha_i$ and $\alpha_{i+1}$ so that $\tilde{f}(\alpha_i) < \tilde{f}(\alpha_{i+1})$. Thus, $\tilde{f}$ is at least `partially' horizontally ordered, moreso than $f$. The horizontal contributions to $C^p(\tilde{f})$ from this row are:
		\begin{align}
			\label{eqn:contrib2}\vert \tilde{f}(\alpha_0) - \tilde{f}(\alpha_1)\vert^p  + \dots &+ \vert \tilde{f}(\alpha_{i-1}) - \tilde{f}(\alpha_{i+1})\vert^p \\ \nonumber &+ \vert \tilde{f}(\alpha_{i+1}) - \tilde{f}(\alpha_{i}) \vert^p \\ \nonumber  
			+ \dots &+ \vert \tilde{f}(\alpha_{N-2}) - \tilde{f}(\alpha_{N-1}) \vert^p\, . 
		\end{align}
		The difference between the horizontal contributions to $C^p(f)$ and to $C^p(\tilde{f})$ is thus
		\begin{align}
			&\vert \tilde{f}(\alpha_{i+1}) - \tilde{f}(\alpha_i) \vert^p +  \vert \tilde{f}(\alpha_{i-1}) - \tilde{f}(\alpha_{i+1}) \vert^p \\ &\qquad  - \vert f(\alpha_{i}) - f(\alpha_{i+1}) \vert^p - \vert f(\alpha_{i-1}) - f(\alpha_i) \vert^p \nonumber \\
			&=  \vert \tilde{f}(\alpha_{i-1}) - \tilde{f}(\alpha_{i+1}) \vert^p - \vert f(\alpha_{i-1}) - f(\alpha_i) \vert^p \\
			&= \vert f(\alpha_{i-1}) - f(\alpha_i) \vert^p - \vert f(\alpha_{i-1}) - f(\alpha_i) \vert^p \\
			&<0\, , \nonumber
		\end{align}
		since $p \geq 1$ and $f(\alpha_{i-1}) < f(\alpha_{i+1}) < f(\alpha_i)$.
		
		To construct $f_h$ is to perform as many swaps of the above form (transforming $f$ into $\tilde{f}$) as is necessary in order to label all $N$ vertices in each row of the lattice in ascending order. Thus, the horizontal contributions to $C^p(f_h)$ are less than the horizontal contributions to $C^p(f)$, with equality if and only if $f=f_h$.
	\end{claimproof}
		\end{claim}
	
	To finish proving S1, consider the impact of horizontal ordering on the vertical contributions to $C^p(f)$ in Equation \ref{eqn:psum}.
	
	\begin{claim} The vertical contributions to $C^p(f_h)$ are less than or equal to the vertical contributions to $C^p(f)$.

	\begin{claimproof} (See Figure \ref{fig:propproof}.) Suppose that $f$ is not horizontally ordered, and that we are in the process of constructing $f_h$ from $f$ by re-ordering vertex labels row-by-row. Then there is a row of the lattice consisting of vertices $\beta_0, \beta_1, \dots, \beta_{N-1}$ such that $f(\beta_i)>f(\beta_{i+1})$ for some $i$. Let $\alpha_i$ and $\alpha_{i+1}$ be the vertices directly above vertices $\beta_i$ and $\beta_{i+1}$, respectively; thus $f(\alpha_i) < f(\alpha_{i+1})$.
		
		
		The vertical contributions to $C^p(f)$ from these four vertices is
		\begin{align}
			\vert f(\beta_i) - f(\alpha_i) \vert^p + \vert f(\beta_{i+1}) - f(\alpha_{i+1}) \vert^p\, . \label{eqn:contrib3}
		\end{align}
		Because $f(\alpha_i) < f(\alpha_{i+1})$ and $f(\beta_i) > f(\beta_{i+1})$, the values $|f(\alpha_i) - f(\beta_i)|^p$ and $|f(\alpha_{i+1}) - f(\beta_{i+1})|^p$ in Equation \ref{eqn:contrib3} are unequal, with the latter being less than the former.
		
		Consider the enumeration scheme $\tilde{f}$ that arises by swapping the labels of vertices $\beta_i$ and $\beta_{i+1}$, so $\tilde{f}(\beta_i) = f(\beta_{i+1})$ and $\tilde{f}(\beta_{i+1}) = f(\beta_i)$. The vertical contributions to $C^p(\tilde{f})$ from the four vertices is
		\begin{align}
			\vert f(\alpha_i) - f(\beta_{i+1}) \vert^p + \vert f(\alpha_{i+1}) - f(\beta_i) \vert ^p \, . \label{eqn:contrib4}
		\end{align}
		Note that for $p > 1$, if $a+b$ is a constant value then $a^p+b^p$ reduces in value the closer in value $a$ is to $b$. Therefore, the vertical contributions to $C^p(\tilde{f})$ in Equation \ref{eqn:contrib4} are less than the vertical contributions to $C^p(f)$ in Equation \ref{eqn:contrib3}. If $p=1$, the contributions are equal. This proves the claim, and furthermore S1. \end{claimproof}	\end{claim}
	
	To prove S2, consider the horizontally ordered enumeration scheme $f_h$, as in Figure \ref{fig:claimproof}.
	Suppose that $f_h$ is not vertically ordered, and let $f_{hv}$ be the vertical ordering of $f_h$.  It remains to show that $f_{hv}$ is horizontally ordered. Suppose that it is not horizontally ordered: thus, there exists $j \in \{0,\dots,N-1\}$ such that the $j$th row of the lattice contains vertices $\alpha_j$, $\beta_j$, with $\beta_j$ to the right of $\alpha_j$ and $f_{hv}(\alpha_j) > f_{hv} (\beta_j)$.
	
	Denote vertices in the column with $\beta_j$ by $\beta_0,\dots,\beta_{N-1}$, and similarly define the vertices  $\alpha_0 ,\dots, \alpha_{N-1}$ as those sharing the column of $\alpha_j$. Because $f_{hv}$ is vertically ordered, then $f_{hv}(\alpha_j) > f_{hv}(\alpha_k)$ and $f_{hv}(\beta_j) > f_{hv}(\beta_k)$ for all $0 \leq k \leq j-1$. As $f_{hv}(\alpha_j) > f_{hv}(\beta_j)$, then under $f_{hv}$, both the $j-1$ vertices above $\alpha_j$ and the $j-1$ vertices above $\beta_j$ have labels less than $f_{hv}(\beta_j)$. Since all that separates $f_h$ and $f_{hv}$ is a permutation of vertex labels that fixes the labels within columns of the lattice, then there should be exactly $j-1$ vertices with labels less than $f_{hv}(\beta_j)$ in both the $\alpha_0,\dots,\alpha_{N-1}$ and $\beta_0,\dots,\beta_{N-1}$ columns of the lattice under enumeration scheme $f_h$ as well.
	
	Because $f_h$ is horizontally ordered, $f_h(\alpha_j) < f_h(\beta_j)$. Therefore at least one of $\alpha_j$ and $\beta_j$ will have a different label under $f_{h}$ than under $f_{hv}$. Suppose, without loss of generality, that the vertex $\beta_j$ satisfies this: that $f_h(\beta_j) \neq f_{hv}(\beta_j)$ and thus there exists some $i \in \{0,\dots,N-1\}$ such that $f_h(\beta_i) = f_{hv}(\beta_j)$. Under $f_h$, there are $j-1$ vertices in the $\beta_0,\dots,\beta_{N-1}$ column with label less than $f_{hv}(\beta_j)$, as was the case under $f_{hv}$. However, consider the vertices in the $\alpha_0 , \dots , \alpha_{N-1}$ column. Each vertex $\alpha_k$ to the left of a vertex $\beta_k$ with $f_h(\beta_k) < f_{hv}(\beta_j)$ also satisfies $f_h(\alpha_k) < f_{hv}(\beta_j)$ due to horizontal ordering; there are $j-1$ of these vertices. Distinct from these vertices, there is also the vertex $\alpha_i$ to the left of $\beta_i$ which must also have $f_{h}(\alpha_i) < f_{h}(\beta_i) = f_{hv}(\beta_j)$. This is the contradiction.
	
	This proves S2 and hence Proposition \ref{prop:order}.
\end{proof}

Section \ref{sec:pg1} discusses the case $p >1$. Henceforth, this proof of Theorem \ref{thm:physical} only concerns the case $p=1$.

Let $G=(V,E)$ be the $N${}$\times${}$N$ square lattice. As a result of Proposition \ref{prop:order}, we need only consider horizontally and vertically ordered enumeration schemes in the search for the solution to the edgesum problem for $G$. Let $f_{hv}$ be the horizontal and vertical ordering of a vertex enumeration scheme $f$ for $G$. Let $\alpha_{i,j}$ denote the vertex of $G$ in the $i$th row and $j$th column, using matrix index notation so that the top--left vertex is $\alpha_{1,1}$. Then the edgesum of $f_{hv}$ is 
\begin{align}
	C^1(f_{hv}) &= \sum_{i,j=1}^{N-1} \big( f_{hv}(\alpha_{i+1,j}) - f_{hv}(\alpha_{i,j}) + f_{hv}(\alpha_{i,j+1}) - f_{hv}(\alpha_{i,j}) \big) \\
	&= 2f_{hv}(\alpha_{N,N}) + \sum_{i=2}^{N-1} f_{hv}(\alpha_{i,N})+ \sum_{j=2}^{N-1} f_{hv}(\alpha_{N,j}) - \sum_{i=2}^{N-1}f_{hv}(\alpha_{i,1})\\ \nonumber & \quad \quad \quad  - \sum_{j=2}^{N-1} f_{hv}(\alpha_{1,j}) - 2f_{hv}(\alpha_{1,1})\\
	&= \sum_{\alpha  \in V_\mathrm{b}} f_{hv}(\alpha) +\sum_{\alpha \in V_\mathrm{r}} f_{hv}(\alpha)   \label{eqn:lrtb}  - \bigg(\sum_{\alpha \in V_\mathrm{l}} f_{hv}(\alpha) + \sum_{\alpha \in V_\mathrm{t}} f_{hv}(\alpha) \bigg) \, ,  
\end{align}
where $V_\mathrm{l}$, $V_\mathrm{r}$, $V_\mathrm{t}$, $V_\mathrm{b}$ are the vertices in the left column, right column, top row and bottom row of the square lattice, respectively.
This result greatly simplifies the calculation of $C^1(f_{hv})$ for the square lattice.

Due to horizontal and vertical ordering, $f_{hv}(\alpha_{1,1})=0$ and $f_{hv}(\alpha_{N,N})=N^2-1$. Without loss of generality, assume that $f_{hv}(\alpha_{N,1})<f_{hv}(\alpha_{1,N})$. As detailed in Figure \ref{fig:proofsummary}, define $U_{f}$ to be the set of vertices with labels between $f_{hv}(\alpha_{1,1})=0$ and $f_{hv}(\alpha_{N,1})$, and let $V_{f}$ be the set of vertices with labels between $f_{hv}(\alpha_{N,1})$ and $f_{hv}(\alpha_{N,N})=N^2-1$. Define $S(U_{f})$ and $S(V_{f})$ to be the sums of the labels of vertices on the boundary of the lattice in the regions $U_f$ and $V_f$, respectively, and define $S(\overline{U_f \cap V_f})$ to be the difference between the sum of labels of vertices in $\overline{U_f \cap V_f}$ in the bottom and top rows of the lattice. By Equation \ref{eqn:lrtb}, these quantities completely determine the edgesum:
\begin{align}
	C^1(f_{hv}) = S(V_f) + S(\overline{U_f \cap V_f}) - S(U_f)\, .
\end{align}

The task now is to identify an improved enumeration scheme $f'$ with the least edgesum of all schemes $g$ that have $U_g=U_f$ and $V_g=V_f$. That is, find an enumeration scheme $f'$ such that
\begin{align}
	f' &=  \argmin_{\substack{\text{enumerations } g \\ U_g=U_{f},V_g=V_{f}}}  C^1(g) = \argmin_{\substack{\text{enumerations } g \\ U_g=U_{f},V_g=V_{f}}} \bigg(S(V_g) + S(\overline{U_g \cap V_g}) - S(U_g) \bigg)\, . \label{eqn:f}
\end{align}
Note that if $f$ is such that $U_f$ and $V_f$ are the same regions as those yielded by a solution to the edgesum problem, then any $f'$ satisfying Equation \ref{eqn:f} is a solution to the edgesum problem. Lemma \ref{lem:maxUV} details how to construct an $f'$ given $f$; using this result, Lemmas \ref{lem:bestallUV}--\ref{lem:mitchcost} detail how to set $f$ such that this $f'$ is a solution to the edgesum problem, completing the proof.

\begin{lemma}\label{lem:maxUV}
	Let $G=(V,E)$ be the $N${}$\times${}$N$ square lattice, with $\alpha_{i,j} \in V$ denoting the vertex in the $i$th row and $j$th column. Let $f_{hv}$ be the horizontal and vertical ordering of an enumeration scheme $f$ for $G$, dividing the lattice into regions $U_f$, $V_f$ and $\overline{U_f \cap V_f}$ as defined above. Then, the steps below construct a horizontally and vertically ordered enumeration scheme $f'$ that satisfies Equation \ref{eqn:f}. Note that while $U_{f'}=U_f$, $S(U_{f'})$ may not be equal to $S(U_f)$ and so on.
	
	\begin{enumerate}
		\item 
		To maximise $S(U_{f'})$:\\
		\textbf{Rule A:} starting at the beginning of a row $\alpha_{1,j}$ (resp.\ column $\alpha_{i,1}$), the enumeration scheme $f'$ proceeds rightwards along the row (resp.\ downwards along the column) into the interior of $U_f$ as far as possible without violating horizontal and vertical ordering. \\
		\textbf{Rule B:} once the enumeration scheme $f'$ has proceeded along a row or column as far as successive applications of Rule A permits, it must begin afresh at the start of the next-topmost row or the next-rightmost column, whichever is longest. \\ 
		Mathematically: suppose that Rule A terminates at the vertex $\alpha_{i,j}$. Let $i'>i$ and $j'>j$ denote the indices of the topmost row and rightmost columns, respectively, that are \textit{yet to be enumerated}, i.e.\ that consist entirely of vertices with labels greater than $f'(\alpha_{i,j})$. Then, let the number of vertices that are both in the $(i')$th row and also contained within $U_f$ be $n_{r}(i')$, and similarly let $n_{c}(j')$ denote the number of vertices in the $(j')$th column that are also contained within $U_f$. The following recursive rule describes the process of maximising $S(U_{f'})$ with Rule A and Rule B:
		\begin{align}
			f'(\alpha_{i,j}) + 1 = \label{eqn:rowcol} \begin{cases}
				f'(\alpha_{i,j+1}) & \text{if } \alpha_{i,j+1}\in U_f  \text{ and}  f'(\alpha_{i-1},j) < f'(\alpha_{i,j})\, , \\
				f'(\alpha_{i+1,j}) & \text{if } \alpha_{i-1,j} \in U_f \text{ and}  f'(\alpha_{i,j-1}) < f'(\alpha_{i,j})\, , \\
				f'(\alpha_{1,j'}) & \text{else if } n_c(j') \geq n_r(i')\, , \\
				f'(\alpha_{i',1}) & \text{else } n_r(i') > n_c(j')\, .
			\end{cases} 
		\end{align}
		\item To minimise $S(\overline{U_{f'} \cap V_{f'}})$: have the enumeration scheme $f'$ begin at the topmost vertex of each column in $\overline{U_f \cap V_f}$ and fill down to the lowest vertex in that column of $\overline{U_f \cap V_f}$.
		\item To minimise $S(V_{f'})$: perform the inverse procedure for maximising $S(U_{f'})$.
	\end{enumerate}
\end{lemma}

\begin{proof}
	Begin by noting that maximising $S(U_{f'})$, minimising $S(\overline{U_{f'} \cap V_{f'}})$ and minimising $S(V_{f'})$ are three independent tasks, thus the Lemma's sequential approach is conceptually valid. Next, with $f'(\alpha_{1,1}) = 0$, it must be the case that either $f'(\alpha_{1,2})=1$ or $f'(\alpha_{2,1}) = 1$. Without loss of generality, take $f'(\alpha_{1,2})=1$.
	
	\begin{enumerate}
		\item To maximise $S(U_{f'})$: Rule A and Rule B, as described in Equation \ref{eqn:rowcol}, work in tandem to label vertices the interior of $U_{f}$ with the lowest indices possible. As illustrated in Figure \ref{fig:uvlem2}, this ensures that the vertices on the boundary of the lattice, i.e.\ the vertices contributing to $S(U_{f'})$, have as high an index as possible. 
		\item To minimise $S(\overline{U_{f'} \cap V_{f'}})$, label the vertices in the columns of $\overline{U_f \cap V_f}$ in ascending order from top-to-bottom. This places labels with the least possible value in the bottom row of the lattice, and labels with the greatest possible value in the top row of the lattice, while preserving horizontal and vertical ordering.
		\item As in step 1, a symmetric argument applies to minimising $S(V_{f'})$.\qedhere
	\end{enumerate} 
\end{proof}

By the definition in Equation \ref{eqn:f}, there must exist a particular shape for regions $U_f$ and $V_f$ such that the $f'$ described in Lemma \ref{lem:maxUV} is a solution to the edgesum problem. There is an $x${}$\times${}$x$ square region of vertices in the top-left region of $U_f$, where $x$ is the greatest integer for which the statement ``the topmost $x$ rows of $U_{f}$ contain at least $x$ vertices" is true. We will call this region the \textit{largest square of $U_{f}$}; define the \textit{largest square of $V_f$} similarly.

\begin{figure}
	\centering
	\includegraphics[width=\linewidth]{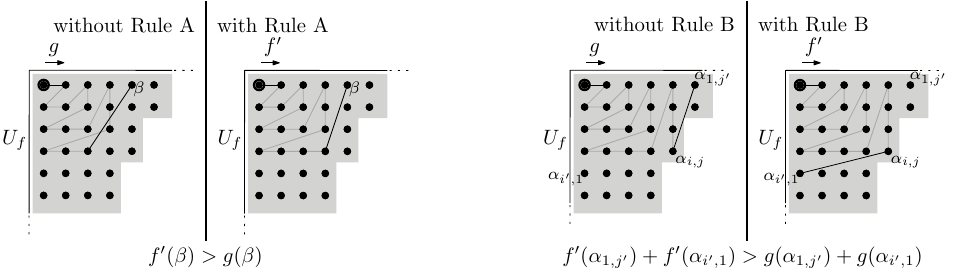}
	\caption{Rules in Lemma \ref{lem:maxUV} for constructing the enumeration scheme $f'$ so that it has the least edgesum of all enumeration schemes $g$ with $U_g=U_f$ and $V_g=V_f$.}
	\label{fig:uvlem2}
\end{figure}

\begin{figure}
	\centering
	\includegraphics[width=\linewidth]{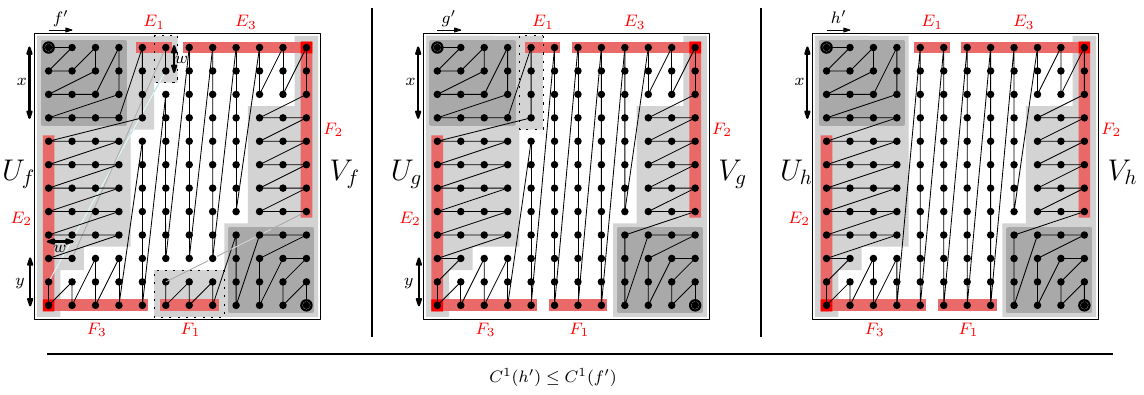}
	\caption{Illustration of Lemma \ref{lem:bestallUV}: starting with an enumeration scheme $f'$, intermediate scheme $g'$, and final scheme $h'$ demonstrating shapes of $U_h$ and $V_h$ that allow for a lower edgesum.}
	\label{fig:uvclaim2}
\end{figure}

\begin{lemma}\label{lem:bestallUV}
	Let $G=(V,E)$ be the $N${}$\times${}$N$ square lattice, and let $f$ be a horizontally and vertically ordered enumeration scheme for $G$, dividing the lattice into regions $U_f$, $V_f$ and $\overline{U_f \cap V_f}$ as defined above. Let $f'$ be the enumeration scheme resulting from applying Lemma \ref{lem:maxUV} to $f$.
	
	Now, consider an enumeration scheme $h$ that divides the lattice into regions $U_h$, $V_h$ and $\overline{U_h \cap V_h}$ where $U_h \subseteq U_f$ consists of all the vertices in $U_f$ except for those to the right of the largest square of $U_f$, and where $V_h \subseteq V_f$ consists of all the vertices in $V_f$ except for those to the left of the largest square of $V_f$. The vertices in $U_f\backslash U_h \cup V_f \backslash V_h$ are thus in $\overline{U_h \cap V_h}$. Let $h'$ be the scheme resulting from applying Lemma \ref{lem:maxUV} to $h$. Then, $C^1(h') \leq C^1(f')$.
\end{lemma}

\begin{proof}
	As in Figure \ref{fig:uvclaim2}, define $E_1$ to be the set containing the topmost vertices in $U_f$ that are to the right of the largest square in $U_f$, $E_2$ to be the set containing the leftmost vertices in $U_f$ beneath the largest square of $U_f$, and $E_3$ to be set containing the topmost vertices of the lattice that are outside of $U_f$. Define $F_1$, $F_2$ and $F_3$ analogously.
	
	While $E_1,\dots,F_3$ refer to fixed sets of vertices defined by the enumeration scheme $f$, let $E_1(g), \dots F_3(g)$ refer to the sums of the labels of those same vertex sets under any enumeration scheme $g$. Using the shorthand $(E_1+E_2)(g)=E_1(g)+E_2(g)$, we have via Equation \ref{eqn:lrtb},
	\begin{align}
		\label{eqn:efedgesum} C^1(f') &= \big((F_1+F_2+F_3) - (E_1+E_2+E_3)\big)(f')\\ \nonumber &\qquad  + \text{labels under $f'$ of boundary vertices}
		\, \text{of largest squares of $U_f$ and $V_f$\, ,}
	\end{align}
	where $f'$ is the enumeration scheme that arises from applying Lemma \ref{lem:maxUV} to $f$.

	\begin{claim}
		With the description of $h$ and $h'$ from  Lemma \ref{lem:bestallUV},
		\begin{equation}
			(E_1+E_2+E_3)(h') \geq (E_1+E_2+E_3)(f')\, .
		\end{equation}

	\begin{claimproof}
		Consider an enumeration scheme $g$ that has $U_g=U_f$ and $V_g=V_h$, and let $g'$ be the enumeration scheme that arises from applying Lemma \ref{lem:maxUV} to $g$. Then, $E_1(g')=E_1(f')$ and $E_2(g') = E_2(f')$ as the vertex labels are unchanged. However, $E_3(g')>E_3(f')$: this is because there are more vertices in $\overline{U_g \cap V_g}$ than in $\overline{U_f \cap V_f}$.
		
		Next, redefine $g$ such that $V_g = V_h$ and $U_g$ is equal to $U_f$ with its rightmost column removed, as in Figure \ref{fig:uvclaim2}. Define $g'$ to be the scheme arising from applying Lemma \ref{lem:maxUV} to $g$.
		
		Let $x$ be the side length of the largest square in $U_f$, let $w$ be the number of vertices in the rightmost column of $U_f$ and let $y$ be the number of rows in $U_f$ with no more than $w$ vertices.
		
		The labels of the bottom-most $y$ vertices in $E_2$ under $g'$ are each $w$ less than those same vertices under $f'$, and so $E_2(g') = E_2(f') - wy$. Meanwhile, $E_1(g')>E_1(f')$ because of the new label for the topmost vertex of column that was removed from $U_f$ to make $U_g$. The label of this vertex increases by at least $xy$, and so $(E_1+E_2)(g') \geq (E_1+E_2)(f')+xy-wy \geq (E_1+E_2)(f')$, as $xy-wy=(x-w)y \geq 0$ because $x\geq w$ (with equality iff.\ $U_g=U_f$).
		
		Repeating the process of deleting the rightmost column of $U_f$ will continue to increase $(E_1+E_2+E_3)(g')$. Note that $E_3$ does not change during the course of this procedure, which terminates when $g=h$, giving
		\begin{equation}
			(E_1+E_2+E_3)(h') \geq (E_1+E_2+E_3)(f')\, .
		\end{equation} 
	\end{claimproof}
		\end{claim}
	
	By the symmetric structure of enumeration patterns that arise from Lemma \ref{lem:maxUV}, this claim also demonstrates that $(F_1+F_2+F_3)(h') \leq (E_1+E_2+E_3)(f')$. Therefore,
	\begin{align}
		 \big((F_1+F_2+F_3)-(E_1+E_2+E_3)\big)(h') \leq \big((F_1+F_2+F_3)-(E_1+E_2+E_3)\big)(f') \, ,
	\end{align}
	which, via Equation \ref{eqn:efedgesum}, implies $C^1(h') \leq C^1(f')$.
\end{proof}

At this stage, the only part of $U_h$ left to scrutinise is its bottom-left corner: Lemma \ref{lem:mitchison} details a shape for it to take (and vice-versa for the top-right corner of $V_h$) in order to produce a new enumeration scheme $\dot{h}'$ with $C^1(\dot{h}') \leq C^1(h')$.

\begin{figure}
	\centering
	\includegraphics[width=\linewidth]{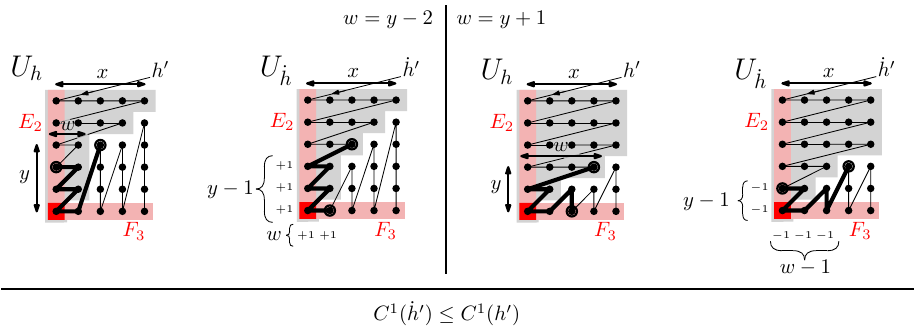}
	\caption{Finding an improvement $U_{\dot{h}}$ to the region $U_h$ in Lemma \ref{lem:mitchison}, and hence an improved enumeration scheme $\dot{h}'$ which has a lower edgesum than $h'$. As a visual aid, in both cases $w=y-2$ and $w=y+1$, the bolded vertices under scheme $h'$ share the same labels as the bolded vertices under $\dot{h}'$.}
	\label{fig:uvlem4i}
\end{figure}

\begin{lemma}\label{lem:mitchison}
	Let $G=(V,E)$ be the $N${}$\times${}$N$ square lattice, and let $h$ be a horizontally and vertically ordered enumeration scheme for $G$ arising from applying Lemma \ref{lem:bestallUV} to some enumeration scheme $f$. Let $h'$ be the enumeration scheme resulting from applying Lemma \ref{lem:bestallUV} to $h$.
	
	Let $x$ be the side length of the largest square in $U_h$. Then, consider a new enumeration scheme $\dot{h}$ such that $U_{\dot{h}}$ differs from $U_h$ by the following: modify $U_h$ by imposing a length of $x$ vertices on all rows down to a height $x$ above the bottom row, and then give the row at height $y<x$ a length of $y$ or $y-1$ for all $y=x-1,\dots,1$. Define $V_{\dot{h}}$ similarly.
	Apply Lemma \ref{lem:maxUV} to $\dot{h}$ to obtain $\dot{h}'$. Then, $C^1(\dot{h}') \leq C^1(h')$.
\end{lemma}

\begin{proof}
	Using Equation \ref{eqn:efedgesum},
	\begin{align}
		C^1(\dot{h}') -C^1({h}') &= \big((F_3-E_2) + (F_2-E_3)\big)(\dot{h}')  -  \big((F_3-E_2) + (F_2-E_3)\big)({h}') \\
		&= \underbrace{(F_3 - E_2)(\dot{h}') - (F_3-E_2)({h}')}_{\text{involves regions $U_h$ and $U_{\dot{h}}$ only}} \label{eqn:ubits}  + \underbrace{(F_2 - E_3)(\dot{h}') + (F_2 - E_3)({h}')}_{\text{involves regions $V_h$ and $V_{\dot{h}}$ only}}   
	\end{align}
	as the vertices in $E_1$ and $F_1$ have the same labels under $h'$ and $\dot{h}'$.
	
	First, consider the difference in edgesum due to the difference in shape between $U_{\dot{h}}$ and $U_h$, given in the first line of Equation \ref{eqn:ubits}. As in Figure \ref{fig:uvlem4i}, for $y=1,\dots,N-x$ take the $y$th row from the bottom of $U_h$, and let $w$ be the number of vertices it contains. The enumeration scheme $\dot{h}'$ only differs from $h'$ if $w<y-1$ or $w>y$.
	
	\textit{Case $w<y-1$}: Suppose $w=y-1-k$ for $k\geq1$ (Figure \ref{fig:uvlem4i} contains an example with $y=4$, $w=2$, $k=1$). Choose the number of vertices in the $y$th row from the bottom of $U_{\dot{h}}$ to be $y-1$. Then the labels of the $y-1$ bottom-most vertices in $E_2$ under $\dot{h}'$ are each greater by $k$ than the same vertices' labels under $h'$, i.e. $E_2(\dot{h}) - E_2(h') = k(y-1)$. Similarly, the labels of the $w$ leftmost vertices in $F_3$ under $\dot{h}'$ are each greater by $k$ than the same vertices' labels under $h'$, i.e. $F_3(\dot{h}') - F_3(h') = kw$. Thus, the top line of Equation \ref{eqn:ubits} is  equal to $k(w-(y-1)) < 0$.
	
	\textit{Case $w>y$}: Suppose $w=y+k$ for $k \geq 1$ (Figure \ref{fig:uvlem4i} contains an example with $y=3$, $w=4$, $k=1$). Choose the number of vertices in the $y$th row from the bottom of $U_{\dot{h}}$ to be $y$. Then the labels of the $y-1$ bottom-most vertices $E_2$ under $\dot{h}'$ each differ by $-k$ from the same vertices' labels under $h'$, i.e.\ $E_2(\dot{h}') - E_2(h') = -k(y-1)$. Similarly, the labels of the $w-1$ leftmost vertices in $F_3$ each differ by $-k$ from the same vertices' labels under $h'$, i.e.\ $F_3(\dot{h}') - F_3(h') = -k(w-1)$. Thus, the top line of Equation \ref{eqn:ubits} is equal to $-k((w-1)-(y-1))<0$.
	
	Note that, working within the successive restrictions of Lemmas \ref{lem:maxUV} and \ref{lem:bestallUV}, the maximum number of vertices permitted in any row of $U_{\dot{h}}$ is $x$. Thus, the rows increase in width from $y=1$ by one vertex at a time before reaching the maximum width of $y=x$ at $y=x$.
	
	By symmetry, the same rules work to construct the region $V_{\dot{h}}$ such that $C^1(\dot{h}') \leq C^1(h')$. \qedhere
\end{proof}

Finally, only one degree of freedom remains: the side lengths of the largest squares in $U_{\dot{h}}$ and $V_{\dot{h}}$. Note that these values only affect the top and bottom lines of Equation \ref{eqn:ubits}, respectively, and hence if there exists an optimal value $x$ for the width of $U_{\dot{h}}$, then it will also be the optimal value for the width of $V_{\dot{h}}$. Thus, we can take the side lengths of the largest squares in $U_{\dot{h}}$ and $V_{\dot{h}}$ to be the same.

Lemma \ref{lem:mitchcost} expresses the edgesum of the pattern as a function of $N$ and $x$. Corollary \ref{cor:optx} then provides the optimal value for $x$ and hence the Mitchison--Durbin pattern $f_\mathrm{M}$, a solution to the edgesum problem.

\begin{lemma}\label{lem:mitchcost}
	Let $G=(V,E)$ be the $N${}$\times${}$N$ square lattice, and let $\dot{h}$ be a horizontally and vertically ordered enumeration scheme that arises from applying Lemma \ref{lem:mitchison} to some horizontally and vertically ordered scheme $f$. 
	
	Let $x$ be the side length of the largest squares of $U_{\dot{h}}$ and $V_{\dot{h}}$, and let $\dot{h}'$ be the enumeration scheme that results from applying Lemma \ref{lem:maxUV} to $\dot{h}$. Then, the edgesum of $\dot{h}$ is
	\begin{align} \label{eqn:mitchcost2}
		C^1(\dot{h}') &= N^3-xN^2+2x^2N-\frac{2}{3}x^3+N^2-xN -2N+\frac{2}{3}x \, . 
	\end{align}
\end{lemma}
\begin{proof}
	See Section \ref{sec:fmcost}.
\end{proof}
\begin{corollary}\label{cor:optx}
	For $N\geq 5$, the value of $x$ that minimises $C^1(\dot{h'})$ in Equation \ref{eqn:mitchcost2} is $x=N-\frac{1}{2}\sqrt{2 N^2-2 N+\frac{4}{3}}$. Round $x$ to the nearest integer to obtain the minimum edgesum over all enumeration schemes for the $N${}$\times${}$N$ square lattice.
\end{corollary}
\begin{proof}
	Treat $N$ and $x$ as continuous variables, and use calculus to find that minimum of $C^1(\dot{h})$ in Equation \ref{eqn:mitchcost2} occurs when $x = N-\frac{1}{2} \sqrt{2N^2-2N+\frac{4}{3}}$.
\end{proof}

This completes the proof.		 The optimal value for $x$ in Corollary \ref{cor:optx} is a refinement on the value $x=\left(1-\frac{1}{\sqrt{2}}\right)N$ that Mitchison and Durbin give in \cite{mitchison1986optimal}. Rounding this approximate value gives the same integer as $N-\frac{1}{2} \sqrt{2N^2-2N+\frac{4}{3}}$ for most values of $N$. 

\begin{figure}
	\centering
	\includegraphics[width=0.55\linewidth]{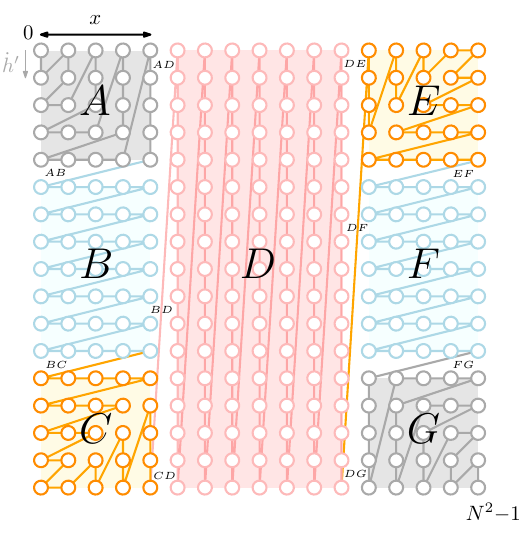}
	\caption{Sections of the $N${}$\times${}$N$ square lattice and an enumeration pattern $\dot{h}'$ that satisfies Lemma \ref{lem:mitchcost}.}
	\label{fig:mitchcount}
\end{figure}
\begin{figure}
	\centering
	\includegraphics[height=0.9\textheight]{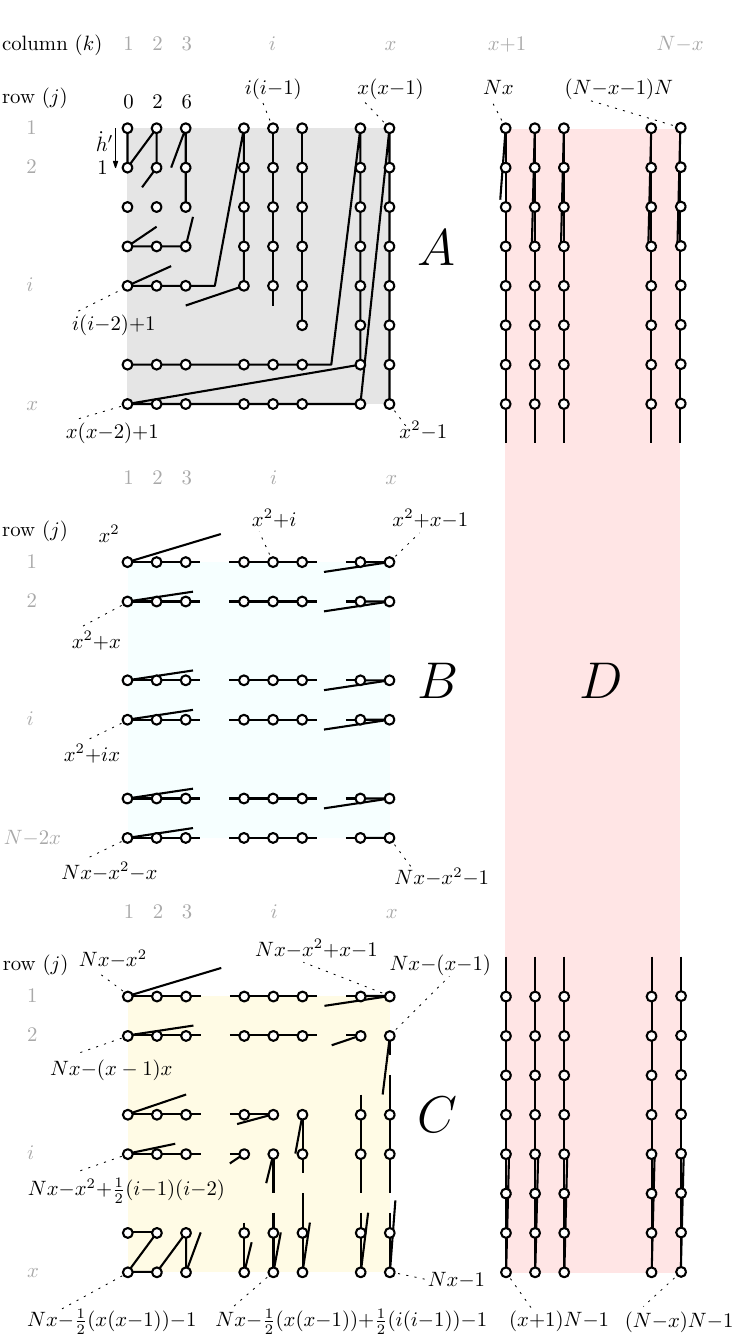}
	\caption{An in-depth look at the enumeration scheme $\dot{h}'$ that arises from Lemma \ref{lem:mitchcost}. Vertex labels indicate the value of each vertex under the enumeration scheme.}
	\label{fig:asum}
\end{figure}

\subsection{Edgesum of the Mitchison--Durbin pattern}\label{sec:fmcost}

This section calculates the edgesum of the enumeration pattern $\dot{h}'$ in Lemma \ref{lem:mitchcost}.
Divide the enumeration pattern on the $N${}$\times${}$N$ grid up into regions $A,B,\dots,G$ as in Figure \ref{fig:mitchcount}. Let the label of each region also denote its edgesum. Therefore
\begin{align} \label{eqn:regionsum}
	C^1(\dot{h}') &= A+B+C+D+E+F+G \\ \nonumber & \quad +AB + AD + BD + BC + CD \\ \nonumber & \quad + DE + EF + DF + FG + DG\, , 
\end{align}
where $AB,\dots,DG$ denote the sums of the differences between vertex labels across the interfaces between each pair of regions. Due to the symmetry of $\dot{h}'$, Equation \ref{eqn:regionsum} becomes
\begin{align}
	C^1(\dot{h}) &= \label{eqn:simplesum} 2(A+B+C)+D  + 2(AB+AD+BC+BD+CD)\, . 
\end{align}

We can derive expressions for the contributions of each region to $C^1(\dot{h}')$ by observing the patterns in the differences between vertex labels. Figure \ref{fig:asum} shows the progression of vertex labels in regions $A$--$D$ of the square lattice. Within each of these regions, the enumeration scheme is horizontally and vertically ordered. Thus, using Equation \ref{eqn:lrtb} on regions $A$ and $C$ gives:
\begin{align}
	A &=  \bigg(\sum_{k=1}^{x-1} \left( x(x-2) +  k \right) + x^2-1 + \sum_{j=1}^x \left(x(x-1)+j-1\right) \bigg) \\
	& \quad - \bigg( \sum_{k=1}^x \left(k(k-1)\right) + \sum_{j=1}^x \left(j(j-2) + 1 \right) \bigg) \nonumber \\
	&= \frac{1}{6}(x-1)(6+x(8x-1))  \label{eqn:first} \\
	&=C  \, .
\end{align}
The other regions' contributions are more straightforward to calculate:
\begin{align}
	B &= (N-2x)(x-1)+x^2(N-2x-1)\\
	D &= N^3 - 2xN^2 - 2xN + 2x - N \, .
\end{align}
The vertex labels in Figure \ref{fig:asum} also make it simple to calculate the contributions to $C^1(\dot{h}')$ from the interfaces between the regions:
\begin{align}
	AD &= \sum_{j=1}^x \bigg(\left( x(x-1)+j -1\right) - \left(Nx+j-1\right) \bigg)  =  Nx^2 - x^3 + x^2 \\
	AB &= 1-2x+2x^2 \\
	BD &= \frac{1}{2} \big(N^2 + N - 3xN + xN^2  + 2x^2 -2x -2x^2 N \big)  \\
	BC &= x^2 \\ 
	CD &= 1-2x+xN+x^2\, . \label{eqn:last}
\end{align}
Substituting Equations \ref{eqn:first}--\ref{eqn:last} into Equation \ref{eqn:simplesum} gives the result,
\begin{align}
	C^1(\dot{h}') &= N^3-xN^2+2x^2N-\frac{2}{3}x^3+N^2  -xN -2N+\frac{2}{3}x \, . 
\end{align}

\section{Constructing the enumeration pattern $f_{\mathrm{M}+2}$}\label{sec:augmentedapp}

Begin with a data register of $N^2$ qubits, and take a fermionic Hamiltonian $H_\text{fermion}$ of the form in Equation \ref{eqn:problemham}, where $G_\mathrm{F}=(V,E)$ is the $N${}$\times${}$N$ square lattice. As the quartic terms become local under any Jordan--Wigner transformation, consider only the hopping terms of the resulting qubit Hamiltonian
\begin{align}
	H_\text{q} &\coloneqq \overline{\text{JW}}_{\dot{h}'} \left( \sum_{(\alpha, \beta) \in E} a_\alpha^\dagger a_\beta + a_\alpha a_\beta^\dagger \right) =  \sum_{h \in\Lambda}c_h h \, ,
\end{align}
where $\Lambda$ is the set of terms of $H_{\text{q}}$, and $\dot{h}'$ is the Mitchison--Durbin pattern with a variable value for $x$, the side--length of the distinctive square in the corners of the pattern, as in Lemma \ref{lem:mitchison}. In $f_\text{M}$ the value $x\approx$\ $0.29 N$; however this only minimises $C^1(f_\mathrm{M})$ for Jordan--Wigner transformations with zero ancilla qubits. Since we are permitting ancillas, as in the proof of Theorem \ref{thm:physical}, treat $x \in \{1,2,\dots,N/2\}$ as a variable, its optimal value currently unknown.

Note that the edgesum $C^1(f_{\mathrm{M}})$ has two costly contributions: the $\mathcal{O}(x^2N)$--valued difference between adjacent vertices in the $AD$-- and $DG$--interfaces as labelled in Figure \ref{fig:mitchcount}. We will construct an auxiliary qubit mapping that reduces the contributions of these edges to $C^1(f_{\mathrm{M}})$, and hence reduces the average Pauli weight of the qubit Hamiltonian.

The fermion--qubit mapping we construct here uses only two ancilla qubits and takes a similar approach to the E--type auxiliary qubit mapping from \cite{steudtner2019quantum}. We will follow the formulation of auxiliary qubit mappings from Section \ref{sec:aqm}, making sure our mapping satisfies the two requirements that section lays out:

\begin{figure}
	\centering
	\includegraphics[width=\linewidth]{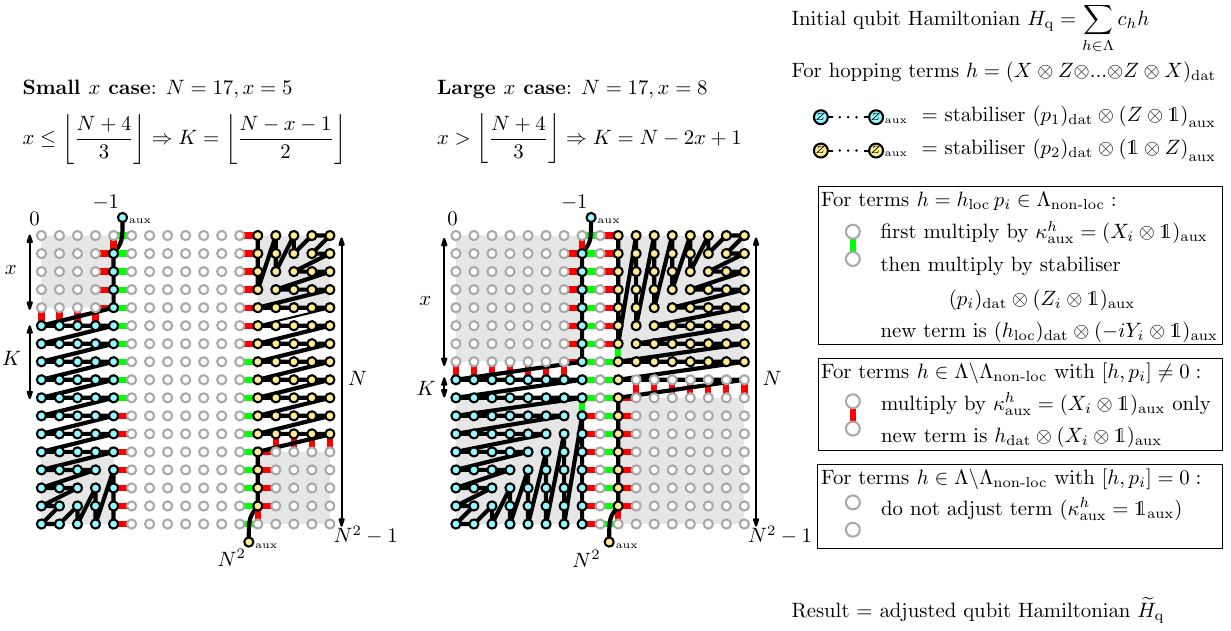}
	\caption{Stabilisers and adjusted Hamiltonian hopping terms for the auxiliary qubit mapping using our augmented Mitchison--Durbin pattern, `$f_{\dot{h}'+2}$'. Initialising the $(N^2${}$+${}$2)$--qubit system in a state stabilised by the strings $p_i \otimes \sigma_i$ nullifies the most costly contributions to the edgesum of the Mitchison--Durbin pattern, resulting in a significantly reduced average Pauli weight for the qubit Hamiltonian $\widetilde{H}_\text{q}$. As the text details, setting $x \approx 0.40N$ in the enumeration scheme $\dot{h}'$ yields the minimum average Pauli weight.}
	\label{fig:augmitchison}
\end{figure}

\begin{enumerate}
	\item Choose $p$--strings $\{p_1,p_2\}$ that correspond to the longest individual strings of Pauli $Z$ matrices of any term $h \in \Lambda$. In this case, we choose
	\begin{align}
		p_1 &=\left(\bigotimes_{i=x(x-1)+1}^{xN-1} Z_i \right) \otimes \left( \bigotimes_{\text{other } } \mathds{1}_i \right)
		\\ p_2 &= \left( \bigotimes_{i=(N-x)N}^{N^2-x^2+x-2} Z_i \right) \otimes \left(\bigotimes_{\text{other }i} \mathds{1}_i \right)\, .
	\end{align}
	These bridge the gaps between the outermost vertices on the $AD$-- and $DG$--interfaces, respectively. Figure \ref{fig:augmitchison} depicts the two $p$--strings for differing values of $x$; the indices of the first and last qubits of each $p$--string result from inspecting the label guide in Figure \ref{fig:asum}.
	
	Then, define $\Lambda_{\text{non-loc}}$ to be all of the Hamiltonian terms $h \in \Lambda$ that contain less $Z$ matrices upon multiplication by either $p_1$ or $p_2$. Note that none of the $h \in \Lambda$ would benefit from multiplication by \textit{both} $p$--strings, unless $x=N/2$ exactly as, in that case, there is no column of vertices separating the stabilisers. Figure \ref{fig:6x6aug} shows this in detail with $N=6$ and $x=3$. These $p$--strings trivially satisfy \textbf{Requirement 1}'s preliminary condition that $[p_1,p_2]=0$.
	
	\item Introduce two auxiliary qubits, labelled $-1$ and $N^2$. Define the unitary mapping $V$ to be a cascade of controlled--NOT operations, storing the net parity of the qubits appearing in $p_i$ in the phase of the $i$th auxiliary qubit.  Figure \ref{fig:Vcirc} shows a valid circuit for $V$.
	
	For any state $\ket{\psi}_\text{dat}$ of the original $N^2$--qubit system, now consider the state
	\begin{align}
		\ket{\smash{\widetilde{\psi}}}_\text{dat,aux} = V\ket{\psi}_\text{dat}\ket{0}^{\otimes 2}_\text{aux}\, .    \end{align}
	This state has two stabilisers:
	\begin{align}
		\big((p_1)_\text{dat} \otimes (Z \otimes \mathds{1})_\text{aux}\big) \ket{\smash{\widetilde{\psi}}}_\text{dat,aux}=\ket{\smash{\widetilde{\psi}}}_\text{dat,aux}\, ,\\
		\big((p_2)_\text{dat} \otimes (\mathds{1} \otimes Z)_\text{aux}\big) \ket{\smash{\widetilde{\psi}}}_\text{dat,aux}=\ket{\smash{\widetilde{\psi}}}_\text{dat,aux}\, .
	\end{align}
	The existence of the stabiliser state $\ket{\smash{\widetilde{\psi}}}_\text{dat,aux}$ satisfies \textbf{Requirement 1}.
	\item To identify which terms $h \in \Lambda$ do not commute with $p_1$ and $p_2$, recall the form of the hopping terms from Equation \ref{eqn:grp}. The only relevant feature of a term $h$ is the location of the endpoint matrices of its non--identity Pauli strings, i.e.\ its $X$ and $Y$ matrices. For brevity, and without loss of generality, abbreviate the terms to just the Pauli string $XZZ...ZX$:
	\begin{align}
		h_\text{dat}= \left( X \otimes Z \otimes \dots \otimes Z \otimes X \right) \otimes \left( \bigotimes_{\text{other}}\mathds{1} \right)\, .
	\end{align}
	As $p_1$ and $p_2$ comprise only $Z$ matrices, $[h,p_i]=0$ if and only if the Pauli string $XZZ...ZX$ of $h$ intersects completely, or not at all, with the Pauli string $ZZZ...Z$ of $p_i$, because $[X,Z]=-2iY \neq 0$. Thus the vast majority of the hopping terms already commute with both $p$--strings; it is only the hopping terms with Pauli strings that have one endpoint in $p_i$ and the other not in  $p_i$ that need to be modified so that they commute with both stabilisers. If  $[h,p_1] \neq 0$, make the adjustment
	\begin{align}
		h_\text{dat} &\longmapsto h_\text{dat} \otimes \left( X \otimes \mathds{1}\right)_\text{aux}\, ;
	\end{align}
	if  $[h,p_2] \neq 0$, make the adjustment
	\begin{align}
		h_\text{dat} &\longmapsto h_\text{dat} \otimes \left(\mathds{1} \otimes X \right)_\text{aux}\, .
	\end{align}
	If $x=N/2$ exactly, then some of the horizontal hopping terms across the $BD$-- and $DF$--interfaces will not commute with either $p_1$ or $p_2$, and need both adjustments (as depicted in Figure \ref{fig:6x6aug}).
	Abbreviating each correction by $\kappa^h$, adjusted Hamiltonian terms then satisfy
	\begin{align}
		[h_\text{dat} \otimes \kappa^h_{\text{aux}},(p_i)_\text{dat} \otimes (\sigma_i)_{\text{aux}}]=0
	\end{align}
	for $i=1,2$ because $[X_\text{dat} \otimes X_\text{aux}, Z_\text{dat} \otimes Z_\text{aux}]=0$. Figure \ref{fig:augmitchison} shows which hopping terms need adjustment in this way.

	\begin{figure}
		\centering
		\includegraphics[width=0.8\linewidth]{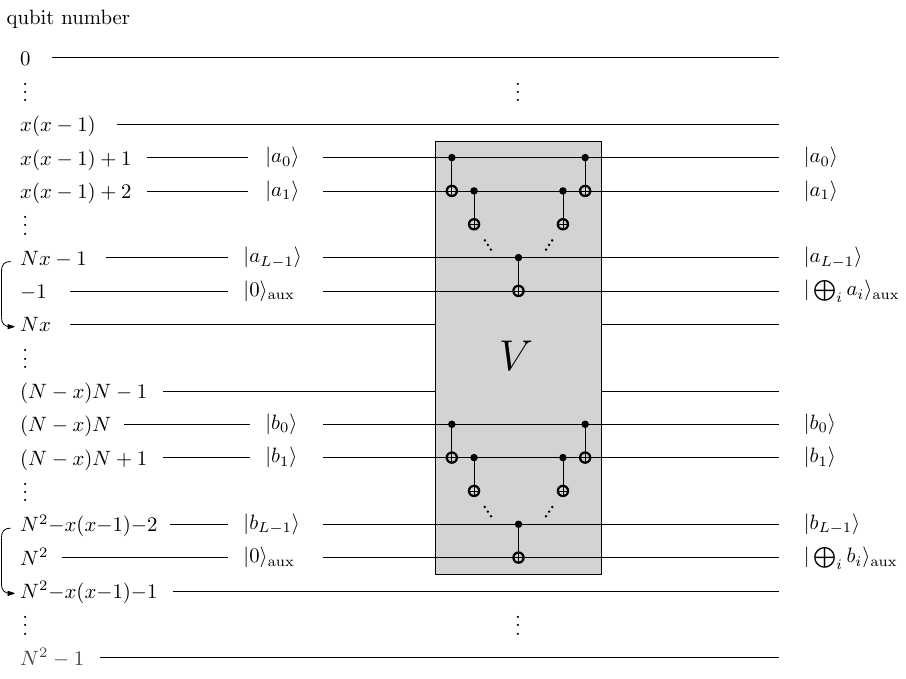}
		\caption{A circuit for $V$, which initialises the stabiliser state $\ket{\smash{\widetilde{\psi}}}_\text{dat,aux}$ from $\ket{\psi}_\text{dat} \ket{0}^{\otimes 2}_\text{aux}$ while preserving the time--evolution of the underlying data system.}
		\label{fig:Vcirc}
	\end{figure}
	
	Finally, check that $V$ and $\kappa^h$ satisfy \textbf{Requirement 2}. Let $h$ be a hopping term that does not commute with $p_1$.
	Consider an arbitrary state for the quantum register
	\begin{align}\ket{\psi}_\text{dat} = \sum_{\vec{a},\smash{\vec{b}}, \vec c} \gamma \ket{a_0 a_1  \dots ,\, b_0 b_1 \dots, \, c_0 c_1 \dots }_\text{dat}\, ,\end{align}
	where $a_j,b_j,c_j\in \{0,1\}$ such that $\vec{a}$ contains the parities of all the qubits involved in $p_1$, $\vec{b}$ contains the parities of the qubits not involved in either $p$--string, and $\vec{c}$ contains the parities of the qubits involved in $p_2$. The coefficients $\gamma \in \mathds{C}$ depend on the bit strings $\vec{a}$, $\vec{b}$ and $\vec{c}$.
	
	As $h$ does not commute with $p_1$, one of the $X$ matrices in its Pauli string must act on a qubit in $\vec{a}$ and the other $X$ matrix most act on a qubit in either $\vec{b}$ (if $x < N/2$) or $\vec{c}$ (if $x = N/2$ exactly). Consider the case $x<N/2$ and let the endpoint qubits of the Pauli string of $h$ be the $j$th and $k$th qubits in each string, respectively. Therefore, the output of $h\ket{\psi}_\text{dat}$ is
	\begin{align}
		h\ket{\psi}_{\text{dat}} =  \sum_{\vec{a},\smash{\vec{b}}, \vec c} \gamma' \ket{a_1  ... \, \overline{a_j} ... , \, b_1  ...  \, \overline{b_k} ..., \, c_1 c_2 ... }_\text{dat}\, ,
	\end{align}
	where each $\gamma'$ is equal to $\gamma$ up to a sign depending on the effect of the $Z$ matrices in $h$.
	Meanwhile, our prescription for $\kappa^h$ is to apply an $X$ on the first auxiliary qubit. Noting that $V=V^\dagger$, observe that
		\begin{align}
			V (h \otimes \kappa^h) \ket{\smash{\widetilde{\psi}}} &= V(h \otimes \kappa^h) \sum_{\vec{a},\smash{\vec{b}} , \vec c} \gamma \ket{a_0 a_1  ... ,\, b_0 b_1 ..., \, c_0 c_1 ... }_\text{dat} \otimes \left(\ket{\oplus_l \, a_l} \ket{\oplus_l \, c_l}\right)_\text{aux}\\
			&= V \sum_{\vec{a},\smash{\vec{b}}, \vec c} \gamma' \ket{a_0... \overline{a_j}..., \, b_0 ... \overline{b_k} ..., \, c_0c_1...}_\text{dat} \otimes \left(|\overline{\oplus_l \, a_l}\rangle \ket{\oplus_l \, c_l }\right)_\text{aux} \\
			&= V \sum_{\vec{a},\smash{\vec{b}}, \vec c} \gamma' \ket{a_0... \overline{a_j}..., \,  b_0 ... \overline{b_k} ... , \, c_0 c_1 ...}_\text{dat} \otimes \left(| a_0 \oplus ... \oplus \overline{a_j} \oplus ...\rangle \ket{\oplus_l \, c_l }\right)_\text{aux} \label{eqn:vunuse}\\
			&= \sum_{\vec{a},\smash{\vec{b}}, \vec c} \gamma' \ket{a_0... \overline{a_j}..., \,  b_0 ... \overline{b_k} ... , \, c_0c_1...}_\text{dat} \otimes \left( \ket{0} \ket{0} \right)_\text{aux} \label{eqn:vuse}\\
			&= h\ket{\psi}_\text{dat} \otimes \ket{0}^{\otimes 2}_\text{aux}\, ,
		\end{align}
	as required. Progressing from Equation \ref{eqn:vunuse} to Equation \ref{eqn:vuse} uses fact that $V$ is a cascade of controlled--NOT gates. A symmetric argument follows for all $h \in \Lambda$ which do not commute for $p_2$. For the case $x=N/2$, there will also be Hamiltonian terms that flip $a_j$ and $c_k$, and a similar argument applies.
\end{enumerate}

This establishes a modified qubit Hamiltonian $\widetilde{H}_\text{q}$ via Equation \ref{eqn:morelocal} and hence an auxiliary qubit mapping with two ancilla qubits, which we denote by $\overline{\text{JW}}_{\dot{h}'+2}$:
\begin{align}
	\overline{\text{JW}}_{\dot{h}'+2} : \big(H_\text{fermion} \text{ hopping terms}\big) \longmapsto \widetilde{H}_\text{q}\, .
\end{align}
By construction, this mapping follows the rules set out in Section \ref{sec:aqmintro}. The mapping has one degree of freedom: the value of $x$ in the enumeration scheme $\dot{h}'$, which relates the $n$ fermionic modes to the $n$ data qubits. 

As established in Section \ref{sec:enumeration} for $n$--mode to $n$--qubit mappings, there are many possible cost functions for the qubit Hamiltonian. In extending this concept to our new $n$--mode to $(n+2)$--qubit mapping $\overline{\text{JW}}_{\dot{h}'+2}$, we will determine the optimal value for $x$ in order to minimise the average Pauli weight of $\widetilde{H}_\text{q}$.

The \textit{total} Pauli weight of all terms in the modified Hamiltonian $\widetilde{H}_\text{q}$ is (1) the total Pauli weight of the original $H_\text{q}$, plus (2) the difference between the new and old $Z$--matrix counts of the terms in $\Lambda_{\text{non-loc}}$ after multiplication by stabilisers, plus (3) the cost of extra operations on the auxiliary qubits from making the $\kappa^h$ adjustments to terms in $\Lambda \backslash \Lambda_\text{non-loc}$. That is,
\begin{align}
	& \text{Total Pauli weight of $\widetilde{H}_\text{q}$} \label{eqn:words}\\ 
	& \qquad = \quad  \text{(1) Total Pauli weight of $H_\text{q}$} \nonumber \\
	& \qquad \quad + \text{(2) difference in $Z$--matrix count of $\widetilde{H}_\text{q}$ and $H_\text{q}$} \nonumber \\
	& \qquad \quad + \text{(3) number of adjusted terms in $\widetilde{H}_\text{q}$ that are not multiplied by stabilisers} \nonumber
\end{align}

Begin by only considering the effects of the first stabiliser, $(p_1)_\text{dat} \otimes \left(Z \otimes \mathds{1}\right)_\text{aux}$. For (1), recall from Equation \ref{eqn:avgpauli} that the total Pauli weight of the square lattice qubit Hamiltonian $H_\text{q}$ is $C^1(f_\mathrm{M}) + 2N(N-1)$.

For (2): First determine which hopping terms should be multiplied by stabilisers. In the original Hamiltonian $H_\text{q}$, a term $h \in \Lambda$ between vertices $\alpha$ and $\beta$ of $G_\mathrm{F}$ has Pauli weight $|f_\mathrm{M}(\alpha)-f_\mathrm{M}(\beta)|+1$; as hopping terms consist of a string of $Z$ matrices between a pair of $X$ or $Y$ matrices, the total number of $Z$ matrices in such a term is then 
\begin{align}
	\# (Z\text{ matrices in $h$}) = |f_\mathrm{M}(\alpha)-f_\mathrm{M}(\beta)|-1\, . \label{eqn:zcount}
\end{align}
Referring to the region labels $A,B,...,D$ in Figure \ref{fig:asum}, it is clear that all of the lengthy hopping terms between vertices on the $AD$--interface will have less $Z$ matrices after multiplication by the stabiliser $(p_1)_\text{dat} \otimes \left(Z \otimes \mathds{1}\right)_\text{aux}$. Using the labels from Figure \ref{fig:asum} in Equation \ref{eqn:zcount}, the original number of $Z$ matrices in \textit{each} of these terms is then
\begin{align}
	&\#(Z\text{ matrices in each row of $AD$, {before}}) = Nx-x(x-1)-1\, . 
\end{align}
By inspection, the topmost hopping term across the $AD$--interface after multiplication by $(p_1)_\text{dat} \otimes \left(Z \otimes \mathds{1}\right)_\text{aux}$ will consist of a single $Z$ matrix on the first auxiliary qubit. For the hopping terms across the $AD$--interface in the $j$th row beneath the top row, there will be $2j$ $Z$ matrices \textit{after} multiplication by $(p_1)_\text{dat} \otimes \left(Z \otimes \mathds{1}\right)_\text{aux}$, including the one on the auxiliary qubit. Thus, the correction term for the $AD$--interface is
\begin{align}
	&\Delta [\#(\text{$Z$ matrices in $AD$})] \\ \nonumber & \qquad \qquad \qquad = 1-\big(Nx-(x(x-1))\big)  + \sum_{j=1}^{x-1} \bigg( 2j-\big(Nx-(x(x-1))\big) \bigg)\, . 
\end{align}

Hopping terms across the $BD$--interface and the first row of the $CD$--interface may benefit from multiplication by $(p_1)_{\text{dat}} \otimes (Z \otimes \mathds{1})_{\text{aux}}$. Call the collection of these hopping terms the $(BC)D$--interface, and consider its $i$th row: once again, the labels from Figure \ref{fig:asum} in Equation \ref{eqn:zcount} yield
\begin{align}
	& \#(Z\text{ matrices in $i$th row of $(BC)D$, {before}})  = Nx +x -1 + i - (x^2 + ix - 1) - 1\, .
\end{align}
By inspection, these terms will have $(2+i)x+i-2$ $Z$ matrices \textit{after} multiplication by the stabiliser $(p_1)_\text{dat} \otimes \left(Z \otimes \mathds{1}\right)_\text{aux}$, including the $Z$ matrix on the auxiliary qubit.

As Figure \ref{fig:augmitchison} shows, if $x$ is large enough then it is possible that a small number of vertical hopping terms within the region $C$ will benefit from stabiliser multiplication as well. Although this reduces the average Pauli weight of the qubit Hamiltonian even further, its effect is negligible compared to the savings across the $(BC)D$--interface, and we omit it for simplicity. The correction term for the $(BC)D$--interface is thus
\begin{align}
	&\Delta [\#(\text{$Z$ matrices in $(BC)D$})] \label{eqn:bcdinter} \\ \nonumber & \qquad \qquad \qquad  = \sum_{i=1}^{K} \bigg( (2+i)x+i-2  - \big(Nx + x -1+ i - (x^2+ix-1)-1\big) \bigg)\, . 
\end{align}
The value $K$ depends on the size of $x$, as Figure \ref{fig:augmitchison} illustrates. From a naïve inspection of Equation \ref{eqn:bcdinter}, the $i$th row of the $(BC)D$--interface hopping terms will observe a reduction in $Z$--matrix count after stabiliser multiplication as long as
\begin{align}
	(2+i)x+i-2 &< Nx+x+i-(x^2+ix+1)\\
	i &< \frac{1}{2x} + \frac{N-x-1}{2}\\
	&\leq \frac{N-x-1}{2}\, .
\end{align}
This might tempt us to set $K=\left\lfloor (N-x-1)/2\right\rfloor$. However, there are $N-2x+1$ rows in the $(BC)D$--interface, so the expression in Equation \ref{eqn:bcdinter} is only valid if
\begin{align}
	K = \min\left\{\left\lfloor \frac{N-x-1}{2} \right \rfloor,\, N-2x+1\right\}\, .
\end{align}
These two values are equal if $x=\lfloor(N+4)/3\rfloor$. Figure \ref{fig:augmitchison} shows the two scenarios where $x$ is above or below this threshold value.

For (3), note that we must adjust all of the hopping terms $h$ that do not commute with $p_1$ by multiplication by $\kappa^h$, which is a single $X$ matrix targeting the first auxiliary qubit. However, in (2) we have already accounted for the cost of an extra auxiliary $Z$ matrix on each of the terms in the $AD$-- and $BD$--interfaces that benefited from stabiliser multiplication. This $Z$ matrix will become $(-iY)$ upon prior adjustment by $\kappa^h$ (premultiplication by $X$). We therefore need only count hopping terms that partially overlap with $p_1$: by inspection of Figure \ref{fig:augmitchison}, the total number of such terms is
\begin{align}
	& \#(\text{$X$ matrices from correction to terms}) = 2x-1+(N-x-K-1) \\ &\quad =N+x-K-1\, . \nonumber
\end{align}

As the Mitchison--Durbin pattern is symmetric, the same rules apply when considering the effects of the second stabiliser, $(p_2)_\text{dat} \otimes \left(\mathds{1} \otimes Z \right)_\text{aux}$. We can now finally substitute these values into Equation \ref{eqn:words} and express the total Pauli weight of $\widetilde{H}_\text{q}$. Note that the factor of 2 in Equation \ref{eqn:tpv} takes into account both stabilisers:

\begin{align}
	\text{Total Pauli weight of } \widetilde{H}_\text{q}  \label{eqn:tpv} &= C^1(f_\mathrm{M}) + 2N(N-1)  \\ &  \quad + 2 \bigg( \left(1-(Nx -x(x-1)-1\right)  +\sum_{j=1}^{x-1} \big(2j-(Nx-x(x-1))\big) \nonumber\\ &   \qquad  +\sum_{i=1}^{K} \left( \left( (2+i)x+i-2\right)   - \left(Nx + x + i - (x^2+ix+1) \right) \right) \nonumber \\ &  \qquad  +N+x-K-1 \bigg) \nonumber \, .
\end{align}

For example, using the value for $x \approx 0.29N$ from the original Mitchison--Durbin pattern $f_\text{M}$ and dividing the total Pauli weight by the number of hopping terms, $2N(N-1)$, the average Pauli weight of our auxiliary qubit mapping, here nicknamed `$f_{\dot{h}'+2}$', produces a dramatic improvement:
\begin{align}
	\text{APV}\left(f_{\dot{h}'+2}\big|_{x\approx 0.29N}\right) &\approx 0.33N+1.77
\end{align}
However, in the scenario $x \leq \lfloor(N+4)/3\rfloor$ with $K = \lfloor(N-x-1)/2 \rfloor$, the extremum $x=\lfloor (N+4)/3 \rfloor$ actually minimises the expression in Equation \ref{eqn:tpv}. The average Pauli weight for large $N$ in this instance is:
\begin{align}
	\text{APV}\left(f_{\dot{h}'+2}\Big|_{x=\left\lfloor\frac{N+4}{3}\right\rfloor}\right) &\approx \frac{26}{81}N+\frac{289}{162} = 0.32N+1.78\, .
\end{align}

But there is a better value yet for $x$: in the other scenario, where $x > \lfloor (N+4)/3 \rfloor$, the value of $K$ is $N-2x+1$. An approximately optimal value for $x$ to minimise the total Pauli weight in this instance is
\begin{align}
	x &= \frac{1}{8}\left(7+N+\sqrt{\frac{15N^2-18N-53}{3}}\right) \label{eqn:actualbestx2}\\
	&\approx \frac{1}{8}\left(1+\sqrt{5}\right)N\\
	&= 0.40 N\, , \label{eqn:approxx2}
\end{align}
which is the result of replacing floor brackets with normal parentheses in Equation \ref{eqn:tpv}. It is for this reason that we define our new pattern $f_{\mathrm{M}+2}$ to have $x$ equal to the rounded value of that in Equation \ref{eqn:actualbestx2}, i.e.\ $x \approx 0.40N$. For large $N$, this yields an average Pauli weight of
\begin{align}
	\text{APV}\left(f_{\mathrm{M}+2}\right) &= \text{APV}\left(f_{\dot{h}'+2} \big|_{x\approx 0.40N}\right) \approx 0.31N+1.68\, .
\end{align}

This analysis does not include the cost of initialising the stabiliser state $\ket{\smash{\widetilde{\psi}}}_{\text{dat,aux}}$ via the unitary operation $V$. It is possible to implement $V$ using fewer gates than the circuit in Figure \ref{fig:Vcirc} at the cost of more ancilla qubits \cite{jiang2020optimal}. However, even adding the cost of the circuit in Figure \ref{fig:Vcirc}, which is $4(L-1)+2 = 4\big((Nx-1)-(x(x-1)+1)-1\big)+2$ $= 4Nx -x^2 + \mathcal{O}(x)$ CNOTs, to Equation \ref{eqn:tpv} has negligible effects on the calculations in this section. The value $x \approx 0.40N$ in Equation \ref{eqn:approxx2} still holds.

\bibliographystyle{quantum}
\bibliography{refs}

\end{document}